\documentclass[11pt]{article} %
\usepackage{graphicx, amsmath,amsthm,amssymb,url}
\usepackage[multiple]{footmisc}

\usepackage{mathtools}
\usepackage[multidot]{grffile}
\usepackage[normalem]{ulem}
\usepackage{algorithm}

\usepackage{tikz}
\definecolor{goldenpoppy}{rgb}{0.99, 0.76, 0.0}
\definecolor{glaucous}{rgb}{0.38, 0.51, 0.71}
\definecolor{blue-violet}{rgb}{0.54, 0.17, 0.89}

\usepackage{float}
\usepackage{fullpage,etoolbox}
\usepackage{color}
\usepackage[square]{natbib}
\usepackage{hyperref}       %
\hypersetup{
    colorlinks=true,
    linkcolor=blue,
    citecolor = blue,
    urlcolor = blue
}

\usepackage[capitalize,noabbrev]{cleveref}
\crefname{algocf}{Algorithm}{Algorithm}

\allowdisplaybreaks

\newtheorem{theorem}{Theorem}
\newtheorem{lemma}[theorem]{Lemma}
\newtheorem{proposition}[theorem]{Proposition}
\newtheorem{corollary}[theorem]{Corollary}

\newtheorem{definition}{Definition}

\usepackage{subcaption}

\usepackage[algo2e,ruled]{algorithm2e}
\usepackage{makecell,booktabs,ctable}
\begin{document}

\newcommand{\diag}{\textup{diag}}
\newcommand{\td}{\textup{TD}}
\newcommand{\E}{\mathbb{E}}
\newcommand{\R}{\mathbb{R}}

\newcommand{\algoname}{\textsc{SAC}}
\newcommand{\algonamefull}{Scalable Actor Critic}

\newcommand{\khop}{{ \kappa}}
\newcommand{\rhok}{{ \rho^{\khop+1}}}
\newcommand{\fk}{{ f(\khop)}}
\newcommand{\nik}{{ N_i^{\khop}}}
\newcommand{\njk}{{ N_j^{\khop}}}
\newcommand{\nminusik}{{ N_{-i}^{\khop}}}
\newcommand{\nminusjk}{{ N_{-j}^{\khop}}}

\newcommand{\final}[1]{{\color{red}#1}}

\title{Thinking Beyond Visibility: A Near-Optimal Policy Framework for Locally Interdependent Multi-Agent MDPs}

\author{Alex DeWeese\thanks{Department of Electrical and Computer Engineering, Carnegie Mellon University. Correspondence to: \url{mdeweese@andrew.cmu.edu}}\qquad Guannan Qu\footnotemark[1]}
\date{}
\maketitle

\begin{abstract}
Decentralized Partially Observable Markov Decision Processes (Dec-POMDPs) are known to be NEXP-Complete and intractable to solve. However, for problems such as cooperative navigation, obstacle avoidance, and formation control, basic assumptions can be made about local visibility and local dependencies. The work \cite{deweese2024locally} formalized these assumptions in the construction of the Locally Interdependent Multi-Agent MDP. In this setting, it establishes three closed-form policies that are tractable to compute in various situations and are exponentially close to optimal with respect to visibility. However, it is also shown that these solutions can have poor performance when the visibility is small and fixed, often getting stuck during simulations due to the so called ``Penalty Jittering'' phenomenon. In this work, we establish the Extended Cutoff Policy Class which is, to the best of our knowledge, the first non-trivial class of near optimal closed-form partially observable policies that are exponentially close to optimal with respect to the visibility for any Locally Interdependent Multi-Agent MDP. These policies are able to remember agents beyond their visibilities which allows them to perform significantly better in many small and fixed visibility settings, resolve Penalty Jittering occurrences, and under certain circumstances guarantee fully observable joint optimal behavior despite the partial observability. We also propose a generalized form of the Locally Interdependent Multi-Agent MDP that allows for transition dependence and extended reward dependence, then replicate our theoretical results in this setting.

\end{abstract}

\section{Introduction}
\label{intro}
Tasks such as cooperative navigation, obstacle avoidance, and formation control are important in applications like autonomous driving, UAVs, and robot navigation. These problems are often characterized by having local dynamic interactions (such as collisions) and partially observable regions that extend from each agent. The partial observability makes optimal navigation problems difficult which is further exacerbated by the so-called curse of dimensionality, the phenomenon where the number of possible states and actions increases exponentially with the number of agents. 

Theoretically, much of the prior literature models cooperative agents in this situation using a Decentralized Partially Observable Markov Decision Process (Dec-POMDP). Unfortunately, Dec-POMDPs are known to be NEXP-Complete (\cite{oliehoek2012decentralized}; \cite{goldman2004decentralized}) and are generally considered intractable to solve. 
In response to this, \cite{deweese2024locally} proposed the Locally Interdependent Multi-Agent MDP that exploits structural properties of these specific problems %
in a way that is amenable to theoretical analysis. In particular, for the proposed setting, it is shown that there exist three closed-form solutions to this partially observable problem called the Amalgam policy, Cutoff policy, and First Step Finite Horizon Optimal policy. These policies are provably exponentially close to optimal with the increase in visibility and are not just more tractable than Dec-POMDP solutions but in many cases are more tractable to compute than solutions to general multi-agent MDPs without visibility restrictions. Partial observability helps with tractability in this case, since agents only need to compute coordination behavior for local agents (see \cref{scalable}). Furthermore, this exponentially decreasing bound is matched with a lower bound up to constant factors, showing that these bounds are near optimal and cannot be improved significantly for this general setting. 

However, for many applications of interest, these few policies with basic asymptotic guarantees may not always be sufficient. In this work, we delineate two types of partial observability assumptions that can arise from real-world problems: (1) tunable observability which is not limited by our application of interest and (2) a small and fixed observability that is enforced by the application (see \cite{han2020cooperative}, \cite{long2018towards}). In scenario (1) we may choose a visibility by using asymptotic results, such as those shown in \cite{deweese2024locally} that will reduce computation time (see \cref{scalable}) and still guarantee the performance of the solution. However, in scenario (2), the asymptotic theoretical guarantees are loose for small and fixed visibilities and cannot be improved beyond what is established in (1) because of the theoretical lower bound. In this case, we are interested in not only the asymptotic results but also the performance of the policy beyond what we can guarantee theoretically for general Locally Interdependent Multi-Agent MDPs. In other words, we are interested in the ``implicit bias'' of the near optimal solutions when the visibility is small and fixed. Unfortunately, the three closed-form policies can often have poor performance in this case due to a fundamental lack of memory of other agents that cause ``forgetfulness'' in the policies resulting in the Penalty Jittering phenomenon (see simulations in \cite{deweese2024locally} and \cref{simulations}). It is unclear whether these exponentially close to optimal policies are viable in this scenario and would require selection across a larger set of these policies that integrate a mechanism for agents to remember other agents beyond their own visibilities. This motivates the following question:
\vspace{1ex}\\
\textit{Can we find a class of partially observable policies for Locally Interdependent Multi-Agent MDPs that are theoretically near optimal and improve performance in small visibility settings?}
\vspace{1ex}\\
\noindent We answer this affirmatively with the following contributions.
\subsection{Contributions}
\label{contributions}
In this work, we propose and provide theoretical guarantees for the Extended Cutoff Policy Class which include partially observable policies that have the ability to remember agents beyond their visibility. We show that these policies can resolve fundamental issues for the small visibility setting and will provide a wider selection of policies to choose from for any given application. This class will also include the original stationary policies (Amalgam Policy, Cutoff Policy, and First Step Finite Horizon Optimal Policy) presented in \cite{deweese2024locally} and reveal the underlying theoretical connection between them. 
We will answer the following questions about the Extended Cutoff Policy Class corresponding to the two scenarios presented in the introduction.
\vspace{1ex}\\
\noindent\textit{(1) Are policies in this class near optimal theoretically?}
In this work, we show that all the partially observable policies contained in the Extended Cutoff Policy Class are exponentially close in performance to the fully observable joint optimal policy with respect to the visibility. Performance will be constant factors away from the theoretical lower bound in \cite{deweese2024locally}, establishing near optimality. Further, we propose the Generalized Locally Interdependent Multi-Agent MDP framework that loosens the assumptions in the setting to allow for transition dependence and extended reward dependence between agents. We also show in this generalized setting that the partially observable policies in the Extended Cutoff Policy Class match the lower bound up to constant factors and are exponentially close in performance to the fully observable joint optimal policy with respect to the visibility.

\noindent\textit{(2) Can we improve performance using this policy class when visibility is small and fixed?}
In the small and fixed visibility setting, the Extended Cutoff Policy Class we propose contains a rich class of partially observable policies characterized by a ``thinking radius'' that can be arbitrarily larger than the visibility radius. Despite the small visibility, policies in our class can reason about the other agents beyond the visibility by predicting the position of other agents, either through random placement or through estimates based on past observations. Despite having the same worst-case guarantee,
the incorporation of this thinking process creates solutions with favorable implicit biases and performance in many small visibility environments. 

This leads to at least the following concrete benefits: (i) The Extended Cutoff Policy Class provides a large body of non-trivial partially observable solutions that are available for any Locally Interdependent Multi-Agent MDP instance with nuanced behavior in the small and fixed visibility setting. Previously, only three options were available from \cite{deweese2024locally} that had known performance issues in the small visibility regime.
(ii) The class includes policies that can resolve the Penalty Jittering issue seen in \cite{deweese2024locally} which causes agents to get stuck frequently in small visibility settings due to a lack of memory beyond its visibility (see \cref{penalty_oscillation}) (iii) Under very basic conditions, a specific instance of the Extended Cutoff Policy Class will achieve fully observable joint optimal behavior. That is, despite the policy being partially observable with a small visibility, it will be theoretically guaranteed to attain the best possible expected discounted rewards that can be achieved even without visibility constraints.

\subsection{Related Work}
This work will primarily build on the Locally Interdependent Multi-Agent MDP that is established in \cite{deweese2024locally}. It is a special case of the general Dec-POMDP (\cite{oliehoek2012decentralized}; \cite{oliehoek2016concise}) with additional assumptions that model partially observable optimal navigation problems appearing in applications such as robot navigation (\cite{han2020cooperative};\cite{long2018towards}), autonomous driving (\cite{palanisamy2020multi}; \cite{aradi2020survey}), and UAVs  (\cite{baldazo2019decentralized}; \cite{batra2022decentralized}). Unfortunately, Dec-POMDPs are known to be NEXP-Complete and much of the successful literature in this area is empirical. There have been other theoretical approaches that make assumptions on the Dec-POMDP structure such as the Interaction Driven Markov Game (IDMG) (\cite{spaan2008interaction}; \cite{melo2009learning}) or the Network Distributed POMDP (\cite{nair2005networked}; \cite{zhang2011coordinated}) but to the best of our knowledge these methods inherit hardness from the general Dec-POMDP and are intractable to solve in general (\cite{goldman2004decentralized}; \cite{bernstein2002complexity}; \cite{allen2009complexity}). This work can be viewed as a part of the scalable RL literature (\cite{qu2022scalable},\cite{qu2020scalableaverage}, \cite{qu2020scalablepolicy}) that make various assumptions on the general multi-agent MDP structure to improve scalability and establish theoretical performance guarantees. However, this literature did not attempt to model the dynamic dependencies that appear in problems such as cooperative navigation, obstacle avoidance, and formation control until the work by \cite{deweese2024locally}. What we present will also have resemblances to the ``centralized training decentralized execution'' concept introduced in \cite{lowe2017multi} but our work more intricately explores this idea in a theoretical context (see \cref{extended_cutoff}).

\section{Background}
\subsection{Standard Locally Interdependent Multi-Agent MDP}
\label{limdp}

\noindent The Locally Interdependent Multi-Agent MDP setting, proposed in \cite{deweese2024locally}, is specifically designed to model systems that perform local tasks such as cooperative navigation, obstacle avoidance and formation control. 

\textbf{State and Action:} We begin with a set of $\mathcal N = \{1,..., n\}$ agents, each of which lies in a common metric space $\mathcal X$ with an associated distance metric $d$ that can be interpreted as the ``environment'' for the agents. Each agent $i \in \mathcal N$ will have a corresponding position $x_i \in \mathcal X$ in this environment, as well as an internal state $y_i\in \mathcal Y_i$. The state of agent $i$ is then defined as $(x_i, y_i) \in \mathcal S_i$ where $ \mathcal S_i = \mathcal X \times \mathcal Y_i$. The distance notation $d$ will be overloaded to define a distance over states $d((x_i, y_i), (x_j, y_j)) = d(x_i, x_j)$ for $(x_i, y_i) \in \mathcal S_i$ and $(x_j, y_j) \in \mathcal S_j$.

This distance metric can be used to define a partition over agents $\mathcal N$ as a function of $s \in \mathcal S$ where $ \mathcal S = \mathcal S_1\times\mathcal S_2 \times ... \times \mathcal S_n$. Specifically for some $\mathcal K>0$ we say that two agents are ``related'' or ``in the same partition'' if there exists a sequence of agents $k_0, k_1, ..., k_\ell \in \mathcal N$ such that $d(s_{k_i}, s_{k_{i + 1}}) \leq \mathcal K$  for $i \in \{1,2,..., \ell - 1\}$. Formally, this defines an equivalence relation and thus defines a partition. As part of the setting, we will define a dependence radius $\mathcal R$ and in the case $\mathcal K = \mathcal R$, we will denote this as the dependence partition $D(s)$. Later we will define the visibility constant $\mathcal V > \mathcal R$ and setting $\mathcal K = \mathcal V$ will define the communication partition $Z(s)$. %

\textbf{Reward:} An agent $i$ lying in the partition $g \in D(s)$ will attain a reward $r_i(s_g, a_g) = \sum_{j \in g} \overline r_{i,j}(s_i, a_i, s_j, a_j)$ for some predefined pairwise reward function $\overline{r}_{i,j}$ where $s_g, a_g$ are the states and actions of agents within the group $g$. In the standard version, we assume $ \overline r_{i,j}(s_i, a_i, s_j, a_j) = 0$ if $d(s_i, s_j) > \mathcal R$ and therefore even if they are in the same $D(s)$ partition, $\overline r_{i,j}$ may be enforced to be $0$ if the agents are farther than $\mathcal R$ apart. This assumption will hold when modeling local phenomena like collisions and is helpful analytically because if $\overline r_{i,j}(s_i, a_i, s_j, a_j)$ is non-zero, it will contribute that value to the reward regardless of $D(s)$ and will be zero otherwise. That is, $\sum_{g \in D(s)} r_g(s_g, a_g) = \sum_{i \in \mathcal N, j \in \mathcal N}\overline r_{i,j}(s_i, a_i, s_j, a_j)$ where $r_g(s_g, a_g) = \sum_{i \in g} r_i(s_g, a_g)$. We loosen this assumption in the Generalized Locally Interdependent Multi-Agent MDP shown in \cref{generalized}. %

\textbf{Transition:} Agents in this setting will move with independent transitions $P_i(s'_i \lvert s_i, a_i)$. Crucially, we assume that if $d(s'_i, s_i) > 1$ then $P_i(s'_i \lvert s_i, a_i) = 0$. In other words, agents are not allowed to move more than a unit distance at each step, serving as a ``speed limit'' for the environment. We loosen the transition independence requirement in the generalized version in \cref{generalized} where local transition dependence is allowed. %

\begin{definition}
(Locally Interdependent Multi-Agent MDP) Assume a set of agents $\mathcal N$ and a set of states in a common metric space $(\mathcal X, d)$. Furthermore, for $i \in \mathcal N$ assume a set of internal states $\mathcal Y_i$, a set of actions $\mathcal A_i$, a transition function $P_i$ such that $P_i(s'_i \lvert s_i,a_i) = 0$ if $d(s_i, s'_i) > 1$, and a reward function $r_{i,j}$ such that $ \overline r_{i,j}(s_i, a_i, s_j, a_j) = 0$ if $d(s_i, s_j) > \mathcal R$ for $j \in \mathcal N$. The Locally Interdependent Multi-Agent MDP $\mathcal{M} = (\mathcal{S}, \mathcal A, P, r, \gamma, \mathcal R)$ is defined as follows:
\begin{itemize}
\item $\mathcal S := \times_{i \in \mathcal N}\mathcal S_i$ where $\mathcal S_i = (\mathcal X, \mathcal Y_i)$
\item $\mathcal A := \times_{i \in \mathcal N}\mathcal A_i$ 
\item $ P(s' \lvert s, a) = \prod_{i\in \mathcal N} P_i(s'_i \lvert s_i,a_i)$ 
\item $ r(s, a) = \sum_{g \in D(s)} r_g(s_g,a_g)$
\end{itemize}
where $r_g(s_g, a_g) = \sum_{i, j\in g} \overline r_{i,j}(s_i, a_i, s_j, a_j)$
\end{definition}

\textbf{Partial Observability and Group Decentralized Policies:} Partial observability is introduced by allowing agents within distance $\mathcal V > \mathcal R$ to communicate and coordinate. Communication is transitive and if there exist two agents separated by a chain of agents within distance $\mathcal V$, we allow those agents to communicate (see \cref{circles} for an illustration). Concisely, agents can communicate and act together if they are within the same communication partition $Z(s)$ defined previously.

Formally, partially observability is captured by the class of group decentralized policies which take the form $\pi(a\lvert s) = \prod_{z \in Z(s)} \pi_z(a_z \lvert s_z)$ or in the deterministic case, $\pi(s) = (\pi_z(s_z): z \in Z(s))$ which represents a concatenation of deterministic sub-policies. Finite horizon group decentralized policies are simply sequences of these stationary group decentralized policies. In this work, we will also introduce non-Markovian group decentralized policies that allow agents to communicate prior observations within their communication group. Specifically, a non-Markovian group decentralized policy will decompose as $\pi'(a\lvert s) = \prod_{z \in Z(s(t))} \pi_{I(\tau^t, z)}'(a_z \lvert s_z(t))$. Here $\tau^t$ represents the trajectory up to time $t$ and $I(\tau^t, z)$ is a set that includes all group states and time step tuples $(s_{z'}(t'), t')$ for some $z' \in Z(s(t'))$ and $t' \in \{0, \ldots, t\}$ such that there is a sequence of group states $s_{z_{t'}}(t'), s_{z_{t' + 1}}(t' + 1), \ldots  s_{z_{t}}(t)$ with $z_{t'} = z'$, $z_{t} = z$ such that $z_{t''} \cap z_{t'' + 1} \neq \varnothing,\forall t'' \in \{t',t' + 1, \ldots, t-1 \}$ and $z_{t''} \in Z(s(t'')),\forall t'' \in \{t' + 1, \ldots, t -1 \}$. In other words, agents may communicate within their visibility group about states for other agents that have been seen in the past. $I(\tau^t, z)$ then contains all the information we allow the group $z$ to use when deciding on an action.

\textbf{Applications:} Notice that this framework can be used to model cooperative navigation, obstacle avoidance, and formation control problems. We can express collisions in cooperative navigation by designing $\overline r_{i,j}(s_i, a_i, s_j, a_j)$ to be negative penalties for some $s_i, s_j$ with $d(s_i,s_j) \leq \mathcal R$ to disincentivize agents from getting close to one another.  Obstacles can be modeled by agents which are unable to take any actions ($\mathcal A_i$ is a singleton) and penalize other agents that are too close. We may also encourage formation control by making $\overline r_{i,j}(s_i, a_i, s_j, a_j)$ to be a positive reward that incentivize agents to move to specific relative locations within $\mathcal R$ of each other. If more complex local behavior is required (such as collisions stopping agent movement or formations that only reward agents when they are all in position), local transition dependence and extended reward dependence are allowed with the generalized form introduced in \cref{generalized}. The group decentralized policies can model the partial observability of the agents that must act only based on their immediate communication group and information that has been seen in the past.

\textbf{Objective:} 
Our goal is to find group decentralized polices $\pi$ that are defined for all Locally Interdependent Multi-Agent MDP instances for any visibility and have high expected discounted rewards 
$V^{\pi} (s) = \mathbb{E}_{\tau \sim \pi\lvert_{s}} \bigg[\sum_{t = 0}^{\infty} \gamma^{t} r(s(t), a(t))\bigg]$. Here, $\tau \sim \pi\lvert_s$ indicates a trajectory in the Locally Interdependent Multi-Agent MDP starting with state $s$ and transitioning with dynamics $P$ with policy $\pi$.

\textbf{Additional Notation:}
A partition $\mathcal P$ is said to be finer than a partition $\mathcal P'$ if for every $p \in \mathcal P$ there exists some $p' \in \mathcal P'$ such that $p \subset p'$. In this case, we also say that $\mathcal P'$ is coarser than $\mathcal P$. Since $D(s)$ is finer than $Z(s)$, we may create the equivalent expression for the reward $r(s,a) = \sum_{g \in D(s)} r_g(s_g,a_g) = \sum_{z \in Z(s)} r_z(s_z,a_z)$ where $ r_z(s_z, a_z) = \sum_{g \in D(s_z)} r_g(s_g,a_g)$. We will also denote an upper bound on the reward at each step as $\tilde r = \sum_{i,j\in \mathcal N} \max_{s_i, a_i, s_j, a_j}\lvert \overline r_{i,j}(s_i, a_i, s_j, a_j)\rvert$.

\begin{figure*}[t!]
    \centering
    \begin{subfigure}[b]{0.45\textwidth}
        \raggedleft
        \includegraphics[width=0.8\textwidth]{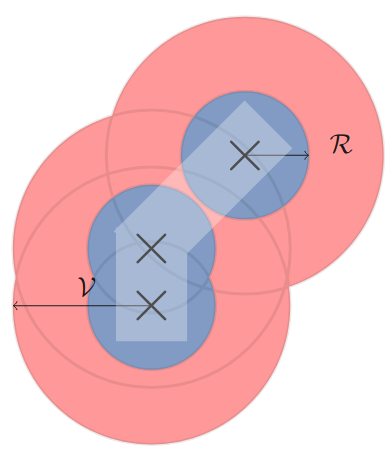}
        \caption{Illustration of 3 agents in $\mathbb R^2$ in the Locally Interdependent Multi-Agent MDP taken from \cite{deweese2024locally}. The top and bottom agent communicate transitively.\label{circles}}
        \label{fig:sub1}
    \end{subfigure}
    \hfill
    \begin{subfigure}[b]{0.45\textwidth}
        \centering
        \includegraphics[width=0.8\textwidth]{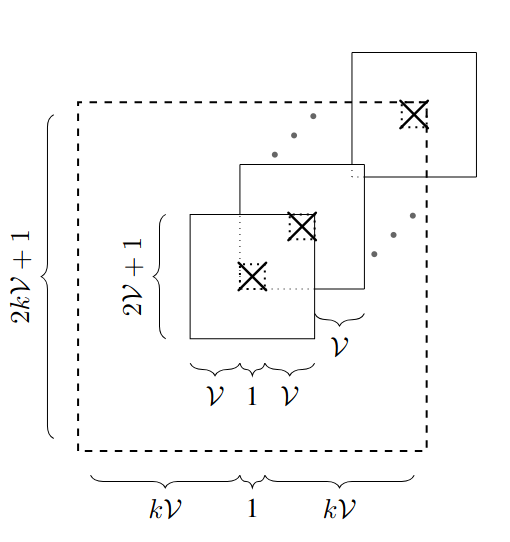}
        \caption{Square of length $\ell = 2k\mathcal V + 1$ of possible locations for $k$ other agents in the example presented in \cref{scalable}.\label{squares}\\}
        \label{fig:sub2}
    \end{subfigure}
\end{figure*}

\subsection{Scalability}
\label{scalable}
In many discrete environments, the number of possible group states $\big\lvert\{s_z: z \in Z(s), s \in \mathcal S\}\big\rvert$ is significantly smaller than the number of possible states $\lvert\mathcal S\rvert$. Since group decentralized policies are defined over group states, this difference in size can improve the computation time in algorithms and reduce storage requirements.

In order to concretely demonstrate the improvement, we will investigate an illustrative example. Consider an environment on a large square grid with $M = L^2$ spaces. We will use the Chebyshev distance metric  $d((c_1, c_2), (c_1', c_2')) = \max (\lvert c_1 - c_1' \rvert, \lvert c_2 - c_2'\rvert)$ and assume that the visibility of the $n$ agents is $\mathcal V$. Notice that for any group state with $k+1$ agents in this environment, if a particular agent is placed at a position, the other $k$ agents cannot be more than $k \mathcal V$ away from that position (see \cref{squares}). Therefore, the placement of the other agents will be restricted to a square of length $\ell = 2(n - 1) \mathcal V + 1$. Letting $m = \ell^2$ and assuming $\mathcal V \geq \frac{1}{2}$, notice that this leads to the upper bound on the number of group states $nM \sum_{k = 0}^{n - 1} {n - 1\choose k} (2k\mathcal V + 1)^{2k} \leq nM \sum_{k = 0}^{n -1 } {n - 1\choose k}  m^k = nM(m + 1)^{n - 1} \leq M (m + 1)^n$ using the binomial theorem. Compared to the total number of states $M^n$ the exponential dependence now is on $m$ instead, and this can be a considerable improvement when $L\gg \ell$. Simply due to the exponent, if $\ell \leq L/r$ for some $r > 1$, then $(m+1)^n \approx m^n$ is at least a factor of $(1/r)^{2n}$ smaller than $M^n$.

Although this can greatly reduce the computation time, we have technically not eliminated the curse of dimensionality because of the persistence of the exponent $n$. This can be resolved by using heuristics or approximations when group sizes $n$ exceed some threshold. See the Random Navigation With Many Agents simulation in \cref{b_many} where large group heuristics and locality are used to tractably learn a navigation strategy for 100 agents that achieve random objectives and do not collide, without the use of deep learning. For this representative example, the heuristics are relied upon minimally because the group sizes remain small most of the time.

\subsection{Cutoff Multi-Agent MDP}
\label{cutoff}
Before constructing the Extended Cutoff Policy Class, we will introduce a companion to the Locally Interdependent Multi-Agent MDP called the Cutoff Multi-Agent MDP. Intuitively, this is an analogous setting where agents that disconnect are not allowed to reconnect. This will allow us to talk about independent trajectories for groups without interactions with other groups.
 
 Defining the intersection of partitions for $\mathcal P_1$ and $\mathcal P_2$ over $\mathcal N$ to be $\mathcal P_1 \cap \mathcal P_2  = \{p_1 \cap p_2 \lvert p_1 \in \mathcal P_1, p_2 \in \mathcal P_2\} \setminus \{\varnothing\}$, the following is the definition of the Cutoff Multi-Agent MDP.
\begin{definition}
(Cutoff Multi-Agent MDP) Given a Locally Interdependent Multi-Agent MDP $\mathcal{M} = (\mathcal{S}, \mathcal A, P, r, \gamma, \mathcal R)$, we may define the Cutoff Multi-Agent MDP $\mathcal{C} = (\mathcal{S}^\mathcal C, \mathcal A^\mathcal C, P^\mathcal C, r^\mathcal C, \gamma, \mathcal R, \mathcal V)$ as follows:
\begin{itemize}
\item $\mathcal S^\mathcal C := \{(s, \mathcal P) : s \in \mathcal S, \mathcal P \in F(Z(s))\}$ 
\item $\mathcal A^\mathcal C := \mathcal A$
\item $ P^\mathcal C((s', Z(s') \cap \mathcal P) \lvert (s, \mathcal P), a) = P(s' \lvert s,a)$ and $0$ otherwise
\item $ r^\mathcal C((s, \mathcal P), a) = \sum_{p \in \mathcal P}r_p(s_p,a_p)$
\end{itemize}
where $F(Z(s))$ represents all possible partitions finer than $Z(s)$ and $r_p(s_p, a_p) = \sum_{g \in D(s_p)} r_{g}(s_{g}, a_{g})$.
\end{definition}

Notice that the definition is similar to the Locally Interdependent Multi-Agent MDP, except the state now includes a partition over the agents. Agents in different partitions are not allowed to interact with each other in either the transition or reward, and this partition only gets finer with time as it is intersected with $Z(s')$ as part of its transition. In other words, when agents are disconnected, they will never reconnect again.

\textbf{Proper Cutoff Policies:} The analogous policy class to the group decentralized policy class in this setting are the proper cutoff policies defined as $\pi(a\lvert (s, \mathcal P)) = \prod_{p \in \mathcal P} \pi_p(a_p \lvert s_p)$ or in the deterministic case, $\pi(s, \mathcal P) = (\pi_p(s_p): p \in \mathcal P)$. In other words, the visibility is also ``cut off'' when agents disconnect. We also define finite horizon proper cutoff policies to be sequences of these stationary proper cutoff policies.

\textbf{Objective:} For some horizon $H$ and non-stationary policy $\pi = (\pi_0, \pi_1, \ldots, \pi_H)$, the notation $\tau_h \sim \pi\lvert_{s, \mathcal P}$ will denote the trajectory $((s(t), \mathcal P(t)), a(t))$ for $t \in \{h, h + 1, ..., H\}$ starting at $s(h) = s$ and $\mathcal P(h) = \mathcal P$ while taking policy $\pi_t$ at time step $t$. We will define the (discounted) finite horizon values as
$V^{\pi}_h (s, \mathcal P) = \mathbb{E}_{\tau_h \sim \pi\lvert_{s,\mathcal P}} \bigg[\sum_{t = h}^{H} \gamma^{t-h} r^\mathcal C((s(t), \mathcal P(t)), a(t))\bigg]$. We will be interested in maximizing $V^\pi_0(s, \mathcal P)$ for various horizons $H$ and various visibilities $\mathcal V$. 
This value function is disambiguated from the value function for the Locally Interdependent Multi-Agent MDP by its two parameters as opposed to one. It can also be shown that the optimal policy  for the Cutoff Multi-Agent MDP is a proper cutoff policy (see \cite{deweese2024locally}). %

\textbf{Computation and Storage:} In the case $\pi$ is a proper cutoff policy, once agents are disconnected, they never interact again. This allows for the following decomposition on the value function $V^\pi_h(s, \mathcal P) = \sum_{p\in \mathcal P} V^{\pi_p}_h(s_p, \{p\})$ where $V^{\pi_p}_h(s_p, \{p\})$ is the value function for the Cutoff Multi-Agent MDP using $p$ in place of $\mathcal N$ as the set of agents with their corresponding state spaces, action spaces, transition functions, and reward functions. 
This permits the storage space improvement described in \cref{scalable} for the value function, as the values only need to be stored for each $(s_p, \{p\})$. When agents are homogeneous, the requirement decreases even further because all value functions of the form $V^{\pi_p}_h(s_p, \{p\})$ will be the same when $\lvert p\rvert$ is the same. 

As mentioned previously, the optimal policy for the Cutoff Multi-Agent MDP is a proper cutoff policy (See \cite{deweese2024locally}), and this improvement will also extend to computation time in the Bellman Optimality Equations for the Cutoff Multi-Agent MDP:
\begin{align*}
&V^*_h(s_p, \{p\}) = \max_{a_p} Q^*_h((s_p, \{p\}), a_p)\\
Q_h^*((s_p, \{p\}), &a_p) = r^{\mathcal C}((s_p, \{p\}), a_p) + \gamma \mathbb{E}_{s_p' \sim P(\cdot\lvert s_p, a_p)}\left[\sum_{z \in Z(s_p')} V_{h + 1}^*(s_z, \{z\})\right]
\end{align*}
Since the values only need to be computed for each $(s_p, \{p\})$ and the equations query the value at other states with similar form, we attain a computational improvement like what is presented in \cref{scalable}. These improvements also extend to the Bellman Consistency Equations and common algorithms such as policy iteration and policy gradients. These properties are why we favor computation in the Cutoff Multi-Agent MDP as opposed to directly solving the Locally Interdependent Multi-Agent MDP.

\section{Main Results}
\subsection{Extended Cutoff Policy Class}
\label{extended_cutoff}
As discussed in the Introduction in \cref{intro}, our goal is to find a policy class that is defined for all Locally Interdependent Multi-Agent MDP instances, is exponentially close optimal with respect to the visibility, and performs well under the small and fixed visibility regime. The policy class we propose in this paper is called the Extended Cutoff Policy Class. 
Intuitively, these are the solutions to the Cutoff Multi-Agent MDP with some horizon and extended computational visibility that are converted into a group decentralized policies for the Locally Interdependent Multi-Agent MDP. This extended computational visibility can be thought of as the ``thinking radius'' and the extended visibility solution can allow a group of agents to act based on an estimation of the position of other agents. We may increase this ``thinking radius'' to make the computation phase more connected until it is fully centralized. In this fully centralized regime, it resembles the ``centralized training decentralized execution'' framework described in \cite{lowe2017multi} which was introduced in an empirical context. In our case, however, we are able to control the level of connectivity during the computation phase, trading off with the computation time improvements described in \cref{scalable}.

\textbf{Extended Cutoff Multi-Agent MDP:} To begin our construction for the Extended Cutoff Policy class, we will consider a Locally Interdependent Multi-Agent MDP with visibility $\mathcal V_{exec}$. During the computation phase, we will solve the corresponding Cutoff Multi-Agent MDP with extended visibility $\mathcal V_{comp} = \mathcal V_{exec} + \xi$ and horizon $c + \eta$ where $c = \lfloor \frac{\mathcal V_{exec} - \mathcal R}{2}\rfloor$, $\xi\geq 0$, and $\eta \in \{0, 1, ...\}$.  We let $Z_{exec}$ and $Z_{comp}$ to be the corresponding visibility partitions of $\mathcal V_{exec}$ and $\mathcal V_{comp}$. In addition, we define the possible group states in the execution phase $\mathcal S_{exec} = \{s_{z}:z \in Z_{exec}(s), s\in \mathcal S\}$ and the computation phase $\mathcal S_{comp} = \{(s_{p}, \{p\}): p \in \mathcal P, (s, \mathcal P)\in\mathcal S^{\mathcal C}\}$ . In terms of $\mathcal S_{comp}$, for any state during the computation phase $(s, \mathcal P) \in \mathcal S^{\mathcal C}$, all agents within the same partition $\mathcal P$ are within the same partition $Z_{comp}(s)$ according to the definition of $\mathcal S^\mathcal C$. We consider all partitions $\mathcal P$ that are finer than $Z_{comp}(s)$ and so $\mathcal S^\mathcal C$ contains the group states $s_{z'}$ for $z' \in Z_{comp}(s)$ which will be split further by every $\mathcal P$, and will include all possible ``subgroups'' that may emerge from the ``permanent disconnection'' in the Cutoff Multi-Agent MDP. Therefore, $\mathcal S_{comp}$ becomes the analog to $\mathcal S_{exec}$ which describes all possible communication groups and their corresponding group states that may appear within the extended visibility Cutoff Multi-Agent MDP. We will solve and find the discounted finite horizon solution to the  extended Cutoff Multi-Agent MDP problem $\pi^{\xi, \eta}_{comp} = (\pi_0^{\xi, \eta}, \pi_1^{\xi, \eta}, ..., \pi_{c + \eta}^{\xi, \eta})$. Notice that as $\xi$ increases, the setting becomes progressively more centralized and the calculation of this solution becomes more expensive (see \cref{scalable}). Although $\pi^{\xi, \eta}_{comp}$ is defined over states in $\mathcal S$ and partitions over $\mathcal N$, it will be understood in the section below that $\pi^{\xi, \eta}_h(\cdot \lvert (s_p, \{p\}))$ for any $(s_p, \{p\}) \in \mathcal S_{comp}$ will be $[\pi^{\xi, \eta}_h]_p (\cdot \lvert (s_p, \{p\}))$ which is the partial policy obtained from the definition of the proper cutoff policy (see \cref{cutoff}).
 
\textbf{Extraction Method:} Now that we have solved an extended visibility Cutoff Multi-Agent MDP, we would like to convert this solution into a policy that may be executed in our original Locally Interdependent Multi-Agent MDP setting with the original visibility $\mathcal V_{exec}$. During execution time at time step $t$, we will obtain some state $s(t)$ in the Locally Interdependent Multi-Agent MDP and need to decide on an action to take for $s_z(t)$ with $z \in Z_{exec}(s(t))$. We construct some extraction method $\rho(I(\tau^t, z))$ from the computation to the execution phase which is a random variable that depends on $s_z(t) \in \mathcal S_{exec}$ as well as the entire observation history of group $z$ during the execution phase, and takes values in the set of group states $\mathcal S_{comp}$ that arise during the computation phase (notation for $I(\tau^t, z)$ is introduced in \cref{limdp}). Since $\mathcal V_{comp}\geq \mathcal V_{exec}$, this extraction method can intuitively be viewed as a ``belief'' of the position of other agents if the visibility were $\mathcal V_{comp}$. Along these lines, denoting $(s_p, \{p\})$ as the outcome of $\rho(I(\tau^t, z))$,  we require for all $z \in Z_{exec}(s(t))$ that $z\in Z_{exec}(s_p)$ and $[s_p]_z = s_z(t)$ . %
In other words, the ``belief'' that a particular group has about the position of the other agents does not contradict the current observations of the visibility group and its immediate vicinity. 
 We will then take the policy $\pi^{\xi,\eta}_{exec}(a \lvert s(t)) = \prod_{z\in Z_{exec}(s(t))}[\pi_{0}^{\xi, \eta}]_z(a_z \lvert \rho(I(\tau^t,z)))$ where $[\pi_{0}^{\xi, \eta}]_z$ is the first step of the finite horizon optimal policy solution to our extended Cutoff Multi-Agent MDP evaluated at the current ``belief'' $I(\tau^t, z)$, restricted to the actions for agents in group $z$. Notice that this is a non-Markovian group decentralized policy introduced in \cref{limdp}.

\begin{definition} (Extended Cutoff Policy Class)
We define the Extended Cutoff Policy Class to be all group decentralized policies of the form $\pi^{\xi,\eta}_{exec}(a \lvert s(t)) = \prod_{z\in Z_{exec}(s(t))}[\pi_{0}^{\xi, \eta}]_z(a_z \lvert \rho(I(\tau^t,z)))$ with any valid extraction method $\rho(I(\tau^t,z))$ for all $\xi \geq 0$, $\eta \in \{0,1, ...\}$.  
\end{definition}
\textbf{Extraction Examples:} In this work, we show that every element of the Extended Cutoff Policy Class is exponentially close to optimal with respect to the visibility $\mathcal V_{exec}$ (see \cref{main_results}).%
\\
Here, we will primarily refer to three types of extraction methods:
\begin{itemize}
    \item Trivial Extraction: $\rho (I(\tau^t, z)) = (s_z(t), \{z\})$ for $z \in Z_{exec}(s(t))$
    \item Aggregate Extraction: Any $\rho (I(\tau^t, z))$ with a distribution that only depends on $s_z(t)$
    \item Simple Memory Based Extraction: The first component of the output by \cref{simple_algo} in \cref{simple_mem}.
\end{itemize}
Intuitively, Trivial Extraction directly queries the value function for the Cutoff Multi-Agent MDP with the current group state during execution. Essentially, the agents ``believe'' that no other agents are in the vicinity. The Amalgam Policy, the Cutoff Policy, and the First Step Finite Horizon Optimal Policy introduced in \cite{deweese2024locally} are types of Trivial Extraction with various $\xi$ and $\eta$. See \cref{corners} and \cref{journey}.
Aggregate Extraction describes extraction methods that place the belief of the agents in a location according to a distribution that does not depend on the previous history.  Notice that Trivial Extraction is a type of Aggregate Extraction. 

A walk-through of the complete Simple Memory Based Extraction algorithm is shown in \cref{simple_mem}. For each agent, the algorithm maintains estimates for the position of previously seen agents. These predictions are shared with the current visibility group and then used to construct a belief state which will be the output of the algorithm.

Out of the three extraction methods we present, the Simple Memory Based Extraction method performs well in deterministic cases when the visibility is fixed, and resolves many of the empirical problems seen in \cite{deweese2024locally} (see \cref{simulations} for simulations and the Results section in \cref{main_results} for more discussion).

\begin{figure}[H]
    \centering
    \includegraphics[width=0.35\linewidth]{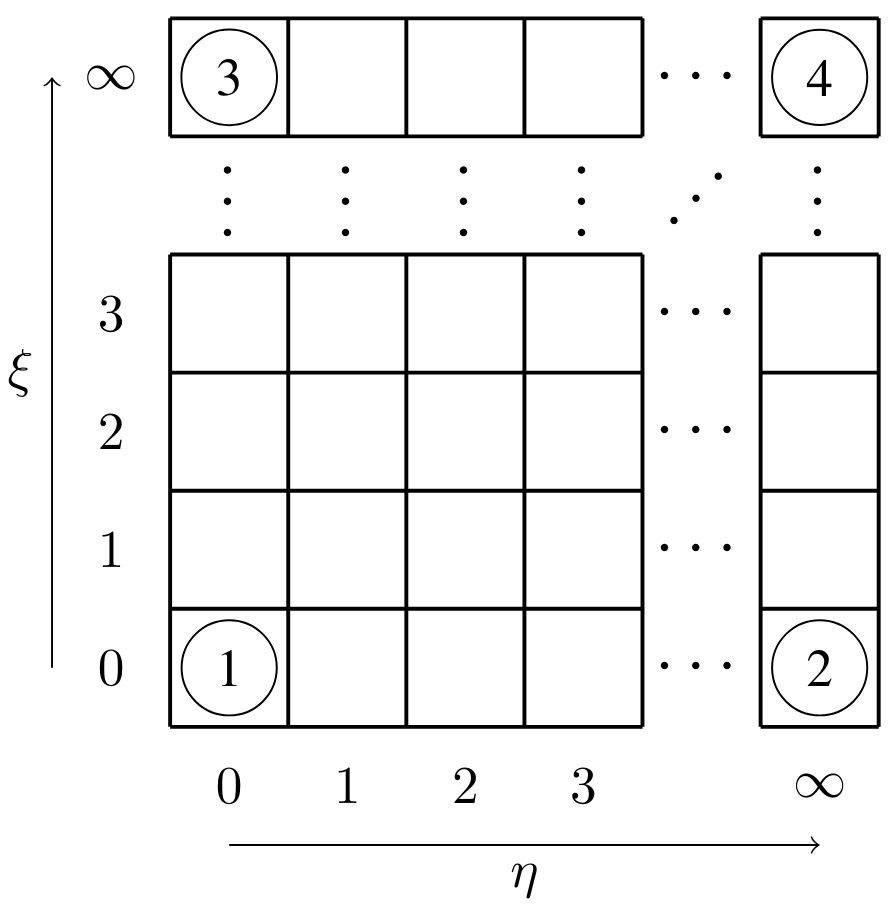}
    \caption{A grid of Trivial Extraction possibilities with $\xi \in \{0,1, \ldots\}$. (1) is the First Step Finite Horizon Optimal Policy ($\eta = 0$, $\xi = 0$),  (2) is the Cutoff Policy ($\eta \rightarrow \infty$, $\xi = 0$), and (4) is the Amalgam Policy ($\eta \rightarrow\infty$, $\xi \rightarrow \infty$) from \cite{deweese2024locally} which are special cases of our constructed Extended Cutoff Policy Class. We call (3) the Finite Amalgam Policy ($\eta = 0$, $\xi \rightarrow \infty$).
 }
    \label{corners}
\end{figure}

\subsection{Results}
Recall from \cref{contributions} that our contributions will be split into two parts. (1) Establishing asymptotic results for variable visibility $\mathcal V_{exec}$ and (2) investigating its performance in the fixed visibility $\mathcal V_{exec}$ setting.

\textbf{Generalized Locally Interdependent Multi-Agent MDP:} In addition to studying the above two questions, we also propose a generalized version of the Locally Interdependent Multi-Agent MDP, for which our main theoretical results to be introduced in this section also hold. Specifically, the generalized setting allows for local transition dependence and extended reward dependence. It assumes arbitrary local transitions $P_g(s'_g\lvert s_g, a_g)$ and arbitrary $\overline r_{i,j} (s_i,a_i, s_j, a_j)$ for $s_i,s_j \in g$ with $g \in D(s)$. This setting can be used to model more complex interactions between agents such as collisions that restrict mobility and more intricate reward functions such as an environment that rewards agent formations only when all agents are in position. The full formal definition is presented in \cref{generalized}. 
\label{main_results}
\subsubsection{Variable Visibility}
In this work, we establish the following guarantee on the asymptotic performance of policies in the Extended Cutoff Policy Class.
\begin{theorem}
\label{extract_opt_bound}
    Assume a (Generalized) Locally Interdependent Multi-Agent MDP $\mathcal{M} = (\mathcal{S}, \mathcal A, P, r, \gamma, \mathcal R)$ with visibility $\mathcal V_{exec}$, and an optimal discounted finite horizon solution $\pi^{\xi, \eta}_{comp} = (\pi_0^{\xi, \eta}, \pi_1^{\xi, \eta} , ..., \pi_{c + \eta}^{\xi, \eta})$ in its corresponding Extended Cutoff Multi-Agent MDP $\mathcal{C} = (\mathcal{S}^\mathcal C, \mathcal A^\mathcal C, P^\mathcal C, r^\mathcal C, \gamma, \mathcal R, \mathcal V_{comp})$ where $\mathcal V_{comp} = \mathcal V_{exec} + \xi$. For its corresponding extracted policy $\pi^{\xi,\eta}_{exec}(a \lvert s(t)) = \prod_{z\in Z_{exec}(s(t))}[\pi_{0}^{\xi, \eta}]_z(a_z \lvert \rho(I(\tau^t,z))$ according to a valid extraction strategy $\rho(I(\tau^t,z))$ for $z \in Z_{exec}(s)$, $s \in \mathcal S$, and $t \in \{0, 1, \dots\}$ we have
    $$V^*(s) - V^{\pi^{\xi, \eta}_{exec}}(s)\leq \beta\frac{\gamma^{c - 1}}{1 - \gamma}$$
where $c = \lfloor \frac{\mathcal V_{exec} - \mathcal R}{2}\rfloor$ and $c' = \lfloor \frac{\mathcal V_{comp} - \mathcal R}{2}\rfloor$. Here, $\beta = \gamma^{c' - c + 1} + \gamma^2 + \gamma^{\eta + 1} + \frac{4 + \gamma + 5\gamma^2}{1 - \gamma}$ for Standard Locally Interdependent Multi-Agent MDPs and $\beta = 2\gamma^{c' - c + 1} + 2\gamma^2 + \gamma^{\eta + 1} + \frac{2\gamma( 1+ \gamma) + 4n(1 + \gamma^2)}{1 - \gamma} $ for General Locally Interdependent Multi-Agent MDPs.
\end{theorem}
The standard version is proved in \cref{extract_opt_bound_proof} and the generalized proof is available in \cref{general_extract_opt_bound_proof}. Both guarantees match the lower bound presented in \cite{deweese2024locally} up to constant factors. The lower bound still holds for the generalized setting, since the standard setting is a special case of the generalized setting. We see that increasing $\mathcal V_{exec}$ brings the policy closer to the fully observable joint optimal policy exponentially fast, trading off with the computation time improvements described in \cref{scalable}. 
 
\Cref{extract_opt_bound} is a special case of a more general theorem that establishes a closeness between extended Cutoff Multi-Agent MDP values and the values of the extracted policy in the Locally Interdependent Multi-Agent MDP. Specifically, for any Consistent Performance Policy defined in \cref{monotone}, we have the following.
\begin{theorem}
\label{extract_bound} 
In the same setting as \cref{extract_opt_bound}, any Consistent Performance Policy $\pi_{comp} = (\pi_0, \pi_1, ..., \pi_{c + \eta})$ and its corresponding extracted policy $\pi_{exec}(a \lvert s(t)) = \prod_{z\in Z_{exec}(s(t))}[\pi_{0}]_z(a_z \lvert \rho(I(\tau^t,z))$ will have
$$\lvert V^{\pi_{comp}}_0(s, Z_{comp}(s)) - V^{\pi_{exec}} (s) \rvert \leq \beta'\frac{\gamma^{c - 1}}{1 - \gamma}.$$
where $c = \lfloor \frac{\mathcal V_{exec} - \mathcal R}{2}\rfloor$. Here, $\beta' = \gamma^2 + \gamma^{\eta + 1} + \frac{4 + \gamma + 5\gamma^2}{1 - \gamma}$ for Standard Locally Interdependent Multi-Agent MDPs and $\beta' = 2\gamma^2 + \gamma^{\eta + 1} + \frac{2\gamma( 1+ \gamma) + 4n(1 + \gamma^2)}{1 - \gamma}$ for General Locally Interdependent Multi-Agent MDPs. 
\end{theorem}
The standard version is proved in \cref{extract_bound_proof} and the generalized proof is available in \cref{general_extract_bound_proof}. Intuitively, Consistent Performance Policies satisfy conditions that allow for a ``good extraction.'' It is shown in \cref{monotone} that the optimal discounted finite horizon solution $\pi_{comp}^{\xi, \eta}$ is a Consistent Performance Policy. %

In short, \cref{extract_opt_bound} with the help of \cref{extract_bound} answers our initial question ``\textit{are policies in this class near optimal theoretically?}'' from the Contributions section (\cref{contributions}) by creating a performance guarantee for all elements in the Extended Cutoff Policy Class which is exponentially close to optimal with respect to the visibility $\mathcal V_{exec}$ and matches the lower bound in \cite{deweese2024locally} up to constant factors.

\subsubsection{Small and Fixed Visibility}
To investigate the nuances of the performance in the fixed visibility setting, we provide 11 grid world simulations in \cref{simulations} that compare the performance of Simple Memory Based Extraction, the Amalgam Policy, Cutoff Policy, and the fully observable joint optimal policy in various environments. We primarily focus on environments that are deterministic similar to navigation problems in the real world. Across these simulations, we see three common observations that correspond to our outline made in \cref{contributions}. (i) Simple Memory Based Extraction out performs the Amalgam Policy and Cutoff Policy for most simulations except for edge cases that ``trick'' the policy into taking suboptimal behavior. (ii) Simple Memory Based Extraction resolves the Penalty Jittering issue for sufficiently large $\mathcal V_{comp}$ and agents no longer get stuck in situations that the Amalgam Policy and Cutoff Policy struggle with (see \cref{penalty_oscillation} for the full discussion). (iii) looking at \cref{simple_algo} in \cref{simple_mem}, when the setting is completely deterministic and all agents start within view, for sufficiently large $\mathcal V_{comp}$, every agent will have the correct estimation of every other agent at all times. This leads to the following proposition:
\begin{proposition}
For any Locally Interdependent Multi-Agent MDP and an initial state $s$ where all agents are within view, the Simple Memory Based Extraction policy $\pi^{\xi, \eta}_{simple}$ will converge to the fully observable joint optimal solution despite being a partially observable group decentralized policy. That is, $V^{\pi^{\xi, \eta}_{simple}}(s) \rightarrow V^*(s)$ as $\xi \rightarrow \infty$, $\eta \rightarrow \infty$.
\end{proposition}
Examples in which this assumption holds are shown in \cref{aisle} and \cref{journey}.

An added benefit to using the Extended Cutoff Policy Class is that once the Cutoff Multi-Agent MDP with extended visibility is solved, we may plug various extraction methods into it to attain many exponentially close to optimal group decentralized policies without additional computation (see \cref{extended_cutoff}).

These improvements to performance in the small visibility setting from the Extended Cutoff Policy Class answer our initial motivation from \cref{contributions}, ``\textit{can we improve performance using this policy class when visibility is small and fixed?}''

\section{Conclusion}
We have introduced the Extended Cutoff Policy Class which, to the best of our knowledge, is the first class of non-trivial partially observable policies that are theoretically guaranteed to be exponentially close to optimal with respect to the visibility for all Locally Interdependent Multi-Agent MDPs. These results were replicated in a new setting we propose called the Generalized Locally Interdependent Multi-Agent MDP which allows for transition dependence and extended reward dependence.

In the regime that the visibility is small and fixed where asymptotic guarantees are less helpful, the Extended Cutoff Policy Class presents a rich class of partially observable policies with favorable implicit biases in Locally Interdependent Multi-Agent MDPs. In particular, we show a specific instance of the Extended Cutoff Policy Class called Simple Memory Based Extraction can resolve the Penalty Jittering phenomenon, which prevents the solutions in \cite{deweese2024locally} from performing effectively in the small and fixed visibility regime. Further, in the case of  a deterministic environment where all agents start within view, Simple Memory Based Extraction will achieve fully observable joint optimality.

 \section*{Acknowledgments}
 This research was supported by NSF Grants 2154171, 2339112, CMU CyLab Seed Funding, C3 AI Institute. In addition, Alex DeWeese is supported by Leo Finzi Memorial Fellowship in Electrical \& Computer Engineering and the David H. Barakat and LaVerne Owen-Barakat CIT Dean's Fellowship.
 
\bibliographystyle{abbrvnat} %
\bibliography{refs} %
\newpage
\appendix
\section{Simulations}
The following are simulations performed in the fixed visibility $\mathcal V_{exec}$ setting. Unless otherwise stated, the Manhattan distance metric is used. Similar to \cite{deweese2024locally} the First Step Finite Horizon Optimal Policy is omitted from these simulations because the behavior is greedy and will often have poor performance in the small and fixed visibility $\mathcal V_{exec}$ regime.

Intuitively, the Amalgam policy (Trivial Extraction with $\xi \rightarrow \infty$, $\eta \rightarrow \infty$ proposed in \cite{deweese2024locally}) will take the optimal action for every group, pretending that no other agents exist. The Cutoff Policy (Trivial Extraction with $\xi = 0$, $\eta \rightarrow \infty$ proposed in \cite{deweese2024locally}) will act as though disconnections in the environment are permanent. In other words, if an agent leaves the communication group, they will never be seen again. In the following simulations, we will treat these previously established policies as a baseline comparison for Simple Memory Based Extraction.
 
Note that poor performance in the fixed visibility $\mathcal V_{exec}$ regime does not necessarily contradict our theoretical results. What is established in the Variable Visibility section of \cref{main_results} is an asymptotic result that does not provide a strong guarantee when $\mathcal V_{exec}$ is small and fixed.
\label{simulations}
\subsection{Aisle Walk Problem}
\label{aisle}
\begin{figure}[H]
\begin{center}
\vspace{-2ex}
\includegraphics[width=0.40\linewidth]{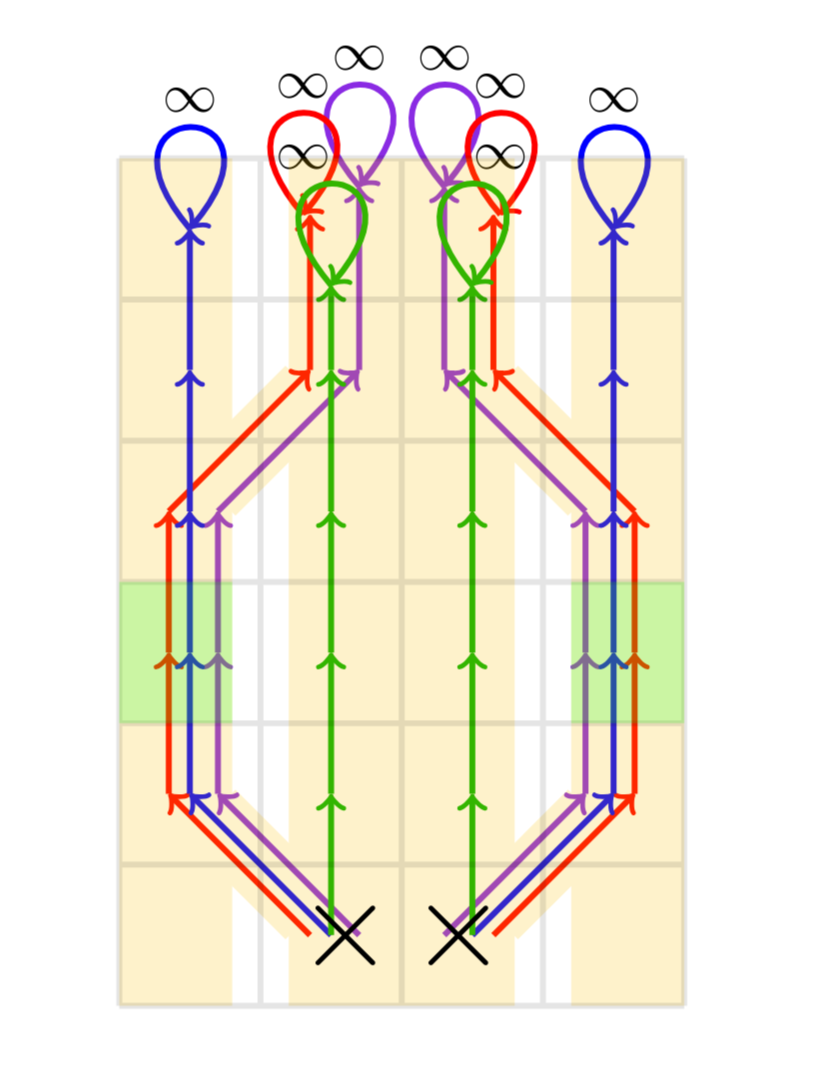}
\vspace{-2ex}

\end{center}
\caption{(Aisle Walk Problem) In red and purple is the roll out for the fully observable joint optimal policy and
the Simple Memory Based Extraction ($\mathcal V_{comp} = 3$) trajectory with discounted sum of rewards 464.44.
In blue is the trajectory for the Amalgam Policy with total discounted rewards 202.00 and in green is the the Cutoff Policy with discounted reward 400.00. All values have been rounded to the second decimal place. Note that these values are different than what is presented in \cite{deweese2024locally} because the rewards in the problem definitions have been modified.}
\label{aisle_fig}
\end{figure}
\noindent Here, we run a Simple Memory Based Extraction policy in the Aisle Walk Problem shown in  \cite{deweese2024locally}.  The purpose of this simulation is to compare the performance of Simple Memory Based Extraction with existing methods and specifically to show that when the assumptions made in \cref{main_results} are held, Simple Memory Based Extraction achieves optimality.
\\\\
The Aisle Walk Problem shown in \cref{aisle_fig} has a visibility of $\mathcal V_{exec} = 2$ and when the agents are within $\mathcal R = 1$ of each other, they each obtain a $+20$ reward. When agents are on the green squares shown in the diagram, they also obtain a reward of $+100$. Agents are required to move forward 1 step at each time step until they reach the top and can maneuver across columns only at specific positions indicated on the diagram. When agents reach the top, they are left to interact for the remainder of the time steps with any nearby agents. The discount factor used is $\gamma = 0.9$.
\\\\
The fully observable joint optimal policy in this situation is to split the agents apart to obtain the single agent rewards $+100$ in the first and last columns and then return together to attain the $+20$ reward for being close to each other. We see that the Amalgam Policy initially recognizes this and splits the agents apart. However, once the agents are out of view, the Amalgam Policy forgets the existence of the other agents and continues going forward. The Cutoff Policy, on the other hand, assumes that it will not have the mechanism to return to the center and compares either going in a straight line or splitting apart and not being able to reconvene. It decides to go in a straight line and achieves a higher reward than the Amalgam Policy. However, Simple Memory Based Extraction with $\mathcal V_{comp} = 3$ is able to remember the position of the agent, and since the assumptions stated in \cref{main_results} are held, the policy is able to achieve fully observable joint optimal behavior.

\subsection{Bullseye Problem}
\label{bullseye}

\begin{figure}[H]
    \centering

\includegraphics[width=0.62\linewidth]{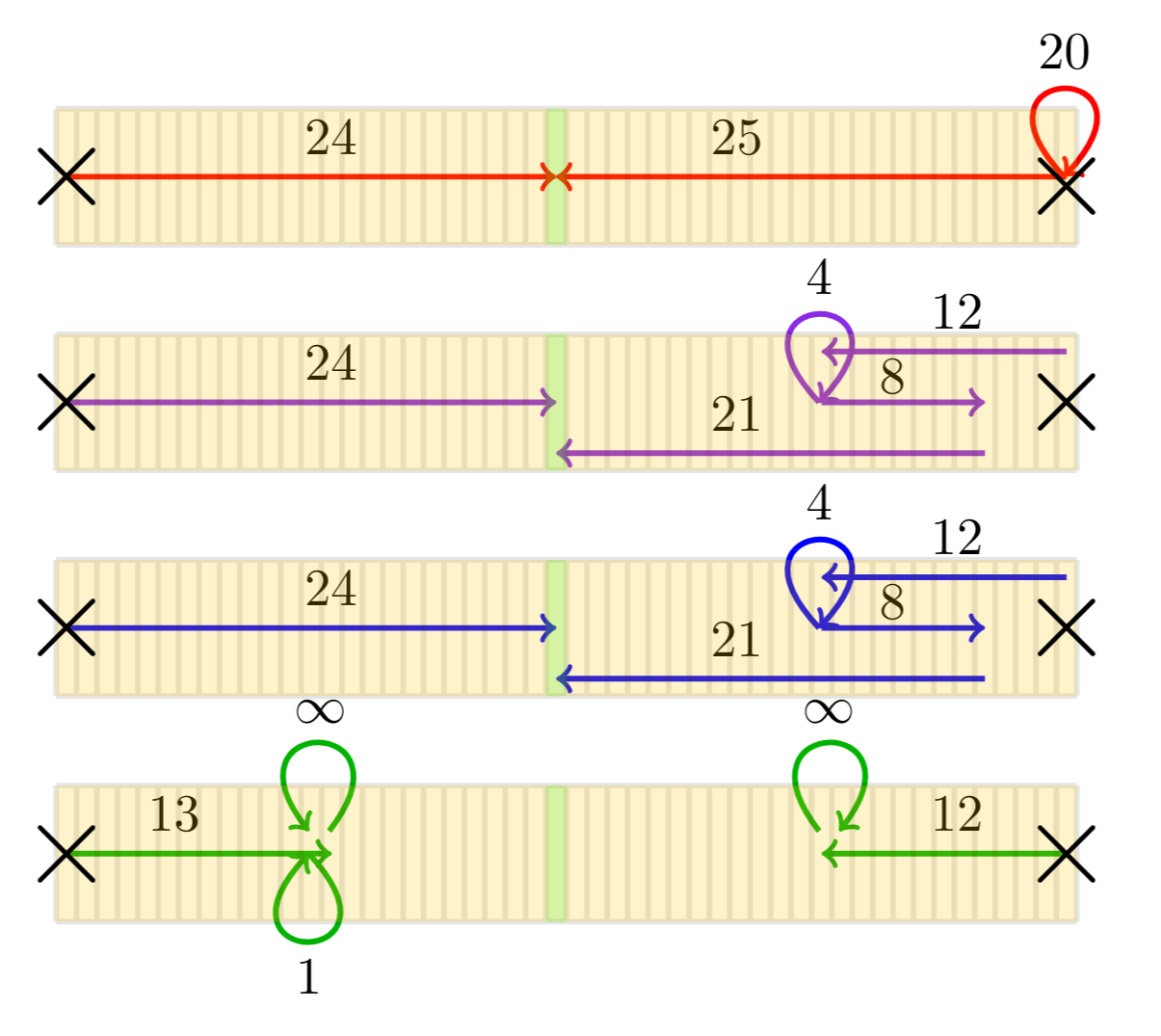}
    \caption{(Bullseye Problem) In red is the fully observable joint optimal policy with discounted sum of rewards 8.85. In purple and blue is the Amalgam and Simple Memory Based Extraction ($\mathcal V_{comp} = 30$) with discounted sum of rewards 6.74. In green is the Cutoff Policy with discounted sum of rewards -5.38. All values are rounded to the second decimal place.}
    \label{bullseye_fig}
\end{figure}
\noindent We will now evaluate the performance of Simple Memory Based Extraction in the Bullseye Problem in \cite{deweese2024locally}. 
This simulation will provide a situation where the Simple Memory Based Extraction may not achieve optimality, introduce the Penalty Jittering phenomenon discussed in \cref{penalty_oscillation}, and set up the simulation introduced in the next section.
\\\\
Each rollout in \cref{bullseye_fig} shows two agents. One agent begins 24 spaces to the left and the other agent 25 spaces to the right from a central bullseye. Agents will receive a $-2$ penalty for moving away from the bullseye, a penalty of $-500$ per agent if they are within $\mathcal R = 20$ spaces of one another, and a reward of $+100$ if they reach the central bullseye. The goal of the agents is to reach the central bullseye without collision with a visibility of $\mathcal V_{exec} = 25$. Once an agent reaches the bullseye, the agent will no longer interact with others and receive no interdependent penalties with any other agent (formally it will be sent along a long line of states starting at the bullseye moving away from the grid). The discount factor used is $\gamma = 0.9$.
\\\\
In this case, the fully observable joint optimal policy is to give the left agent that is closer to the bullseye $20$ steps right of way before following behind with the right agent. However, since the visibility is limited, the agents will not be able to achieve this behavior and will need to backtrack. Once the agents are in view, the Amalgam Policy and Simple Memory Based Extraction with $\mathcal V_{comp} = 30$ are able to perform back tracking to give the closer agent the right of way and eventually bring both agents to the center. The Cutoff Policy on the other hand, believes that if the agent leaves its visibility, will disappear completely because of the permanent disconnection property of the Cutoff Multi-Agent MDP (see \cref{cutoff}). So, the policy moves the agents outwards and tries to go towards the center again. This repeats indefinitely, causing a continual back tracking. This is an example of the Penalty Jittering phenomenon we discuss in \cref{penalty_oscillation}.

\subsection{Modified Bullseye Problem}
\label{bullseye_mod}
\begin{figure}[H]
\centering

\includegraphics[width=0.60\linewidth]{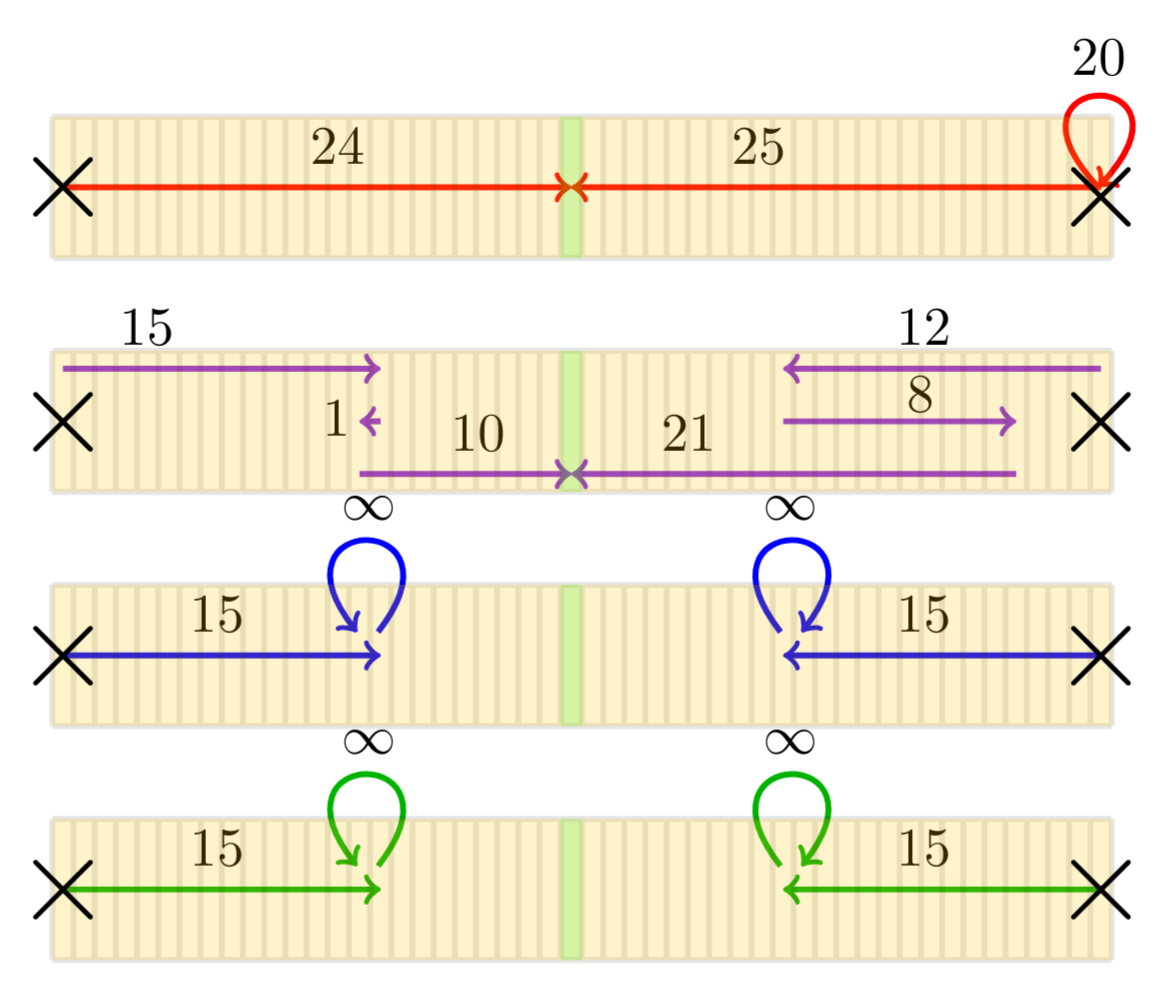}

\caption{
(Modified Bullseye Problem) Again, in red is the rollout for the fully observable joint optimal policy with discounted sum of rewards 8.85. In purple is Simple Memory Based Extraction ($\mathcal V_{comp} = 30$) with discounted sum of rewards -201.96. The Amalgam Policy and Cutoff Policy in blue and green respectively have value -1087.97. All values are rounded to the second decimal place.
}
\label{bullseye_mod_fig}
\end{figure}
\noindent Here, we show that with a minor modification to the Bullseye Problem (\cref{bullseye}), Penalty Jittering can be induced in the Amalgam Policy as well. On the other hand, Simple Memory Based Extraction will continue to perform reasonably in this scenario.
\\\\
Shown in \cref{bullseye_mod_fig} is the same scenario as the Bullseye Problem (\cref{bullseye}) except the visibility is set to $\mathcal V_{exec} = 20.5$. Since the step size is always $1$, agents will not be able to see each other until they collide. This is equivalent to setting $\mathcal V_{exec} = 20$ but adheres to $\mathcal V_{exec} > \mathcal R$. The fully observable joint optimal policy in this situation remains to give the left agent 20 steps right of way. The Cutoff Policy will also fail for the same reason as in the Bullseye Problem (\cref{bullseye}). The Amalgam Policy, however, will also fail in this situation. The optimal action that will be followed by the Amalgam Policy (see beginning of \cref{simulations} for intuition) 
pushes the agents apart when they collide, but when this is done, the agents are pushed out of view, and they will forget about each other. The optimal policy for each individual agent on its own is to move towards the center. This repeats indefinitely, producing the jittering behavior described in \cref{penalty_oscillation}. Simple Memory Based Extraction (still with $\mathcal V_{comp} = 30$) is able to remember agents out of view and is able to successfully backtrack and bring both agents to the center while adhering to our visibility constraints.

\subsection{Highway Problem}
\label{highway}

\begin{figure}[H]
\centering

\includegraphics[width=0.32\linewidth]{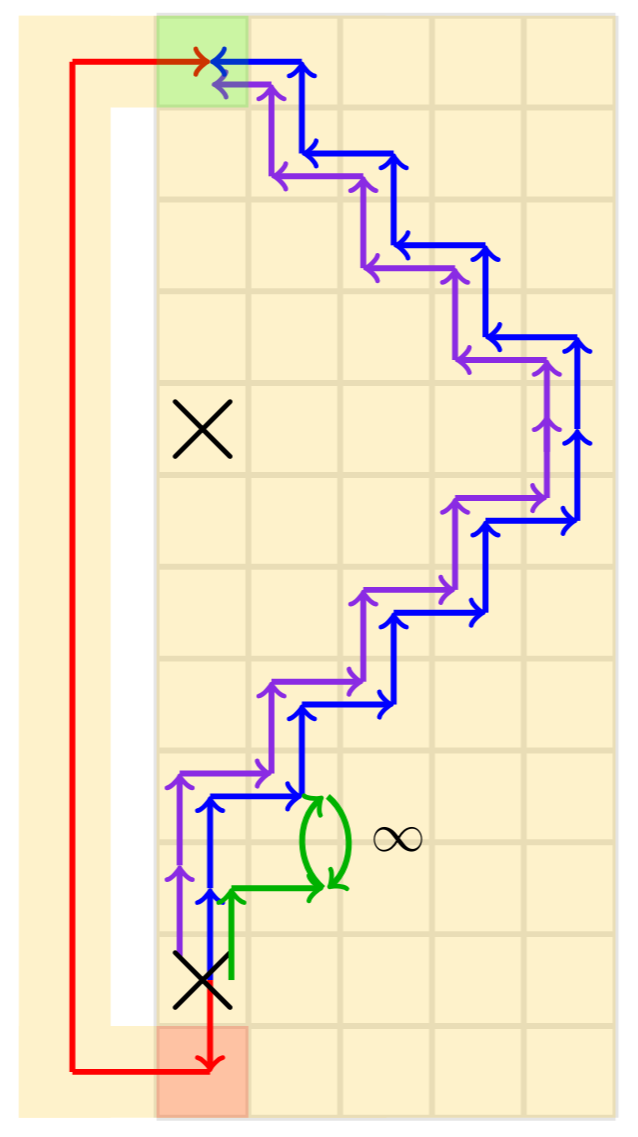}

\caption{
(Highway Problem) In red is the rollout for the fully observable joint optimal policy with discounted sum of rewards 73.50. In blue and purple is the Amalgam Policy and Simple Memory Based Extraction ($\mathcal V_{comp} = 9$) with value 70.93 and in green is the Cutoff Policy with value 0. All values are rounded to the second decimal place.
}
\label{highway_fig}
\end{figure}

This simulation will again evaluate the performance of Simple Memory Based Extraction in a previously established environment from \cite{deweese2024locally} and set up the next section.

Illustrated in \cref{highway_fig}, agents will receive a reward of $+100$ just once when they reach the square at the top left indicated in green. Agents must reach this point while avoiding getting within $\mathcal R = 3$ spaces of each other in which they will incur a penalty of $-500$. There is also a square on the bottom left indicated in red where agents can be transported to the green square on the top left, but will incur a penalty of $-25$. One agent will be placed towards the center and will not be able to move. The other agent will start towards the bottom left and will decide whether to take the quicker ``highway'' incurring the penalty or travel there using the longer path to get to the green square on the top left. Agents will act according to a visibility $\mathcal V_{exec} = 5$. For the distance metric, to ensure that the highway does not violate our step movement of at most 1, we take the graph distance while treating each square as a node in a graph and possible movement directions (up, down, left, and right) as edges. The discount factor used is $\gamma = 0.98$.
 
The fully observable joint optimal solution is to take the highway, as the agent blocking the path will require more time to maneuver around, discounting our $+100$ reward. However, this simulation is designed so that if the highway is not taken initially, the optimal path is to maneuver around the center agent and continue taking the longer route.  The partially observable policies will not initially see the center agent and therefore will take the longer path. The Amalgam Policy and Simple Memory Based Extraction with $\mathcal V_{comp} = 9$ are able to successfully maneuver around the center agent and move to the green square on the top left. The Cutoff Policy on the other hand, once the center agent is in view, acts as though disconnections are permanent. So the agent will decide to move one step away anticipating a permanent disconnection and then once again moves forward. This repeats indefinitely resulting in the Penalty Jittering phenomenon described in \cref{penalty_oscillation}.

\subsection{Modified Highway Problem}
\label{highway_mod}
\begin{figure}[H]
\centering
\label{modified_highway}

\includegraphics[width=0.30\linewidth]{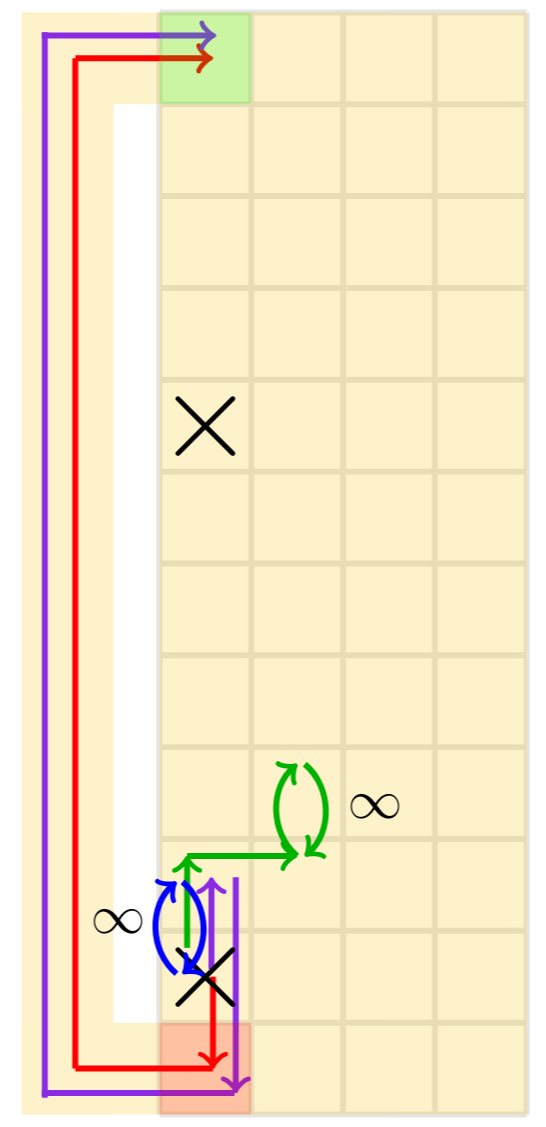}

\caption{
(Modified Highway Problem) In red is the rollout for the fully observable joint optimal policy with discounted sum of rewards 73.50. Simple Memory Based Extraction ($\mathcal V_{comp} = 9$) is in purple with value 70.59. The Cutoff Policy and Amalgam Policy are in green and blue respectively with value 0. All values are rounded to the second decimal place.
}
\label{highway_mod_fig}
\end{figure}

We will make a minor modification to the Highway Problem (\cref{highway}) and once again induce Penalty Jittering in the Amalgam Policy. Simple Memory Based Extraction is again able to succeed in producing reasonable results in this environment.
 
For this simulation shown in \cref{highway_mod_fig}, we take the Highway Problem (\cref{highway}) and reduce the length of the grid from 5 to 4. Since $\mathcal R = 3$ and the center agent allocates a square, agents are not able to pass in this direction. In other words, the highway is the only path to the $+100$ reward in the top left. The fully observable joint optimal policy remains taking the highway. The Amalgam policy takes a step forward and sees the agent and moves back to take the highway. However, when the agent back tracks, the agent forgets about the agent in the center and takes a step forward again which continues indefinitely. The Cutoff Policy fails for a similar reason to the Highway Problem (\cref{highway}). The Simple Memory Based Extraction on the other hand is able to remember the placement of the agent and succeed in taking the highway to achieve the $+100$ reward on the top left.

\subsection{Penalty Jittering}
\label{penalty_oscillation}
\begin{figure}[H]
\centering

\includegraphics[width=0.75\linewidth]{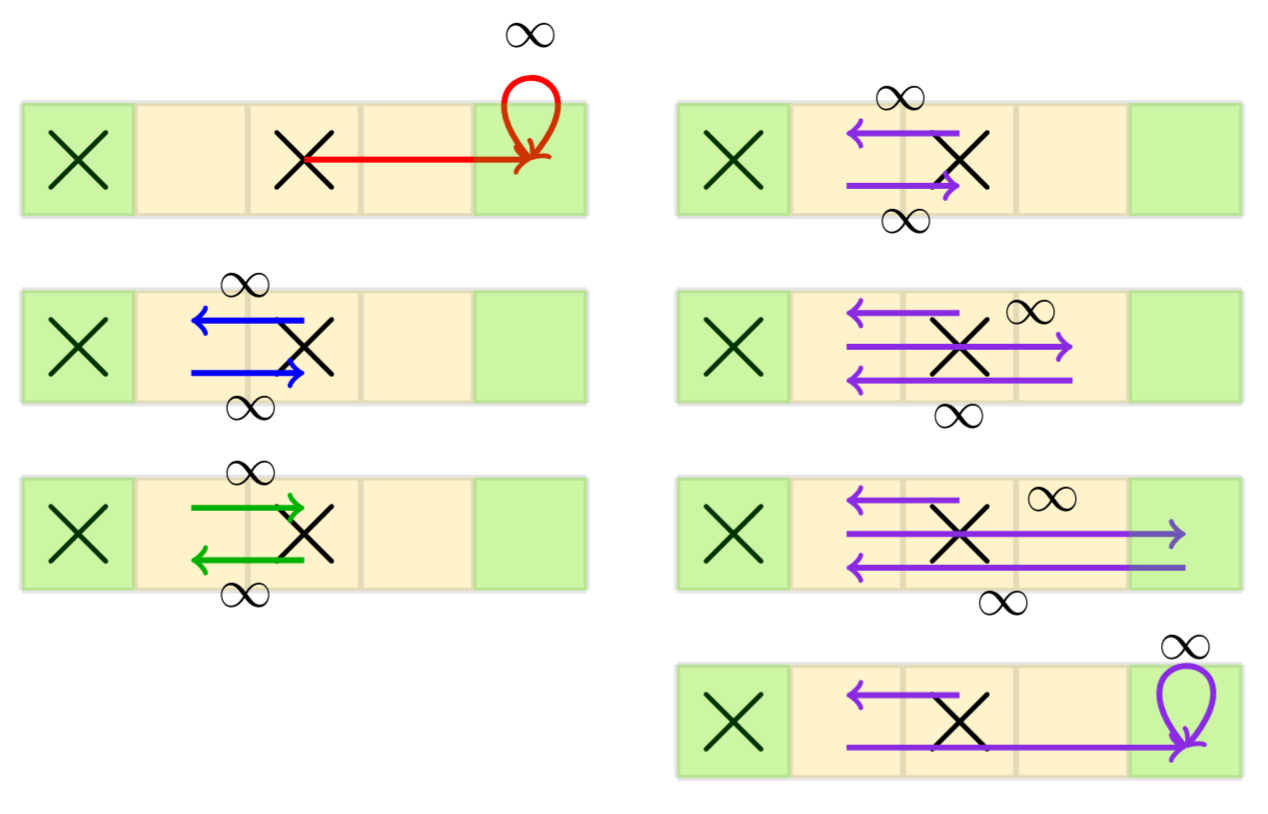}

\caption{
(Penalty Jittering) In red is the rollout for the fully observable joint optimal policy with discounted sum of rewards 2405.00. In blue and green is the Amalgam Policy and Cutoff Policy respectively with value 2000.00. In purple is Simple Memory Based Extraction with increasing $\mathcal V_{comp}$. Starting at the top $\mathcal V_{comp} = 1$ and $\mathcal V_{comp} = 2$ has value 2000.00. $\mathcal V_{comp} = 3$ has value 2070.01 and $\mathcal V_{comp} = 4$ has value 2328.05. All values are rounded to the second decimal place.
}
\label{penalty_jit_fig}
\end{figure}

This example will demonstrate the Penalty Jittering phenomenon and show an illustrative example of how Penalty Jittering is overcome by Simple Memory Based Extraction. 
 
Intuitively, Penalty Jittering occurs when the partially observable agent is pushed out of a communication group towards some main objective. When the agent leaves the communication group, the agent forgets about the communication group it was previously a part of and acts according to its own interest towards another objective, which happens to be in the direction of the communication group it just left. It goes back and forth entering and exiting this communication group creating a ``jittering'' behavior and completely stopping the movement of this agent. This primarily occurs when the interdependent dependencies are negative penalties because if the interdependent dependencies are positive, leaving and forgetting the local communication group will only make the original main objective more enticing and the agent will not re-enter the communication group. We demonstrate this occurrence concretely with the following simulation.
 
Shown in \cref{penalty_jit_fig}, agents can obtain a reward of $+200$ for moving to the leftmost square and a reward of $+50$ for moving to the rightmost square and are penalized by $-500$ for overlapping ($\mathcal R = 0$). One agent will be placed in the leftmost square and another two spaces to the right. The visibility of all agents is $\mathcal V_{exec} = 1$ and the discount factor is $\gamma = 0.9$.
 
The fully observable joint optimal policy is to keep the left agent in the leftmost square and move the right agent to the rightmost square. Unfortunately, all Trivial Extraction policies, such as the Cutoff Policy and Amalgam Policy, will struggle in this example because of the memorylessness of the agents. Since the reward for the leftmost state is higher, the right agent will move left. However, the agent will see that the next space is occupied ($\mathcal V_{exec} = 1$) and decide that moving the agent toward the rightmost square will lead to higher discounted rewards. Once the right agent moves one space to the right, the agent will forget about the existence of the left agent. The right agent will then move back to the left. This repeats indefinitely, resulting in the Penalty Jittering phenomenon.
 
This is a fundamental occurrence due to the lack of memory of the agents and stops the movement of the agents (see \cref{bullseye}, \cref{bullseye_mod}, \cref{highway}, and \cref{highway_mod} for more examples). This heavily detracts from the viability of the closed form policies proposed in \cite{deweese2024locally} in the small and fixed $\mathcal V_{exec}$ regime.
 
Simple Memory Based Extraction allows the policy to remember the placement of the other agent and to go closer and closer to the rightmost state until it eventually reaches and stays at the rightmost square. However, notice that for smaller limited computational visibilities $\mathcal V_{comp}$, the back and forth movements are still present. These movements become longer with the increase in the computational visibility $\mathcal V_{comp}$ until it stays at the rightmost square. In other words, the Penalty Jittering becomes more of a ``Penalty Oscillation'' for various computation visibilities in Simple Memory Based Extraction. This kind of oscillation can occur frequently in Simple Memory Based Extraction simulations similar to Penalty Jittering but can often be resolved by simply increasing the computational visibility $\mathcal V_{comp}$.

\subsection{Long Journey}
\label{journey}
\begin{figure}[H]
\centering

\includegraphics[width=0.80\linewidth]{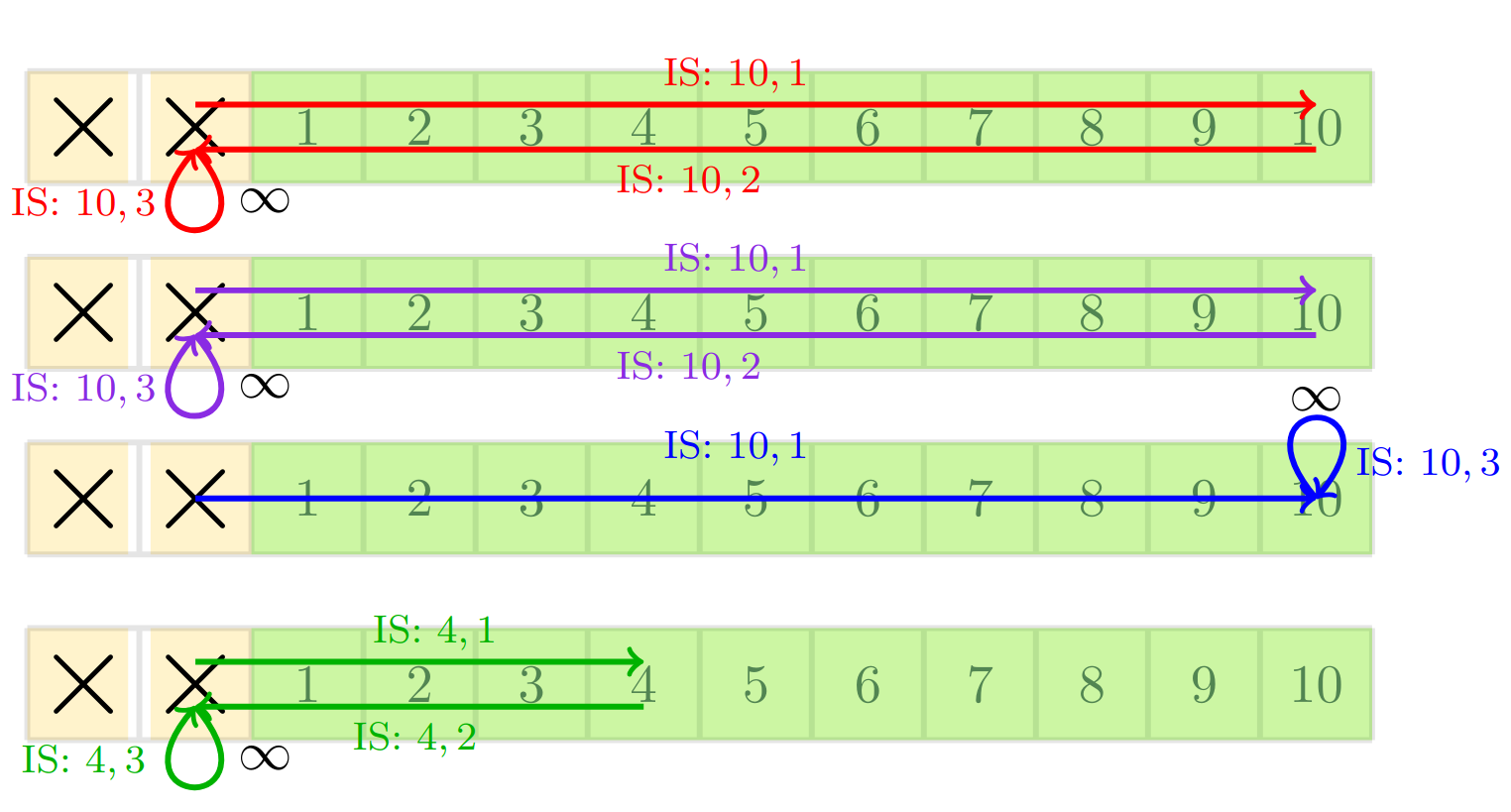}

\caption{
(Long Journey) In red and purple is the trajectory for the fully observable joint optimal policy and Simple Memory Based Extraction ($\mathcal V_{comp} = 11$) respectively with discounted sum of rewards 200. In blue is the rollout for the Amalgam Policy with discounted sum of rewards 10. In green is the trajectory for the Cutoff Policy with discounted sum of rewards 140. 
}
\label{journey_fig}
\end{figure}
This simulation will compare the performance of Trivial Extraction and Simple Memory Based extraction as a whole (for every relevant $\xi, \eta$) on an illustrative example. This is displayed in a table similar to that shown in \cref{corners}, but with a concrete example. This table will also illustrate the effect of the increase in the computational horizon $\eta$.
 
As shown in \cref{journey_fig}, agents will not only have a position in a common state space but also an internal state (see \cref{limdp}) consisting of a tuple of two values that will be specified below.  The discount factor $\gamma$ can be an arbitrary value in $(0, 1)$. One agent will be fixed in the leftmost square and another agent will be placed one space to the right. The free agent on the right may freely move through the numbered squares but must choose a number between $x = 1,\ldots, 10$ as its objective, which will be set as its first internal state. When the agent reaches the numbered square corresponding to its choice, the agent can then choose to stay (preventing further movement) by changing its second internal state to $3$ and obtaining a one-time reward of $\frac{(10 - x + 1)\times 10}{\gamma^x}$. The agent may instead choose an internal state of $2$ that allows a return to its original position to attain a one-time reward of $\frac{100 + 10x}{\gamma^{2x}}$ for being within $\mathcal R = 1$ of the leftmost agent. This will transfer the right agent to a second internal state of 3 preventing further movement. 
Here, the discount factors in the denominator adjust for the discounting that will be placed on the reward for the steps required to move to the square. The visibility in this setting is $\mathcal V_{exec} = 5$.
 
This example is designed so that the highest rewards can be achieved by traveling to and returning from higher-numbered grid spaces. On the other hand, if the free agent fails to make the round trip, attempting higher number grid spaces will diminish the reward. 
 
The fully observable joint optimal policy is to make the round trip to square 10. Simple Memory Based Extraction can achieve this optimal behavior for sufficiently large $\xi$ since the conditions described in \cref{main_results} hold. The Amalgam Policy attempts the highest $10$ square but is unable to complete the round trip because of the memorylessness. The Cutoff Policy is able to take this into account and aims for the largest numbered square that can be completed while staying within $\mathcal V_{exec} = 5$ of the other agent. 
 
A table of performance outcomes for both Trivial Extraction and Simple Memory Based Extraction is provided in \cref{heatmap}. Intuitively, Trivial Extraction in this scenario is not able to effectively utilize the extra computation that is performed as in Simple Memory Based Extraction. Increasing $\eta$ (with horizon $c + \eta$ where $c = \lfloor \frac{\mathcal V_{exec} - \mathcal R}{2}  \rfloor = 2$) allows agents to consider rewards farther away at higher number squares and increasing $\xi$ allows us to compute the behavior of agents farther away ($\mathcal V_{comp} = \mathcal V_{exec} + \xi$). However, Trivial Extraction is not able to effectively use this information for this example because of the lack of memory beyond its visibility. On the other hand, we see that Simple Memory Based Extraction is able to steadily incorporate this information leading to better policies as we increase the computational visibility eventually leading to the fully observable joint optimal policy. The effect of increasing $\xi$ on the computation time is discussed in \cref{scalable}.

\newpage
\begin{figure}[H]
    \centering
    
\includegraphics[width=1.00\linewidth]{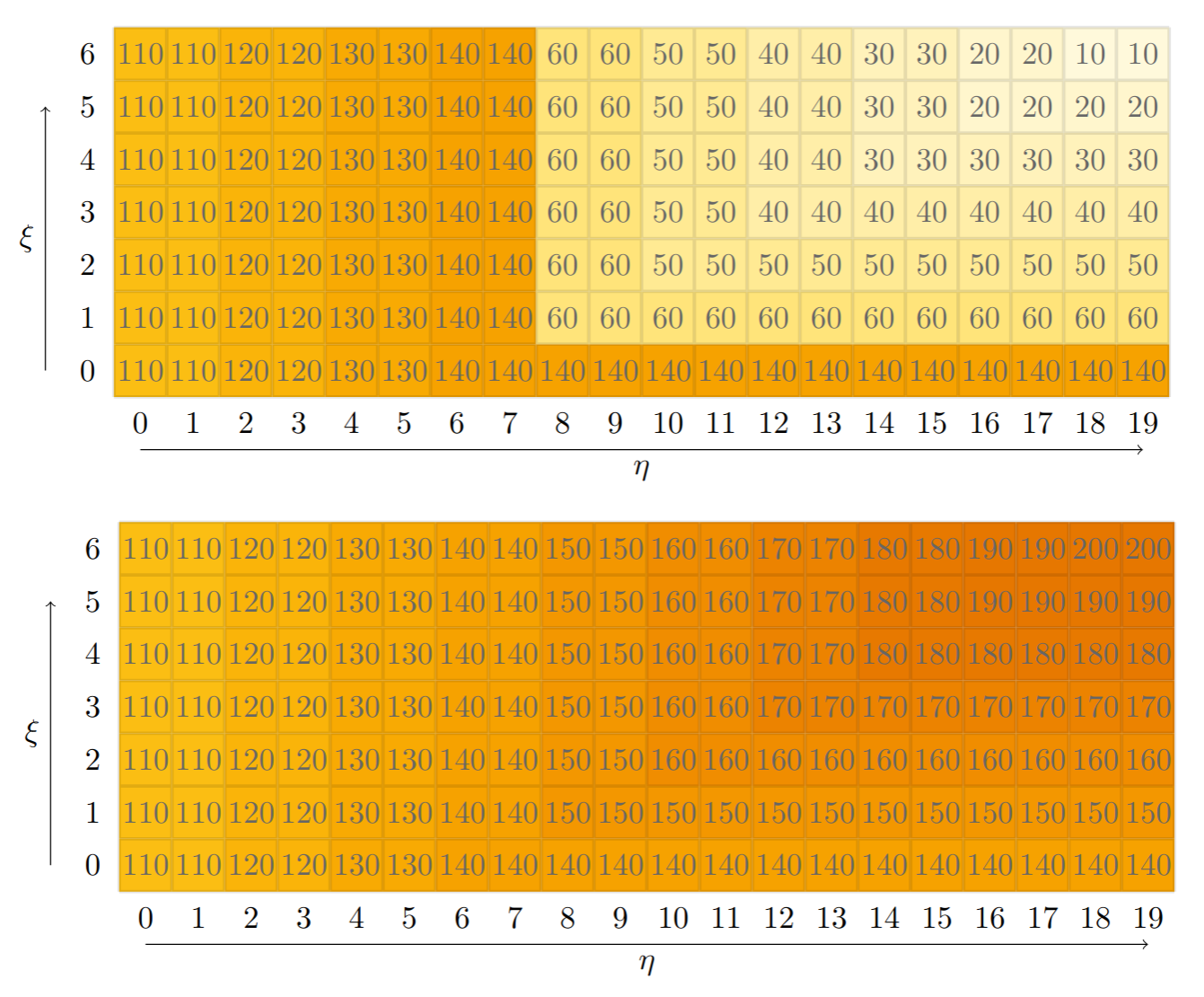}

    \caption{
    A table of performance outcomes in the Long Journey simulation for Trivial Extraction (top) and Simple Memory Based Extraction (bottom). Notice the resemblance of the top Trivial Extraction Table to  \cref{corners}. The corners correspond to the Amalgam Policy (top right), Finite Amalgam Policy (top left), Cutoff Policy (bottom right), First Step Finite Horizon Optimal Policy (bottom left).}
    \label{heatmap}
\end{figure}

\subsection{Random Navigation With Many Agents}
\begin{figure}[H]
    \centering
    \includegraphics[width=130mm]{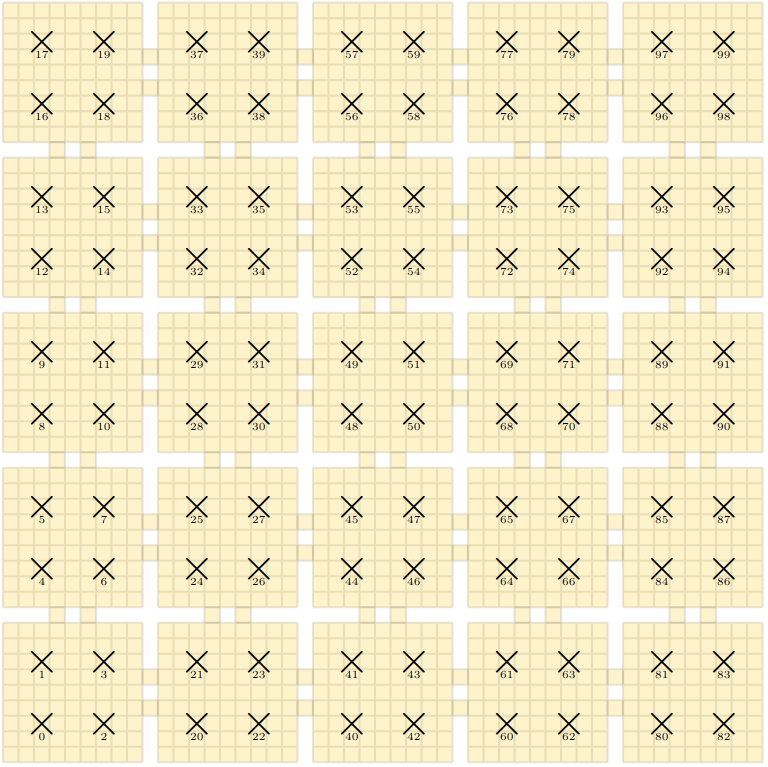}
    \caption{
    (Random Navigation With Many Agents) The environment and the initial placement of the 100 agents for the Random Navigation With Many Agents simulation.
    }
    \label{b_many_start}
\end{figure}
\label{b_many}
This is an example of how large group heuristics can be used to simulate coordination behavior with a many agents. This large-scale simulation does not use deep learning as is often the case in this literature (see \cite{long2018towards} for an example).
 
100 agents will be placed in an environment with various rooms, each aiming for random objectives starting with the placement in \cref{b_many_start}. Every agent in a room is randomly assigned one of the exits to the room. Agents will attempt to claim a reward of $+100$ for reaching the assigned exit, but will receive a penalty of $-500$ if the agent is within $\mathcal R = 1$ of another agent and needs to navigate with a visibility of $\mathcal V_{exec} = 3$. At the exit, the agent is transferred to an intermediate waiting space until there is clearance for the agent in the next room (no agents within $\mathcal V_{exec}$ of the entrance). Agents inside the waiting space will not collide with other agents. Intuitively, this waiting space models a long road which serves as a queue until space is available in the next area. 
 
The rooms in this modular setup act as a second layer of locality. At any moment agents across rooms cannot interact with each other so each agent only needs to consider the agents in the same room when navigating. However, we still need to create a mechanism to handle the potentially large number of agents that may be within a room at a single time, which will still easily make the computation intractable. Here, we will only compute the values for the extended Cutoff Multi-Agent MDP for 3 agents with $\mathcal V_{comp} = 5$ with discount factor $\gamma = 0.9$ and use a stochastic heuristic when group sizes are $\geq 4$. Specifically, we will generate a greedy action for each agent towards their objective with 0.8 probability, and a completely random action otherwise. We will sample this distribution until we have an action that does not result in a collision. If this action is not found after a large number of samples, we take a random action (effectively resulting in a collision). Although this is a fairly simple heuristic, we will soon see that the heuristic is not relied upon significantly.
 
Here, we will use a variation on Simple Memory Based Extraction with visibility $\mathcal V_{comp} = 5$ that has three modifications to \cref{simple_algo} in \cref{simple_mem}. Firstly, when group sizes in memory (belief states $[s^{belief}]_{z^{belief}}$) are 4 agents or larger, we will clear our memory. This is because the stochasticity of our heuristic will interfere with our estimation of other agents. Secondly, if the actual group size is 4 or more agents, we will take the stochastic heuristic. Lastly, if an agent in memory has a distance that exceeds $\mathcal V_{comp}$, we will remove it from memory. This will reduce the number of times our memory is cleared from large belief states.
 
We begin with the initial state shown in \cref{b_many_start} and show snapshots of the placement of the agents and the trajectory for 1000 time steps of two agents (agents 0 and 99) in \cref{b_many_trajectory}. We consider an objective to be complete when an agent successfully exits a room. The following are statistics taken during the rollout.
 
\begin{center}
\begin{tabular}{ |p{5cm}||p{1cm}|p{1cm}|p{1cm}| p{1cm}|  }
 \hline
 \multicolumn{5}{|c|}{Rollout Statistics} \\
 \hline
Timestep & 50 & 150 & 500  & 1000 \\
 \hline
average objectives per agent&  3.58 & 11.28 & 37.92 & 77.05\\
minimum objectives by agent&  2 & 7 & 31 & 67\\
maximum objectives by agent&  5 & 14 & 42 & 86\\
total objectives complete &  358 & 1128 & 3792 & 7705\\
max agents in single room & 10 & 11 & 13 & 13\\
max group size & 6 & 6 &6 & 7\\
average group size & 1.19 & 1.17 & 1.15 & 1.15\\
proportion of group size 1 & 86.7\%&  88.4\% & 89.4\% & 89.4\%\\
proportion of group size 2& 9.13\% & 8.00\% & 7.70\% & 7.60\%\\
proportion of group size 3&  2.83\% & 2.64\% & 2.02\% & 2.07\%\\
proportion of group size $\geq$ 4&  1.28\% & 1.01\% & 0.876\% & 0.890\%\\
collisions &  0 & 0 & 0 & 0\\

 \hline
\end{tabular}
\end{center}
We see that the agents do not not collide at any point. Further, since the minimum number of objectives achieved by an agent increases with time, we can infer that agents do not get stuck with the Penalty Jittering phenomenon or with congestion. As expected, the number of total objectives complete increases steadily with time. We see that the heuristics are taken approximately $0.9\%$ over the course of 1000 time steps and much more often rely on our computed values. This means that our tractably computed extended Cutoff Multi-Agent MDP values account for more than $99\%$ of the group actions taken in this system.
\\

\newpage

\begin{figure}[H]
\centering

\includegraphics[width=155mm]{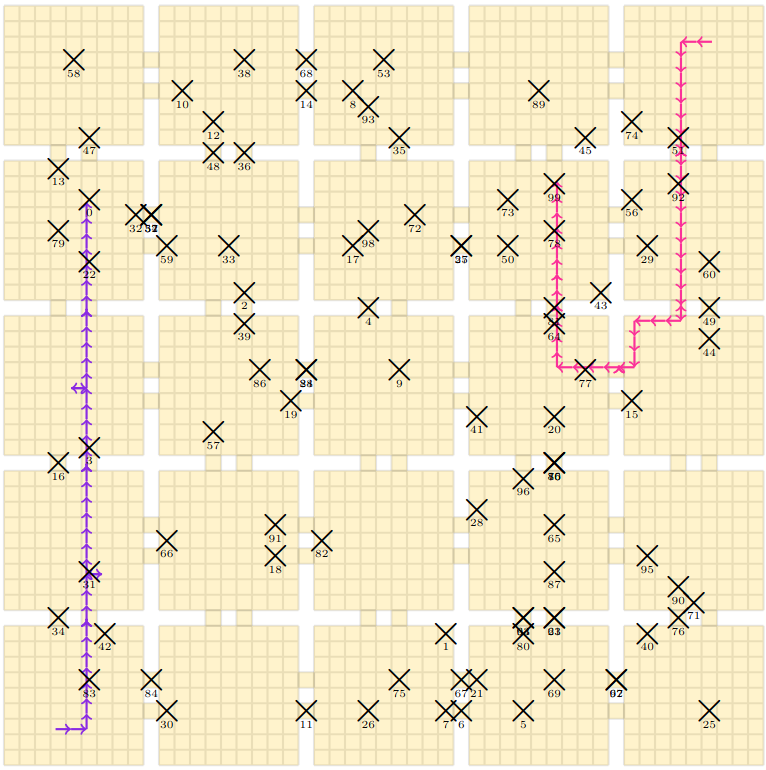}
     \caption{(Random Navigation With Many Agents) Snapshot at time step 50. Trajectories are traced for agents 0 and 99}
     \label{b_many_trajectory}
\end{figure}

\begin{figure}[H]
\centering

\includegraphics[width=155mm]{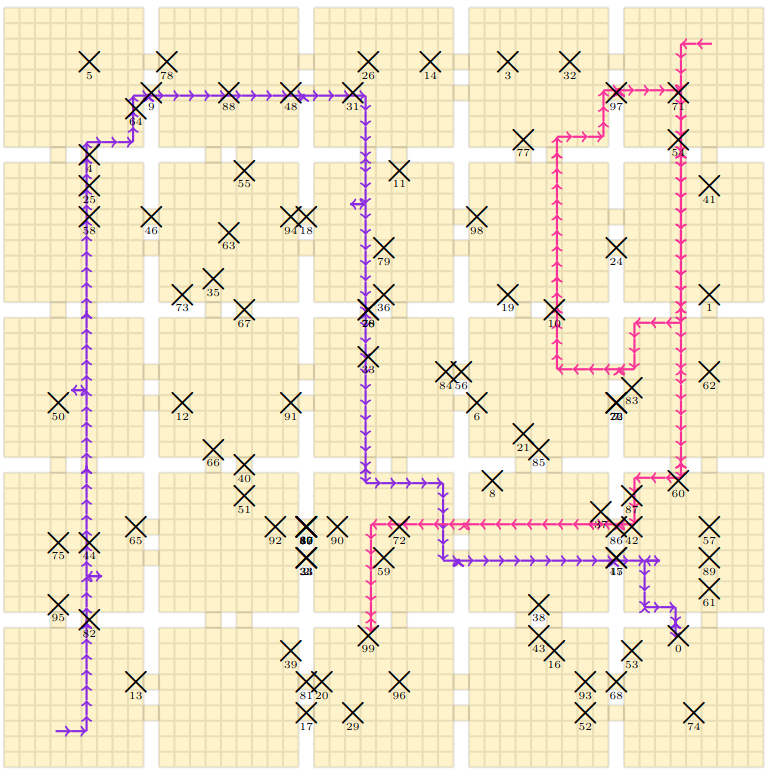}
     \caption{(Random Navigation With Many Agents) Snapshot at time step 150. Trajectories are traced for agents 0 and 99}
     \label{b_many_trajectory}
\end{figure}

\begin{figure}[H]
\centering

\includegraphics[width=155mm]{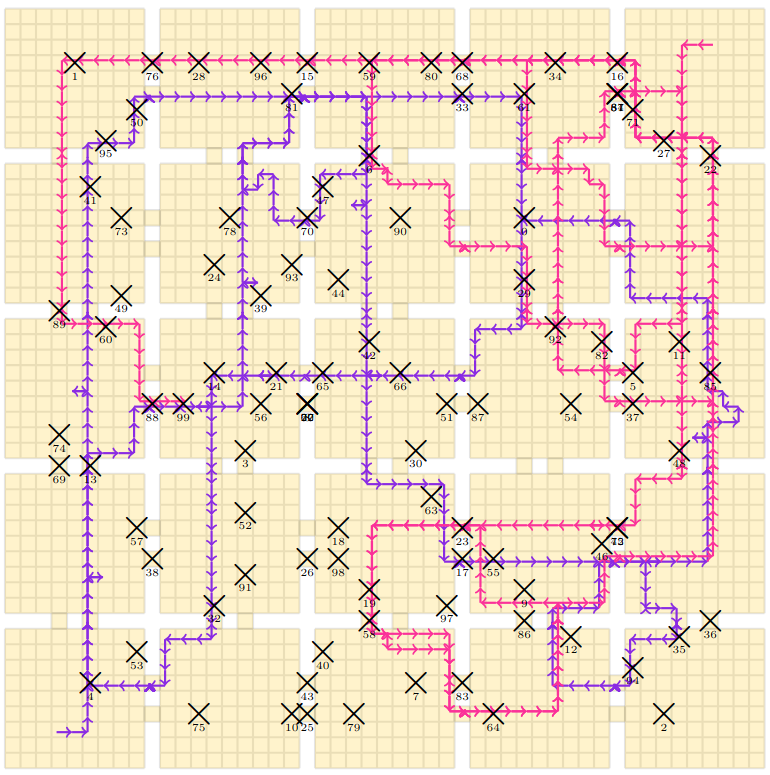}
     \caption{(Random Navigation With Many Agents) Snapshot at time step 500. Trajectories are traced for agents 0 and 99}
     \label{b_many_trajectory}
\end{figure}

\begin{figure}[H]
\centering

\includegraphics[width=155mm]{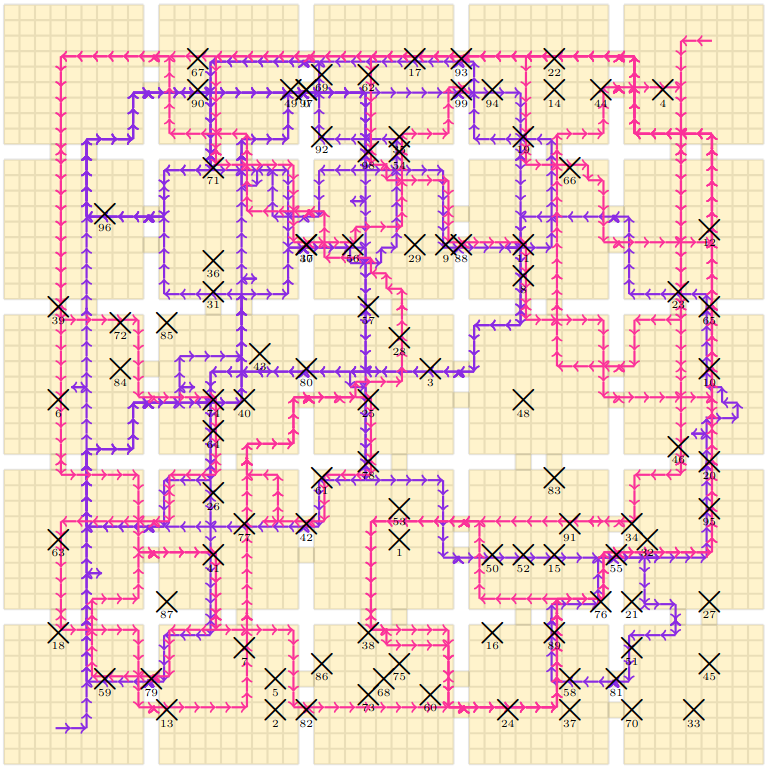}
     \caption{(Random Navigation With Many Agents) Snapshot at time step 1000. Trajectories are traced for agents 0 and 99}
     \label{b_many_trajectory}
\end{figure}

\subsection{Stochastic Transitions}
\label{stochastic}

\begin{figure}[H]
\centering
\includegraphics[width=0.42\linewidth]{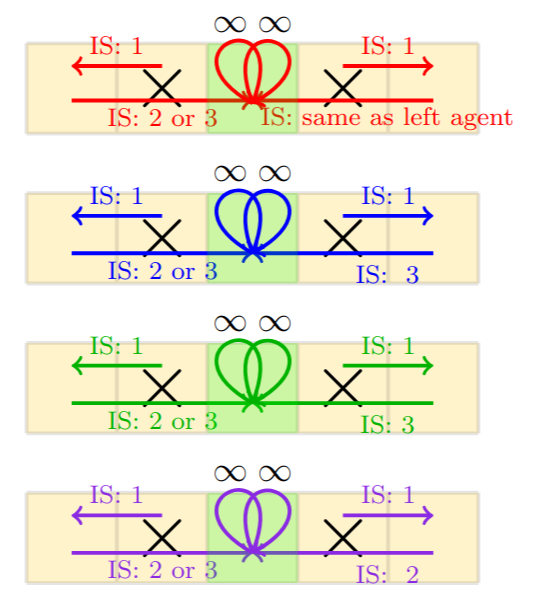}

\caption{(Stochastic Transitions) In red is the rollout for the fully observable joint optimal policy with value 160.09. In blue and green is the Amalgam Policy and Cutoff Policy respectively with value 93.17. Lastly, in purple is the rollout for Simple Memory Based Extraction ($\mathcal V_{comp} = 4$) with value 81.50. All values are rounded to the second decimal place.}
\end{figure}

This simulation is an adversarially chosen example for Simple Memory Based Extraction which introduces stochasticity into the environment in such a way that it results in lower performance than the Amalgam Policy and Cutoff Policy. The stochasticity breaks one of the conditions for the optimality of Simple Memory Based Extraction described in \cref{main_results}. In this case, Trivial Extraction may out-perform Simple Memory Based Extraction even for large computational visibilities.
 
Agents initially start within view of each other with the visibility of $V_{exec} = 3$. Agents will be deterministically pushed one step outward with an internal state 1 and then pushed toward the center with an internal state of 2 or 3. The left agent will be assigned 2 with probability $51\%$ and 3 otherwise. The right agent may choose between internal states 2 and 3 in the rightmost state where the left agent is out of view. In the center, an agent with internal state 3 will obtain a recurring $+2$ reward, and if the two agents overlap in the center $\mathcal R = 0$ with the same internal state, they will both receive an additional recurring $+10$ reward. The discount factor is $\gamma = 0.9$.
 
The fully observable joint optimal policy is to simply match the internal state of the right agent to the left agent. The Amalgam Policy and Cutoff Policy will simply take the greedy option for the right agent and always choose 3. On the other hand, Simple Memory Based Extraction will use that the left agent has a higher probability of being assigned internal state 2, and choose 2 instead. This will result in a lower performance than the Trivial Extraction policies.
 
Note that this can be resolved by modifying Simple Memory Based Extraction to take into consideration stochasticity more delicately such as removing agents from memory when the confidence is low.

\subsection{Unanticipated Out Of View Agent}
\label{oov}

\begin{figure}[H]
\centering

\includegraphics[width=0.70\linewidth]{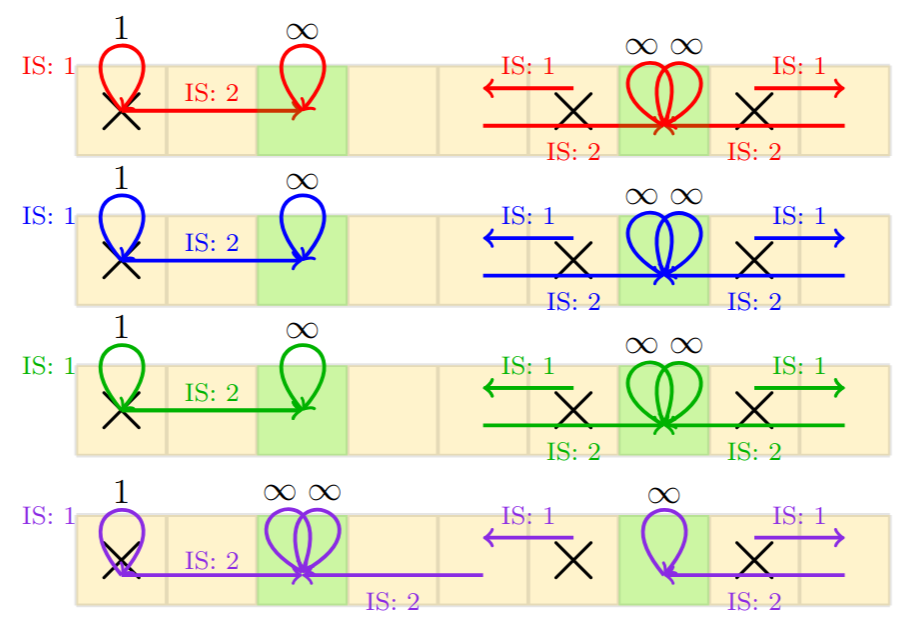}

\caption{(Unanticipated Reward Dependence) In red, blue, and green are rollouts for the trajectories for the fully observable joint optimal policy, Amalgam Policy, and Cutoff Policy respectively with discounted sum of rewards 8748.00. In purple is the trajectory for Simple Memory Based Extraction ($\mathcal V_{comp} = 4$) with value 7932.59. All values are rounded to the second decimal place.}
\end{figure}

This simulation is another adversarially chosen example for Simple Memory Based Extraction which is deterministic but not all agents are initially in view. Specifically, an out-of-view agent will cause an unforeseen penalty. This violates one of the conditions that ensures fully observable joint optimal behavior for Simple Memory Based Extraction. In this example as well, Trivial Extraction out performs Simple Memory Based Extraction even for large computational visibilities.
 
In this simulation, there are two squares indicated in green; the left square gives a recurring $+400$ reward and the right square gives a recurring $+500$ reward. Agents will receive a penalty $-100$ for overlapping $\mathcal R = 0$. The discount factor is $\gamma = 0.9$. The two agents on the right initially start within view $\mathcal V_{exec} = 3$ with internal state 1 and then are forced to move one step outward. They will transition to internal state 2 then move towards one of the squares without any information of the agent on the left. The agent on the left starts in internal state 1 and stays in its position changing to internal state 2, then will move to the closest reward square. 
 
Since collision is inevitable when placing three agents on the two squares, the optimal strategy is to move two agents to the right square and one agent to the left square. The Amalgam Policy and Cutoff Policy assume a 1 agent system for all the agents when they are out of view of each other and greedily send the middle agent towards the right square which happens to be the fully observable joint optimal action. On the other hand, the Simple Memory Based Extraction will send the middle agent towards the left because it is able to coordinate with the agent on the right and assumes a 2 agent system. However, there is an unexpected agent on the left which will result in a lower performance than the greedy Trivial Extraction policies.

\subsection{Coordination With Out Of View Agents}

\begin{figure}[H]
\centering

\includegraphics[width=0.70\linewidth]{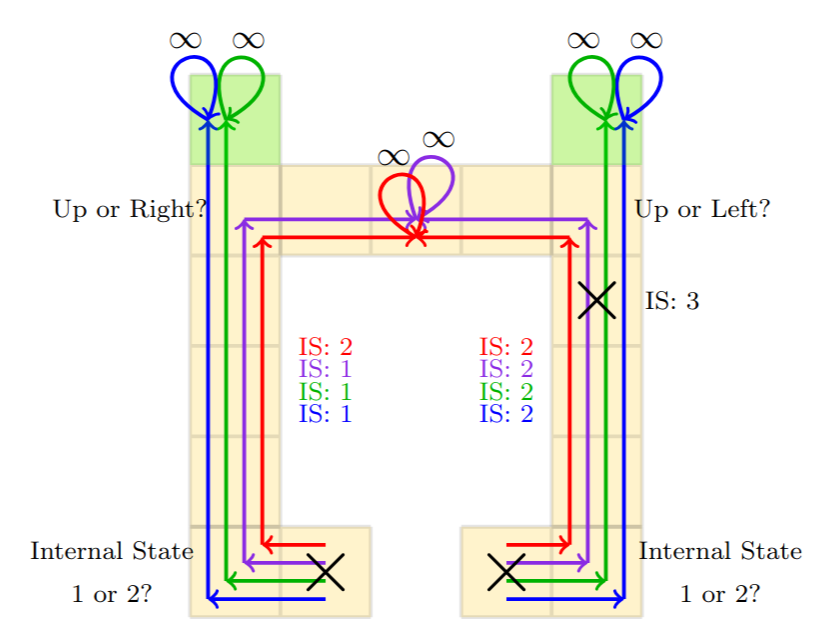}

\caption{
(Out of View Coordination) In red is the trajectory for the fully observable joint optimal policy with discounted sum of rewards 95.66. In blue and green is the Amalgam Policy and Cutoff Policy respectively with value of 2.70. In purple is Simple Memory Based Extraction ($\mathcal V_{comp} = 4$) with value 0. All values are rounded to the second decimal place.
}
\end{figure}

The purpose of this simulation is to show another adversarially chosen case for Simple Memory Based Extraction which is deterministic but not all of the agents are in view at the beginning of the simulation. Specifically, an agent will interact with an agent out of view and will make agent estimations inaccurate. This again violates one of the conditions that ensures fully observable joint optimal behavior for Simple Memory Based Extraction. In this example as well, Trivial Extraction will outperform Simple Memory Based Extraction even for large computational visibilities.
 
For this simulation, similar to \cref{stochastic}, two agents will split apart and come together and try to match their internal states. However, instead of stochasticity, there is a third agent that will force one of the agents to choose a specific internal state. 
 
Specifically, two agents will initially start within the view of each other. Once out of view, agents can decide whether to change their internal state to 1 or 2, with 1 leading to a higher reward if they can match internal states and overlap $\mathcal R = 0$. However, the right agent can see a stationary third agent with $\mathcal V_{exec} = 3$ not previously known which is in internal state 3. When an agent with internal state 1 and internal state 3 overlaps, both agents incur a penalty of $-500$. Therefore, the agent on the right must choose internal state 2 to avoid collision. Once at the top, agents can decide whether to move to a space that gives a recurring $+1$ reward or whether to go to the center. At the upper center square, if both overlap with a matching internal state 1, they will obtain a recurring $+50$ reward or if both overlap with a matching internal state 2, they will obtain a recurring $+10$ reward and otherwise obtain no reward. The discount factor is $\gamma = 0.9$.
 
We see that the fully observable joint optimal solution is to match the internal states 2 at the upper center square. The Amalgam Policy and Cutoff Policy, once the two non-stationary agents are out of view, will go towards the reward squares on the upper left and upper right forgetting about each other. In Simple Memory Based Extraction the left agent, which does not see the stationary agent, will believe the the right agent will not incur any obstacles and will coordinate in moving to the upper center square with internal states 1. The right agent will see the stationary agent and will change their internal state to 2 without being able to communicate with the agent on the left. This results in a mistaken estimation by the left agent of the right agent, and they both go towards the upper center square with mismatching internal states resulting in no reward.

\section{Simple Memory Based Extraction}
\label{simple_mem}

{
\footnotesize
\IncMargin{2em}
\begin{algorithm2e}[H]
\caption{Simple Memory Based Extraction}
\label{simple_algo}
\small
\label{algo:ValueIter}
\KwIn {$I(\tau^t, z)$}
        \vspace{1ex}
        
    $M_{i,j}(t) \leftarrow None$ for $i,j \in \mathcal N$\\
    
    \tcp{Retrieve previous memory / agent position estimates}
\uIf{$t > 0$}{

    \label{recursion}
    \For{$z' \in Z(s(t - 1))$ such that $z'\cap z \neq \varnothing$}{
        Let $\{M_{i,j}(t-1)$: $i \in z'\} \leftarrow$ \cref{simple_algo} $\big(I(\tau^{t - 1}, z')\big)$\\
        Let $M_{i,j}(t) \leftarrow M_{i,j}(t - 1)$ for $i \in z'$%
    }
     }
        
        Let $M_{i,i}(t)\leftarrow (s_i(t), t)$, and $M_{z,i}\leftarrow (s_i(t), t)$ for all $i \in z$\\
        \label{update_curr_groups}
        \vspace{1ex}
        
        \tcp{Obtain most recent observation for out of view agents}
        Let $M_{z,j}\leftarrow \text{argmax}_{M_{i,j}(t), i \in z} M_{i,j}^{time}(t)$ for $j \notin z$ 
        \label{consolidation}%
        
    \vspace{1ex}
    \tcp{Construct belief state}
     Let $s^{belief} \leftarrow (M_{z,j}^{state}: j \in \mathcal N, M_{z,j} \neq  None )$\label{belief_state}%
    
    Let $z^{belief} \leftarrow z'$ where $z' \in Z_{comp}(s^{belief})$ such that $z \subset z'$\label{belief_state_next}\\
    
    \vspace{1ex}

    \vspace{1ex}
    \tcp{Predict next state for belief agents}
    $a^{belief} \leftarrow \pi^{\xi,\eta}_{0}(\cdot \lvert ([s^{belief}]_{z^{belief}}, \{z^{belief}\}))$\\
    
    $s^{belief}_{next} = \text{argmax}_{s'} P(s'\lvert [s^{belief}]_{z^{belief}}, a^{belief})$\label{most_likely}\\
    \vspace{1ex}
    
    \tcp{Update memory for all agents in group}
    $M_{i,j}(t) \leftarrow ([s_{next}^{belief}]_j, M_{z,j}^{time})$ for $i \in z$ and $j \in z^{belief}$\label{before_last_update}\\
    $M_{i,j}(t) \leftarrow None$ for $i \in z$ and $j \not\in z^{belief}$\label{last_update}\\
    \vspace{1ex}
\KwRet{$([s^{belief}]_{z^{belief}}, \{z^{belief}\})$, $\{M_{i,j}(t)$: $i \in z$, $j \in \mathcal N\}$}
    
\end{algorithm2e}
}
\vspace{2ex}
The complete algorithm for Simple Memory Based Extraction is provided in \cref{simple_algo}. Intuitively, for every agent, the algorithm retains estimates of other agents that have been seen. When an agent joins a group, these estimates are shared and consolidated with the group based on the most recent observations. Then a consolidated ``belief state'' is constructed and the next state of this constructed state is predicted using the most likely outcome. This estimate is then used to update the predictions for each agent of the location of agents outside of the group.

Formally, for each agent, the algorithm will maintain memory variables $M_{i,j}(t)$ that contain position estimates for agents outside of view. 
In \cref{recursion}, the algorithm begins by retrieving the memory from the previous time step if it exists and in \cref{update_curr_groups} it will be updated with the current observation of the group. Temporary memory variables $M_{z,j}$ are created which represent ``group memory'' and consolidate the memory among agents in the group in \cref{consolidation} based on estimates made by the most recent observation of each agent. Here, the notation $M_{i,j}^{state}(t)$ will denote the first component of $M_{i,j}$ and $M_{i,j}^{time}(t)$ the second component. Then in \cref{belief_state} and \cref{belief_state_next}, a belief state $s^{belief}$ is constructed and in that belief state the visibility group $z^{belief}$ (with visibility $\mathcal V_{comp}$) that includes $z$ is retrieved. The next predicted state for this belief state visibility group is constructed in \cref{most_likely} and finally the memory is updated for all agents in $z$ in \cref{before_last_update} and \cref{last_update}.

This extraction method is referred to as ``simple'' because there are obvious improvements to this algorithm that can be made, such as removing estimates for agents that have low confidence, adding a more sophisticated consolidation of memories within the group, or more intricate transition updates to the agent estimates (rather than the most likely outcome in \cref{most_likely}). Note that these variations would still be valid extraction methods and would be part of the Extended Cutoff Policy Class.

\section{Generalized Locally Interdependent Multi-Agent MDP}
\label{generalized}
In the generalized setting, an agent $i$ lying in partition $g \in D(s)$ will obtain a reward $r_i(s_g, a_g) = \sum_{j \in g} \overline r_{i,j}(s_i, a_i, s_j, a_j)$ without restrictions on the pairwise reward function $\overline{r}_{i,j}$. Transition dependence is introduced by allowing agents to move according to arbitrary local interdependent transitions $P_i(s'_i \lvert s_g, a_g)$. We still assume that if $d(s'_i, s_i) > 1$ and agent $i$ is in partition $g \in D(s)$, that $P_i(s'_i \lvert s_g, a_g) = 0$ which implies $P_g(s'_g \lvert s_g, a_g) = 0$ where $P_g(s'_g \lvert s_g,a_g) = \prod_{i \in g} P_i(s'_i \lvert s_g,a_g)$. 
 
We now present the definition of the Generalized Locally Interdependent Multi-Agent MDP as follows:

\begin{definition}
(Generalized Locally Interdependent Multi-Agent MDP) Assume a set of agents $\mathcal N$ and a set of states in a common metric space $(\mathcal X, d)$. Furthermore, for $i \in \mathcal N$ assume a set of internal states $\mathcal Y_i$, a set of actions $\mathcal A_i$, a transition function $P_i$ depending on group states such that $P_i(s'_i \lvert s_g,a_g) = 0$ if $d(s_i, s'_i) > 1$, and a reward function $r_{i,j}$. The Locally Interdependent Multi-Agent MDP $\mathcal{M} = (\mathcal{S}, \mathcal A, P, r, \gamma, \mathcal R)$ is defined as follows:
\begin{itemize}
\item $\mathcal S := \times_{i \in \mathcal N}\mathcal S_i$ where $\mathcal S_i = (\mathcal X, \mathcal Y_i)$
\item $\mathcal A := \times_{i \in \mathcal N}\mathcal A_i$ 
\item $ P(s' \lvert s, a) = \prod_{g \in D(s)} P_g(s'_g \lvert s_g,a_g)$
\item $ r(s, a) = \sum_{g \in D(s)} r_g(s_g,a_g)$
\end{itemize}
where $r_g(s_g, a_g) = \sum_{i \in g} r_i(s_g, a_g)$ and $r_i(s_g, a_g) = \sum_{j \in g} \overline r_{i,j}(s_i, a_i, s_j, a_j)$%
\end{definition}
 
Similar to the standard definition, $D(s)$ is finer than $Z(s)$ for visibility $\mathcal V > \mathcal R$ so we may write an equivalent reward and transition definition $P(s'\lvert s,a) = \prod_{z \in Z(s)} P_z(s'_z\lvert s_z, a_z)$ where $P_z(s'_z\lvert s_z, a_z) = \prod_{g \in D(s_z)} P_g(s'_g\lvert s_g, a_g)$ and $ r(s, a) = \sum_{z \in Z(s)} r_z(s_z,a_z)$ where $ r_z(s_z, a_z) = \sum_{g \in D(s_z)} r_g(s_g,a_g)$.
 
In the Cutoff Multi-Agent MDP defined in \cref{cutoff}, we replace the definition of the transition with $ P^\mathcal C((s', Z(s') \cap \mathcal P) \lvert (s, \mathcal P), a) = \prod_{p \in \mathcal P} P(s'_p \lvert s_p,a_p)$ where $P(s'_p \lvert s_p,a_p) = \prod_{g \in D(s_p)}P(s'_{g} \lvert s_{g},a_{g})$ and the definition of the reward $r^\mathcal C((s, \mathcal P), a) = \sum_{p \in \mathcal P} r_p(s_p, a_p)$ uses the new less restricted quantities $\overline r_{i,j}$ defined above. All else remain the same.
 
Other definitions such as group decentralized policies, Extended Cutoff Policy Class, Consistent Performance Policy etc. are defined the same as in the Standard Locally Interdependent Multi-Agent MDP.

\section{Consistent Performance Policies}
\label{monotone}
We denote a sample trajectory in the Cutoff Multi-Agent MDP starting at state $(s, \mathcal P)$ with a non-stationary policy $\pi = (\pi_h)_{h = 0}^{c + \eta}$ to be $\pi\lvert_{s,\mathcal P}$ which is a sequence of elements $((s(t), \mathcal P(t)), a(t))$ for $t \in \{0,1, ..., c + \eta\}$ where $s(0) = s$ and $\mathcal P(0) = \mathcal P$ and $c = \lfloor \frac{\mathcal V - \mathcal R}{2}\rfloor$. This trajectory is sampled according to $a(t) \sim \pi_t(\cdot \lvert (s(t), \mathcal P \cap \bigcap_{t' = 0}^t Z(s(t')))$ %
and $(s(t + 1), \mathcal P(t + 1)) \sim P^{\mathcal C}(\cdot \lvert (s(t), \mathcal P \cap \bigcap_{t' = 0}^t Z(s(t'))), a(t))$. 
Recall from our transition definition in \cref{cutoff} that $\mathcal P(t) = \mathcal P \cap \bigcap_{t' = 0}^t Z(s(t'))$. Going forward, for any $\mathcal P'$ which is either finer or coarser than $\mathcal P$, we will define a trajectory $\pi\lvert_{s,\mathcal P}^{\mathcal P'}$ to be the trajectory $((s'(t), \mathcal P'(t)), a'(t))$ where the policy acts according to $a'(t) \sim \pi_t(\cdot \lvert (s'(t), \mathcal P' \cap \bigcap_{t' = 0}^t Z(s'(t')))$ and the transition the same as above $(s'(t + 1), \mathcal P'(t + 1)) \sim P^{\mathcal C}(\cdot \lvert (s'(t), \mathcal P \cap \bigcap_{t' = 0}^t Z(s'(t'))), a'(t))$ where $s(0) = s$ and $\mathcal P(0) = \mathcal P$. %
Intuitively, the policy acts as if the initial partition was $\mathcal P'$ rather than $\mathcal P$ even though the actual initial partition is $\mathcal P$. In the case where $\mathcal P'$ is coarser, the policy ignores the observability constraint of the Cutoff Multi-Agent MDP and coordinates groups across partitions. On the other hand, if $\mathcal P'$ is finer, this will act as if there are disconnections between agents, even though they are connected.
 
We will also use the notation $\tau_h^{h'} \sim \pi\lvert_{s, \mathcal P}$ to denote the trajectory $((s''(t), \mathcal P''(t)), a''(t))$ for $t \in \{h, h + 1, ..., h'\}$ starting at $s''(h) = s$ and $\mathcal P''(h) = \mathcal P$ while taking policy $\pi_t$ at time step $t$. This is defined similarly for $\tau_h^{h'} \sim\pi\lvert_{s,\mathcal P}^{\mathcal P'}$.
 
Here, the intersection of partitions for any subsets $A, B \subset\mathcal N$, for a partition $\mathcal P_A$ over $A$ and a partition $\mathcal P_B$ over $B$, $\mathcal P_A \cap \mathcal P_B$ will indicate a partition over $A\cap B$ where $\mathcal P_A \cap \mathcal P_B  = \{p_A \cap p_B \lvert p_A \in \mathcal P_A, p_B \in \mathcal P_B\} \setminus \{\varnothing\}$.
 
Then we define a variation on the value function for the Cutoff Multi-Agent MDP. For some $g\subset \mathcal N$%

\begin{center}
$[V^{\pi}_{h, h'}]_g (s,\mathcal P) = \mathbb{E}_{\tau_h^{h'} \sim \pi\lvert_{s,\mathcal P}} \bigg[\sum_{t = h}^{h'} \gamma^{t - h} r^{\mathcal C}_g((s(t), \mathcal P(t)), a(t))\bigg]$ 
\end{center}
\begin{center}
$[V^{\pi\lvert^{P'}}_{h, h'}]_g(s, \mathcal P) = \mathbb{E}_{\tau_{h}^{h'} \sim \pi\lvert_{s,\mathcal P}^{\mathcal P'}} \bigg[\sum_{t = h}^{h'} \gamma^{t-h} r^{\mathcal C}_g((s(t), \mathcal P(t)), a(t))\bigg]$ \\
\end{center}

where 
$r_g^\mathcal C ((s(t), \mathcal P(t)),a(t)) = \sum_{p' \in \mathcal P(t) \cap \{g\}}r_{p'}(s_{p'}(t),a_{p'}(t))$. %
If $h'$ is omitted, it will be assumed that $h'$ is the maximum possible horizon length depending on the context. If $g$ is omitted, we assume $g = \mathcal N$.
 
Suppose we take a non-stationary Proper Cutoff Policy $\pi = (\pi_h)_{h = 0}^{c + \eta}$ and append a stationary Proper Cutoff Policy $\pi_{c + \eta + 1}$ to create a new policy $\pi' = (\pi_h')_{h = 0}^{c + \eta + 1}$. We say $\pi$ is a Consistent Performance Policy if we can find a policy $\pi_{c + \eta + 1}$ such that $\pi$ and $\pi'$ satisfy the following for all $\mathcal P'$ finer than $\mathcal P$:
\begin{itemize}
    \item  (Constructive Improvement) $[V_h^{\pi\lvert^{\mathcal P'}}]_p (s, 
    \mathcal P)\leq [V_h^\pi]_p (s, \mathcal P)$ for $h \in \{0,1\}$ and $p \in \mathcal P$
    \item (Deconstructive Improvement) $[V_h^{\pi\lvert^\mathcal P}]_{p'}(s, \mathcal P') \leq [V_h^\pi]_{p'}(s, \mathcal P')$ for $h \in \{0,1\}$ and $p' \in \mathcal P'$
    \item (1-Step Displaced Improvement) $[V_{1,c + \eta + 1}^{\pi'}]_p(s, \mathcal P) \leq [V_0^{\pi}]_p(s, \mathcal P)$ for $p \in \mathcal P$
    \item (1-Step Contracted Improvement) $[V_{0, c + \eta - 1}^{\pi}]_p(s, \mathcal P) \leq [V_1^{\pi}]_p(s, \mathcal P)$ for $p \in \mathcal P$
\end{itemize}
Notice that these conditions are trivially satisfied by the optimal finite horizon policies $\pi^{\xi,\eta}_{comp}$. 
 
Intuitively, these conditions describe a policy that is effective with respect to its intention and ability. For example, in the case of Constructive and Deconstructive Improvement, the policy on the left-hand side acts as though the visibility partition is finer or coarser than it actually is. We assert that when the policy acts in a congruent way with the environment on the right-hand side, the performance is better. A similar intuition carries over to the other two conditions, but with respect to the horizon. The name ``Consistent Performance Policy'' comes from the formal statement in \cref{z_intersect}. Using $\delta = 0$, the lemma shows that if the partition is further split by the visibility groups, the performance will be similar, give or take a $o(\gamma^c)$ term.

\section{Proofs}

\subsection{Standard Locally Interdependent Multi-Agent MDP}
\subsubsection{Main Proof}
\label{indep_proofs}

Similar to the proofs in \cite{deweese2024locally}, we will rely on a key insight called the Dependence Time Lemma. Roughly, this lemma formalizes the intuition that, since $\mathcal V > \mathcal R$, agents in different visibility partitions $Z(s)$ will have a time and distance buffer before they influence each others reward based on the dependence partition $D(s)$.
 
For the following, we define a local movement trajectory $(s(t), a(t))$ for $t \in \{0, \ldots, \mathcal T\}$ to be any sequence of states and actions such that $d(s_i(t + 1), s_i(t)) \leq 1$ for $t \in \{0,\ldots, \mathcal T - 1\}$ and for all $i\in \mathcal N$.
\begin{lemma}
\label{dtl}
(Dependence Time Lemma for Rewards) For any local movement trajectory $(s(t), a(t))$ for $t \in \{0, \ldots, \mathcal T\}$ in a Locally Interdependent Multi-Agent MDP and a time step $T \in \{0, \ldots, \mathcal T - \delta\}$, $D(s(T + \delta))$ is a finer partition than $Z(s(T))$ and
$$ r(s(T + \delta), a(T + \delta)) = \sum_{z \in Z(s(T))} r_z(s_z(T + \delta),a_z(T + \delta))$$ for $\delta \in \{0, \ldots, c\}$ where $c = \lfloor \frac{\mathcal V - \mathcal R}{2}\rfloor$.
\end{lemma}
\begin{proof}
Suppose we have two agents $j,k \in \mathcal N$ in different partitions of $Z(s(T))$. This implies $d(s_j(T), s_k(T)) > \mathcal V$. Then in $\delta$ timesteps, $d(s_j(T + \delta), s_j(T + \delta)) \leq \sum_{t = 0}^{\delta - 1}d(s_j(T + t), s_j(T + t + 1)) \leq \delta$ and similarly for agent $k$. Therefore,
\begin{align*}
    d(s_j(T + \delta), s_k(T + \delta)) &\geq d(s_j(T), s_k(T)) - d(s_j(T + \delta), s_j(T)) - d(s_k(T + \delta), s_k(T))\\
    &> \mathcal V - 2 \delta
    \geq \mathcal V - 2c 
    \geq \mathcal V - 2\bigg(\frac{\mathcal V - \mathcal R}{2}\bigg)
    = \mathcal R
\end{align*}
Since $d(s_j(T), s_k(T)) > \mathcal V$ implies $d(s_j(T + \delta), s_k(T + \delta)) > \mathcal R$, we can say $D(s(T + \delta))$ is a finer partition than $Z(s(T))$. And so $r(s(T + \delta),a(T + \delta)) = \sum_{g \in D(s(T+ \delta))} r_g(s_g(T + \delta), a_g(T + \delta)) = \sum_{z \in Z(s(T))} r_z(s_z(T + \delta), a_z(T + \delta))$.
    
\end{proof}
There is also a Dependence Time Lemma for Transitions in \cref{dtl_transition} that is used for General Locally Interdependent Multi-Agent MDPs. 
The Dependence Time Lemma is crucial in proving various lemmas which will be used in the proof below and be deferred to \cref{indep_lemmas}.

The following is the proof for \cref{extract_opt_bound} and \cref{extract_bound} for Standard Locally Interdependent Multi-Agent MDPs. The proof for General Locally Interdependent Multi-Agent MDPs is provided in \cref{dependent_proofs}. Here, $c = \lfloor \frac{\mathcal V_{exec} - \mathcal R}{2}\rfloor$ and $c' = \lfloor \frac{\mathcal V_{comp} - \mathcal R}{2}\rfloor$ and recall $\mathcal V_{comp} \geq \mathcal V_{exec}$ which implies $c' \geq c$.
\\
\begin{proof}\textbf{For \cref{extract_opt_bound} (Standard)}
\label{extract_opt_bound_proof}
\begin{align}
V^*(s) -& V^{\pi^{\xi, \eta}_{exec}}(s)\leq V^*(s) - V^{\pi^{\xi, \eta}}_0(s,Z_{comp}(s)) \label{using_extract}\\
&\hspace{10ex}+ \bigg(\gamma^2 + \gamma^{\eta + 1} + \frac{4 + \gamma + 5\gamma^2}{1 - \gamma}\bigg) \frac{\gamma^{c - 1}}{1 - \gamma}\nonumber\\
&\leq V^*(s) - V^{\pi^*}_0(s,Z_{comp}(s))\label{using_opt_sub}\\
&\hspace{10ex}+ \bigg(\gamma^2 + \gamma^{\eta + 1} + \frac{4 + \gamma + 5\gamma^2}{1 - \gamma}\bigg) \frac{\gamma^{c - 1}}{1 - \gamma}\nonumber\\
&\leq  \bigg( \gamma^{c' - c + 1} + \gamma^2 + \gamma^{\eta + 1} + \frac{4 + \gamma + 5\gamma^2}{1 - \gamma}\bigg) \frac{\gamma^{c - 1}}{1 - \gamma}\label{using_same}
\end{align}
In \cref{using_extract}, \cref{extract_bound} is used. In \cref{using_opt_sub}, since $\pi^{\xi, \eta}$ is the optimal extended Cutoff Multi-Agent MDP solution with horizon $c + \eta$ and visibility $\mathcal V_{comp} = \mathcal V_{exec} + \xi$ %
, any policy will have a lower value. Here we substitute the potentially non-proper policy $\pi^*(\cdot \lvert s, \mathcal P) = \pi^*(\cdot \lvert s)$ and evoke \cref{same_policy} in \cref{using_same}. 
    
\end{proof}
\begin{proof}\textbf{For \cref{extract_bound}  (Standard)}
\label{extract_bound_proof}
Intuitively, this proof revolves around expressing the expected rewards in the value function of $V^{\pi_{exec}}$ using the difference of other well understood value functions. Concretely, our proof uses the following simple fact that for any two sequences $(v_t)$ and $(v_t')$ with bounded magnitude and a sequence $\hat r_t = v_t - \gamma v_{t + 1}'$, we have 
\begin{align}
\sum_{t = 0}^{\infty} \gamma^t \hat r_t = v_0 - \sum_{t = 1}^\infty  \gamma^t(v_{t}' - v_t)\label{telescoping}
\end{align}
for any discount factor $\gamma \in [0,1)$. Then, to bound the difference $\lvert\sum_{t = 0}^{\infty} \gamma^t \hat r_t - v_0\rvert$ it would suffice to bound $\Delta = \lvert v'_{t} - v_t\rvert$. 

In light of the above simple fact, our proof will then consist of three steps. 1) rewrite $V^{\pi_{exec}}$ and apply the decomposition above 2) bound the $\Delta$ difference terms 3) perform additional manipulation to match the objective.

We use the notation $\tau \sim \pi_{exec}\lvert_s$ to denote a trajectory $\tau = (s(t), a(t))$ for $t \in \{0,1, \ldots\}$ that is sampled using policy $\pi_{exec}$ starting at state $s$.
\\\\
\noindent\textbf{Step 1:} We may rewrite the expression for $V^{\pi_{exec}}$ to obtain,
{\small
\begin{align}
    V^{\pi_{exec}}(s) =& \mathbb{E}_{\tau \sim \pi_{exec}\lvert_s}\bigg[\sum_{t = 0}^\infty \gamma^t r(s(t), a(t))\bigg]\nonumber = \sum_{t = 0}^\infty \gamma^t \mathbb{E}_{\tau \sim \pi_{exec}\lvert_s}\bigg[r(s(t), a(t))\bigg]\nonumber\\
     = &\sum_{t = 0}^\infty \gamma^t \bigg(\mathbb{E}_{\tau \sim \pi_{exec}\lvert_s}\bigg[V^{\pi_{comp}}_0(s(t), Z_{exec}(s(t))) \bigg] \nonumber\\
     &\hspace{3.5ex}- \gamma \mathbb{E}_{\tau \sim \pi_{exec}\lvert_s}\bigg[V^{\pi_{comp}}_0(s(t + 1), Z_{comp}(s(t + 1)\cap Z_{exec}(s(t)))  \bigg]\bigg)\label{using_condition}\\
    &\hspace{3.5ex}+ \zeta\bigg(4\frac{(1 + \gamma^2)}{(1 - \gamma)^2}\gamma^{c - 1} \tilde r\bigg) + \zeta\bigg(\frac{\gamma^{c + \eta}}{1 - \gamma} \tilde r\bigg)\nonumber
\end{align}
}
\\
where $\zeta(b)$ denotes some value in $[-b, b]$ for $b > 0$. In \cref{using_condition} we apply \cref{tele_condition} to express the reward in terms of the value functions for the Cutoff Multi-Agent MDP setting with visibility $\mathcal V_{comp} = \mathcal V_{exec} + \xi$. Applying the decomposition from \cref{telescoping} %
, we obtain, %
$$V^{\pi_{exec}}(s) = V_0^{\pi_{comp}}(s, Z_{exec}(s)) - \mathbb{E}_{\tau \sim \pi_{exec}\lvert_s}\bigg[\sum_{t = 1}^\infty \gamma^t \Delta_t^\tau\bigg] + \zeta(4\frac{(1 + \gamma^2)}{(1 - \gamma)^2}\gamma^{c - 1} \tilde r) + \zeta(\frac{\gamma^{c + \eta}}{1 - \gamma} \tilde r)$$
where $
\Delta_t^\tau = V_0^{\pi_{comp}}(s(t), Z_{comp}(s(t))\cap Z_{exec}(s(t - 1))) - V_0^{\pi_{comp}}(s(t), Z_{exec}(s(t)))$. %

\vspace{2ex}
\noindent\textbf{Step 2:} To proceed with the motivation described in our initial discussion, we are interested in bounding $\Delta_t^\tau$. We use \cref{z_intersect} which bounds the difference of values for Consistent Performance Policies when the partition is intersected by the visibility partition. This can be used bound $\Delta_t^\tau$ as follows
\begin{align}
\lvert\Delta_t^\tau \rvert&= \lvert V_0^{\pi_{comp}}(s(t), Z_{comp}(s(t))\cap Z_{exec}(s(t - 1))) - V_0^{\pi_{comp}}(s(t), Z_{exec}(s(t)))\rvert\nonumber\\
&\leq  \lvert V_0^{\pi_{comp}}(s(t), Z_{comp}(s(t))\cap Z_{exec}(s(t - 1))) - V_0^{\pi_{comp}}(s(t), Z_{exec}(s(t))\cap Z_{exec}(s(t - 1))) \rvert \nonumber\\
&\hspace{15ex}+ \lvert V_0^{\pi_{comp}}(s(t), Z_{exec}(s(t))\cap Z_{exec}(s(t - 1)))- V_0^{\pi_{comp}}(s(t), Z_{exec}(s(t)))\rvert\nonumber\\
&\leq  \frac{\gamma^{c + 1}}{1 - \gamma}\tilde r + \frac{\gamma^c}{1 - \gamma}\tilde r = \frac{(1 + \gamma)\gamma^c}{1 - \gamma}\tilde r\label{using_z_intersect}
\end{align}\\
where in \cref{using_z_intersect} we apply \cref{z_intersect} twice. In the application to the first term we can use
$Z_{exec}(s(t)) \cap Z_{exec}(s(t - 1)) = (Z_{exec}(s(t)))\cap (Z_{comp}(s(t)) \cap Z_{exec}(s(t - 1)))$ since $Z_{exec}$ is finer than $Z_{comp}$, and then apply \cref{z_intersect}.
 
We now obtain a bound on the difference of the value functions\\
$$\lvert V^{\pi_{exec}}(s) - V_0^{\pi_{comp}}(s, Z_{exec}(s)) \rvert \leq \frac{(1 + \gamma)\gamma^c}{(1 - \gamma)^2}\tilde r +4\frac{(1 + \gamma^2)}{(1 - \gamma)^2}\gamma^{c - 1} \tilde r + \frac{\gamma^{c + \eta}}{1 - \gamma} \tilde r.$$
\\
\textbf{Step 3:} Finally, we will replace $Z_{exec}(s)$ with $Z_{comp}(s)$ and adjust the bound as follows
\begin{align*}
\lvert V^{\pi_{exec}}(s) - V_0^{\pi_{comp}}(s, Z_{comp}(s)) \rvert &\leq \lvert V^{\pi_{exec}}(s) - V_0^{\pi_{comp}}(s, Z_{exec}(s)) \rvert \\
&\hspace{10ex}+ \lvert V_0^{\pi_{comp}}(s, Z_{exec}(s)) - V_0^{\pi_{comp}}(s, Z_{comp}(s)) \rvert \\
&  \leq \frac{\gamma^{c + 1}}{1 - \gamma}\tilde r + \frac{(1 + \gamma)\gamma^c}{(1 - \gamma)^2}\tilde r +4\frac{(1 + \gamma^2)}{(1 - \gamma)^2}\gamma^{c - 1} \tilde r + \frac{\gamma^{c + \eta}}{1 - \gamma} \tilde r\\
&= \bigg(\gamma^2 + \gamma^{\eta + 1} + \frac{4 + \gamma + 5\gamma^2}{1 - \gamma}\bigg) \frac{\gamma^{c - 1}}{1 - \gamma}
\end{align*}
once again by applying \cref{z_intersect} and using $Z_{exec}(s) = Z_{comp}(s) \cap Z_{exec}(s)$.
\end{proof}

\subsubsection{Lemmas}
\label{indep_lemmas}
In this section, we will prove the lemmas used in the proof of \cref{extract_bound} assuming transition independence shown in \cref{extract_bound_proof}.
 
For the remainder of the proofs in this section, assume a Locally Interdependent Multi-Agent MDP $\mathcal{M} = (\mathcal{S}, \mathcal A, P, r, \gamma, \mathcal R)$, with some Consistent Performance Policy $\pi_{comp} = (\pi_0, \pi_1 , ..., \pi_{c + \eta})$ in its corresponding extended Cutoff Multi-Agent MDP $\mathcal{C} = (\mathcal{S}^\mathcal C, \mathcal A^\mathcal C, P^\mathcal C, r^\mathcal C, \gamma, \mathcal R, \mathcal V + \xi)$. We will also assume a corresponding extracted policy $\pi_{exec}(a \lvert s(t)) = \prod_{z\in Z_{exec}(s(t))}[\pi_{0}]_z(a_z \lvert \rho(I(\tau^t,z))$ according to a valid extraction strategy $\rho(I(\tau^t,z))$ for $z \in Z_{exec}(s)$, $s \in \mathcal S$, and $t \in \{0, 1, \dots\}$.
 
Notationally, for any partition $\mathcal P$ over $\mathcal N$, we will indicate $E(\mathcal P) = \{(j,k): j,k \in p \text{ for some p }\in \mathcal P\}$ and $E^c(\mathcal P) = (\mathcal N \times \mathcal N )\setminus E(\mathcal P)$.
 
Recalling our initial definition $\tilde r = \sum_{i,j\in \mathcal N} \max_{s_i, a_i, s_j, a_j}\lvert \overline r_{i,j}(s_i, a_i, s_j, a_j)\rvert $, we will also define, for any $g \subset \mathcal N$, the following quantity $\tilde r_g = \sum_{i,j\in g} \max_{s_i, a_i, s_j, a_j}\lvert \overline r_{i,j}(s_i, a_i, s_j, a_j)\rvert $ and for any partition $\mathcal P$, $\tilde r_{\mathcal P} = \sum_{i,j\in E(\mathcal P)} \max_{s_i, a_i, s_j, a_j}\overline \lvert r_{i,j}(s_i, a_i, s_j, a_j)\rvert $ and $\tilde r_{\mathcal P}^c = \sum_{i,j\in E^c(\mathcal P)} \max_{s_i, a_i, s_j, a_j}\overline \lvert r_{i,j}(s_i, a_i, s_j, a_j)\rvert$.

We also note the following properties of the value function in the Cutoff Multi-Agent MDP.

\begin{proposition}
For any partition $\mathcal P$ finer than $Z_{comp}(s)$, $p = \bigcup_{i = 0}^{\ell} p_i$ where $p_i \in \mathcal P$, and all $h \in \{0, \ldots, c + \eta\}$,\label{extra_properties}
\begin{itemize}
    \item $[V^{\pi}_h]_{p}(s, \mathcal P) = [V^{\pi}_h]_{p}(s', \mathcal P)$ where $s' = (s_p, s'_{\mathcal N \setminus p})$ and $s'_{\mathcal N \setminus p}$ is some state over the agents in $\mathcal N \setminus p$
    \item $[V^{\pi}_h]_p(s, \mathcal P) = [V^{\pi}_h]_p(s, \mathcal P')$ for any partition $\mathcal P'$ finer than $Z_{comp}(s)$ with $\mathcal P' \cap \{p\} = \mathcal P \cap \{p\}$
    \item $[V^{\pi}_h]_p(s, \mathcal P) = [V^{\pi}_h]_p(s_{\mathcal N_{small}}, \mathcal P\cap \{\mathcal N_{small}\})$ where $\mathcal N_{small} \subset \mathcal N$ and $p \subset \mathcal N_{small}$ 
\end{itemize}
\end{proposition}
For the following, recall our definition of a local movement trajectory defined previously. A local movement trajectory $(s(t), a(t))$ for $t \in \{0, \ldots, \mathcal T\}$ is any sequence of states and actions such that $d(s_i(t + 1), s_i(t)) \leq 1$ for $t \in \{0,\ldots, \mathcal T - 1\}$ and for all $i\in \mathcal N$. In the case of a trajectory in the Cutoff Multi-Agent MDP $((s(t), \mathcal P(t)), a(t))$ for $t \in \{0, \ldots, \mathcal T\}$, we also say it is a local movement trajectory if $d(s_i(t + 1), s_i(t)) \leq 1$ for $t \in \{0,\ldots, \mathcal T - 1\}$.
\begin{lemma}
\label{z_intersect}
Assume a local movement trajectory $(s(t), a(t))$ for $t \in \{0, \ldots, \mathcal T\}$ and $\delta \in \{0,\ldots, c\}$, and  some $T$ such that $T, T- \delta \in \{0, \ldots, \mathcal T\}$. For any $\mathcal P$ finer than $Z_{comp}(s(T))$ and a partition $\mathcal P'$ coarser than $Z_{exec}(s(T - \delta))$, for $h \in \{0,1\}$ we have the following:
\begin{align*}
&\bigg\lvert V^{\pi_{comp}}_{h}(s(T), \mathcal P \cap \mathcal P') -V^{\pi_{comp}}_{h}(s(T),\mathcal P) \bigg\rvert \leq \frac{\gamma^{c - \delta -h + 1}}{1 - \gamma} \tilde r_{\mathcal P'}^c\\
&\bigg\lvert V^{\pi_{comp}\lvert^{\mathcal P}}_{h}(s(T),\mathcal P \cap\mathcal P') -V^{\pi_{comp}}_{h}(s(T),\mathcal P) \bigg\rvert \leq \frac{\gamma^{c - \delta -h + 1}}{1 - \gamma} \tilde r_{\mathcal P'}^c\\
&\bigg\lvert V^{\pi_{comp}}_{h}(s(T), \mathcal P \cap\mathcal P') -V^{\pi_{comp}\lvert^{\mathcal P\cap\mathcal P'}}_{h}(s(T),\mathcal P) \bigg\rvert \leq \frac{\gamma^{c - \delta -h + 1}}{1 - \gamma} \tilde r_{\mathcal P'}^c\\
\end{align*}
\end{lemma}
\begin{proof}
Notice that for a trajectory $((s'(t'), \mathcal P'(t')), a'(t'))$ where $\tau'_h\sim \pi_{comp}\lvert_{s(T), \mathcal P}$ in the Cutoff Multi-agent MDP, if we create a new trajectory $(s(T - \delta), a(T - \delta)), \ldots, (s(T-1), a(T-1)), (s'(h), a'(h)) \ldots, (s'(c + \eta), a'(c + \eta))$ this trajectory will also be a local movement trajectory by definition and we may evoke the Dependence Time Lemma from \cref{dtl}. Therefore, 
 
\begin{align} 
 &V^{\pi_{comp}}_h (s(T), \mathcal P) \nonumber\\
 &= \mathbb{E}_{\tau'_h \sim \pi_{comp}\lvert_{s(T), \mathcal P}}\bigg[\sum_{t' = h}^{c + \eta} \gamma^{t' - h} r^{\mathcal C}((s'(t'), \mathcal P'(t')), a'(t'))\bigg]\nonumber\\
 &= \mathbb{E}_{\tau'_h \sim \pi_{comp}\lvert_{s(T), \mathcal P}}\bigg[\sum_{t' = h}^{c - \delta} \gamma^{t' - h} \sum_{p \in \mathcal P \cap \mathcal P'} r_{p}(s_{p}'(t'), a_{p}'(t'))\bigg] \label{using_dtl_val} \\
 & \hspace{15ex}+ \mathbb{E}_{\tau' \sim \pi_{comp}\lvert_{s(T), \mathcal P}}\bigg[\sum_{t' = c - \delta + 1}^{c + \eta} \gamma^{t' - h}\sum_{p \in \mathcal P} r_{p}(s_{p}'(t'), a_{p}'(t'))\bigg]\nonumber\\
 &= \mathbb{E}_{\tau'' \sim \pi_{comp}\lvert^{\mathcal P}_{s(T), \mathcal P\cap \mathcal P'}}\bigg[\sum_{t' = h}^{c + \eta} \gamma^{t' - h} \sum_{p \in \mathcal P \cap \mathcal P'} r_{p}(s_{p}'(t'), a_{p}'(t'))\bigg] \label{using_move_over_first}\\
 & \hspace{15ex}+ \mathbb{E}_{\tau' \sim \pi_{comp}\lvert_{s(T), \mathcal P}}\bigg[\sum_{t' = c - \delta + 1}^{c + \eta} \gamma^{t' - h}\sum_{(j,k) \in E(\mathcal P) \cap E^c(\mathcal P \cap \mathcal P')} \overline r_{j,k}(s_j'(t'), a_j'(t'), s_k'(t'), a_k'(t'))\bigg]\nonumber\\
 &\leq V_{h}^{\pi_{comp}\lvert^{\mathcal P}}(s(T), \mathcal P \cap \mathcal P') +\frac{\gamma^{c - \delta -h + 1}}{1 - \gamma} \tilde r_{\mathcal P'}^c\label{using_val_def}\\
 &\leq V_{h}^{\pi_{comp}}(s(T), \mathcal P \cap \mathcal P')  +\frac{\gamma^{c - \delta -h + 1}}{1 - \gamma} \tilde r_{\mathcal P'}^c,\label{using_monotone_val}
\end{align}
where in \cref{using_dtl_val} we apply the Dependence Time Lemma from \cref{dtl}. In \cref{using_move_over_first} we move over reward terms from the second term into the first and adjust the trajectory notation to account for the policy taking actions according to the starting partition $\mathcal P$. In \cref{using_val_def}, we use our definitions noticing that all edges $(j,k)\in E(\mathcal P) \cap E^c(\mathcal P\cap \mathcal P')$ are in $E^c(\mathcal P')$ as well.  Lastly, for \cref{using_monotone_val}, we use the definition of the Consistent Performance Policy.
 
We can use similar steps to derive the other direction as follows:
\begin{align*} 
 &V^{\pi_{comp}}_h (s(T), \mathcal P\cap \mathcal P') \\
 &= \mathbb{E}_{\tau' \sim \pi_{comp}\lvert_{s(T), \mathcal P\cap \mathcal P'}}\bigg[\sum_{t' = h}^{c + \eta} \gamma^{t' - h} r((s'(t'), \mathcal P(t')), a'(t'))\bigg]\\
 &= \mathbb{E}_{\tau' \sim \pi_{comp}\lvert_{s(T), \mathcal P\cap \mathcal P'}}\bigg[\sum_{t' = h}^{c - \delta} \gamma^{t' - h} \sum_{p \in \mathcal P} r_p(s_p'(t'), a_p'(t'))\bigg] \\
 & \hspace{15ex}+ \mathbb{E}_{\tau' \sim \pi_{comp}\lvert_{s(T), \mathcal P\cap \mathcal P'}}\bigg[\sum_{t' = c - \delta + 1}^{c + \eta} \gamma^{t' - h}\sum_{p \in P\mathcal \cap \mathcal P'} r_p(s_p'(t'), a_p'(t'))\bigg]\\
 &= \mathbb{E}_{\tau'' \sim \pi_{comp}\lvert^{\mathcal P\cap \mathcal P'}_{s(T), \mathcal P}}\bigg[\sum_{t' = 0}^{c + \eta} \gamma^{t' - h} \sum_{p \in \mathcal P } r_p(s_p'(t'), a_p'(t'))\bigg] \\
 & \hspace{15ex}- \mathbb{E}_{\tau' \sim \pi_{comp}\lvert_{s(T), \mathcal P\cap \mathcal P'}}\bigg[\sum_{t' = c - \delta + 1}^{c + \eta} \gamma^{t' - h}\sum_{(j,k) \in E(\mathcal P)\cap E^c(\mathcal P \cap \mathcal P')} r_{j,k}(s_j'(t'), a_j'(t'), s_k'(t'), a_k'(t'))\bigg]\\
 &\leq V_{h}^{\pi_{comp}\lvert^{\mathcal P\cap \mathcal P'}}(s(T), \mathcal P) + \frac{\gamma^{c - \delta - h + 1}}{1 - \gamma} \tilde r^c_{\mathcal P'}\\
 &\leq V_{h}^{\pi_{comp}}(s(T), \mathcal P)  + \frac{\gamma^{c - \delta - h + 1}}{1 - \gamma} \tilde r^c_{\mathcal P'}.
\end{align*}
To summarize our findings above, we have established
\begin{align*}
V_h^{\pi_{comp}}(s(T), \mathcal P) &\leq V_h^{\pi_{comp}\lvert^\mathcal P}(s(T), \mathcal P\cap \mathcal P') + \frac{\gamma^{c - \delta -h + 1}}{1 - \gamma} \tilde r_{\mathcal P'}^c \\
&\leq V_h^{\pi_{comp}}(s(T), \mathcal P\cap \mathcal P') + \frac{\gamma^{c - \delta -h + 1}}{1 - \gamma} \tilde r_{\mathcal P'}^c\\
&\leq V_h^{\pi_{comp}\lvert^{\mathcal P\cap \mathcal P'}}(s(T), \mathcal P) + \frac{2\gamma^{c - \delta -h + 1}}{1 - \gamma} \tilde r_{\mathcal P'}^c\\
&\leq V_h^{\pi_{comp}}(s(T), \mathcal P) + \frac{2\gamma^{c - \delta -h + 1}}{1 - \gamma} \tilde r_{\mathcal P'}^c.
\end{align*}
Which gives us our desired results
\begin{align*}
&\bigg\lvert V^{\pi_{comp}}_{h}(s(T), \mathcal P \cap \mathcal P') -V^{\pi_{comp}}_{h}(s(T),\mathcal P) \bigg\rvert \leq \frac{\gamma^{c - \delta -h + 1}}{1 - \gamma} \tilde r_{\mathcal P'}^c\\
&\bigg\lvert V^{\pi_{comp}\lvert^{\mathcal P}}_{h}(s(T),\mathcal P \cap\mathcal P') -V^{\pi_{comp}}_{h}(s(T),\mathcal P) \bigg\rvert \leq \frac{\gamma^{c - \delta -h + 1}}{1 - \gamma} \tilde r_{\mathcal P'}^c\\
&\bigg\lvert V^{\pi_{comp}}_{h}(s(T), \mathcal P \cap\mathcal P') -V^{\pi_{comp}\lvert^{\mathcal P\cap\mathcal P'}}_{h}(s(T),\mathcal P) \bigg\rvert \leq \frac{\gamma^{c - \delta -h + 1}}{1 - \gamma} \tilde r_{\mathcal P'}^c\\
\end{align*}
\end{proof}

\begin{corollary}
\label{make_coarser}
Consider a local movement trajectory $(s(t), a(t))$ for $t \in \{0, \ldots, \mathcal T\}$ with some $T \in \{0, \ldots, \mathcal T\}$ and $\delta \in \{0,\ldots, c\}$, such that $T- \delta \in \{0, \ldots, \mathcal T\}$. Let $\mathcal P$ be some partition finer than $Z_{comp}(s(T))$,$z \in Z_{exec}(s(T - \delta))$, 
    and  $h \in \{0,1\}$. Denoting $\mathcal P_z = \{z, \mathcal N \setminus z\}$, we have\\
    $$\bigg\lvert [V_h^{\pi_{comp}}]_z (s(T), \mathcal P\cap \mathcal P_z) - [V_h^{\pi_{comp}\lvert ^{\mathcal P}}]_z(s(T), \mathcal P\cap\mathcal P_z) \bigg\rvert \leq 2\frac{\gamma^{c - \delta -h + 1}}{1 - \gamma} \tilde r_{\mathcal P_z}.$$
\end{corollary}
\begin{proof}
Using the triangle inequality and \cref{z_intersect}, recognizing $\mathcal P_z$ is coarser than $Z_{exec}(s(T - \delta))$, we may establish
\begin{align*}
 & \bigg\lvert V_h^{\pi_{comp}} (s(T), \mathcal P\cap \mathcal P_z) - V_h^{\pi_{comp}\lvert ^{\mathcal P}}(s(T), \mathcal P\cap \mathcal P_z) \bigg\rvert \\
 & \leq \bigg\lvert V_h^{\pi_{comp}} (s(T), \mathcal P\cap \mathcal P_z) - V_h^{\pi_{comp}}(s(T), \mathcal P) \bigg\rvert+ \bigg\lvert  V_h^{\pi_{comp}}(s(T), \mathcal P)- V_h^{\pi_{comp}\lvert ^{\mathcal P}}(s(T), \mathcal P\cap \mathcal P_z) \bigg\rvert \\
 & \leq 2\frac{\gamma^{c - \delta -h + 1}}{1 - \gamma} \tilde r_{\mathcal P_z}.
\end{align*}
Using the definition of the Consistent Performance Policy, we also have:
\begin{align}
 & V_h^{\pi_{comp}} (s(T), \mathcal P\cap \mathcal P_z) - V_h^{\pi_{comp}\lvert ^{\mathcal P}}(s(T), \mathcal P\cap \mathcal P_z) \nonumber\\
&= [V_h^{\pi_{comp}}]_{z} (s(T), \mathcal P\cap \mathcal P_z) + [V_h^{\pi_{comp}}]_{\mathcal N\setminus z} (s(T), \mathcal P\cap \mathcal P_z)\nonumber \\
&\hspace{10ex} - [V_h^{\pi_{comp}\lvert ^{\mathcal P}}]_z(s(T), \mathcal P\cap \mathcal P_z) - [V_h^{\pi_{comp}\lvert ^{\mathcal P}}]_{\mathcal N\setminus z}(s(T), \mathcal P\cap \mathcal P_z)\nonumber\\
  &\geq [V_h^{\pi_{comp}}]_{z} (s(T), \mathcal P\cap \mathcal P_z) - [V_h^{\pi_{comp}\lvert^\mathcal P}]_{z} (s(T), \mathcal P\cap \mathcal P_z).\nonumber
\end{align}
Combining the results above and using the lower bound from the definition of Consistent Performance Policy, we have
$0 \leq [V_h^{\pi_{comp}}]_{z} (s(T), \mathcal P\cap \mathcal P_z) - [V_h^{\pi_{comp}\lvert^\mathcal P}]_{z} (s(T), \mathcal P\cap \mathcal P_z) \leq 2\frac{\gamma^{c - \delta -h + 1}}{1 - \gamma} \tilde r_{\mathcal P_z}$.
\end{proof}

\begin{lemma}
\label{mono_horizon}
For any state $s$, partition $\mathcal P$ finer than $Z_{comp}(s)$, and $p \in \mathcal P$, 
$$\bigg\lvert [V_0^{\pi_{comp}}]_p(s,\mathcal P) - [V_1^{\pi_{comp}}]_p(s,\mathcal P)\bigg\rvert \leq \gamma^{c + \eta} \tilde r_p$$

\end{lemma}
\begin{proof}
Using the 1-Step Contracted Improvement property of $\pi_{comp}$, we have
\begin{align*}
[V_0^{\pi_{comp}}]_p&(s,\mathcal P) - [V_1^{\pi_{comp}}]_p(s,\mathcal P) = \mathbb{E}_{\tau_0^{c + \eta}\sim \pi_{comp}\lvert_{s, \mathcal P}}\bigg[\sum_{t = 0}^{c + \eta} \gamma^t r^\mathcal C_p((s(t), \mathcal P(t)), a(t))\bigg] - [V_1^{\pi_{comp}}]_p(s,\mathcal P)\\
&=[V_{0,c + \eta - 1}^{\pi_{comp}}]_p(s, \mathcal P) - [V_1^{\pi_{comp}}]_p(s, \mathcal P) + \mathbb{E}_{\tau_0^{c + \eta}\sim \pi_{comp}\lvert_{s, \mathcal P}}[\gamma^{c + \eta} r^\mathcal C_p((s(c + \eta), \mathcal P(c + \eta)), a(c + \eta))]\\
&\leq \mathbb{E}_{\tau_0^{c + \eta}\sim \pi_{comp}\lvert_{s, \mathcal P}}[\gamma^{c + \eta} r^\mathcal C_p((s(c + \eta), \mathcal P(c + \eta)), a(c + \eta))]\\
&\leq \gamma^{c + \eta} \tilde r_p.
\end{align*}
Now, let $\pi_{c + \eta + 1}$ be the stationary Proper Cutoff Policy from the Consistent Performance Policy assumption. Letting $\pi' = (\pi_h)_{h = 0}^{c + \eta + 1}$ we can show
\begin{align*}
[&V_0^{\pi_{comp}}]_p(s,\mathcal P) - [V_1^{\pi_{comp}}]_p(s,\mathcal P) \geq [V_0^{\pi_{comp}}]_p(s,\mathcal P) - \mathbb{E}_{\tau_1^{c + \eta}\sim \pi_{comp}\lvert_{s, \mathcal P}}\bigg[\sum_{t = 1}^{c + \eta} \gamma^{t - 1} r^\mathcal C_p((s(t), \mathcal P(t)), a(t))\bigg]\\
 &= [V_0^{\pi_{comp}}]_p(s,\mathcal P) - \mathbb{E}_{\tau'^{c + \eta + 1}_1\sim \pi'\lvert_{s, \mathcal P}}\bigg[\sum_{t = 1}^{c + \eta} \gamma^{t - 1} r^\mathcal C_p((s'(t), \mathcal P'(t)), a'(t))\bigg]\\
 &= [V_0^{\pi_{comp}}]_p(s,\mathcal P) - [V_1^{\pi'}]_p(s,\mathcal P) + \mathbb{E}_{\tau'^{c + \eta + 1}_1\sim \pi'\lvert_{s, \mathcal P}}[\gamma^{c + \eta} r^\mathcal C_p((s'(c + \eta + 1), \mathcal P'(c + \eta + 1)), a'(c + \eta + 1))]\\
 &\geq \mathbb{E}_{\tau'^{c + \eta}_1\sim \pi'\lvert_{s, \mathcal P}}[\gamma^{c + \eta} r^\mathcal C_p((s'(c + \eta + 1), \mathcal P'(c + \eta + 1)), a'(c + \eta + 1))]\\
&\geq -\gamma^{c + \eta} \tilde r_p
\end{align*}
Combining the results above, we obtain $\bigg\lvert [V_0^{\pi_{comp}}]_p(s,\mathcal P) - [V_1^{\pi_{comp}}]_p(s,\mathcal P) \bigg\rvert  \leq \gamma^{c + \eta} \tilde r_p$.

\end{proof}

\begin{lemma}
    \label{tele_condition} 
For any $T \in \{0,1,\dots\}$, we have

    \begin{align*}
    &\mathbb{E}_{\tau \sim \pi_{exec}\lvert_s}\bigg[r(s(T), a(T))\bigg]
     = \mathbb{E}_{\tau \sim \pi_{exec}}\bigg[V^{\pi_{comp}}_0(s(T), Z_{exec}(s(T))) \bigg] \\
     &\hspace{10ex}- \gamma \mathbb{E}_{\tau \sim \pi_{exec}}\bigg[V^{\pi_{comp}}_0(s(T + 1), Z_{comp}(s(T + 1)\cap Z_{exec}(s(T)))  \bigg]+ \zeta(4\frac{(1 + \gamma^2)}{1 - \gamma}\gamma^{c - 1} \tilde r) + \zeta(\gamma^{c + \eta} \tilde r)
    \end{align*}
    
\end{lemma}
\begin{proof}
    Suppose we are given some $s \in \mathcal S$. Then we can create the following expression for the reward:
    \begin{align}
    \mathbb{E}_{a \sim \pi_{exec}(\cdot \lvert s)}&[r(s, a)] = \mathbb{E}_{a \sim \pi_{exec}(\cdot \lvert s)}\bigg[\sum_{z \in Z_{exec}(s)}r_z(s_z, a_z)\bigg] \nonumber\\
     &= \sum_{z \in Z_{exec}(s)}\mathbb{E}_{a_z \sim \pi_{exec}(\cdot \lvert s_z)}\bigg[r_z(s_z, a_z)\bigg] \nonumber\\
    & =\sum_{z \in Z_{exec}(s)}\mathbb{E}_{(s_p', \{p\}) \sim \rho(I(\tau^t, z))}\bigg[ \mathbb{E}_{a_p \sim [\pi_{0}]_p(\cdot \lvert (s_p', \{p\}))}\bigg[r^{\mathcal C}_z((s_p', \{p\}), a_p)\bigg]\bigg].\label{using_first_given}
    \end{align}
 Here, we abuse notation and let $\rho(I(\tau^t, z))$ denote both the random variable and its own distribution. Therefore $(s_p', \{p\}) \sim \rho(I(\tau^t, z))$ is the outcome sampled from the distribution of the random variable $\rho(I(\tau^t, z))$.
We also use the notation $[\pi_{0}]_p$ to denote the marginalized policy from the proper cutoff policy definition acting only on the agents in $p$ (a Consistent Performance Policy defined in \cref{monotone} is always proper and see from our proper cutoff policy notation introduced in \cref{cutoff}) . 
 
Now consider the expression $\mathbb{E}_{a_p \sim [\pi_{0}]_p(\cdot \lvert (s_p', \{p\}))}\bigg[r^{\mathcal C}_z((s_p', \{p\}), a_p)\bigg]$ for some given $(s_p', \{p\}) \sim \rho(I(\tau^t, z))$. We will denote $\mathcal P_z = \{z, p \setminus z\}$.
Notice that by definition, %
\begin{align*}
[&V^{\pi_{comp}\lvert^{\{p\}}}_{h}]_z(s'_p, \mathcal P_z) = \mathbb{E}_{\tau_{h} \sim \pi_{0}\lvert_{s'_p,\mathcal P_z}^{\{p\}}} \bigg[\sum_{t = h}^{c + \eta} \gamma^{t - h} r^{\mathcal C}_z((s''_p(t), \mathcal P''(t)), a''(t))\bigg]\\
&= \mathbb{E}_{a_p' \sim [\pi_{0}]_p(\cdot \lvert (s_p', \{p\}))}[r^{\mathcal C}_z((s'_p, \mathcal P_z), a'_p)] \\
&\hspace{5ex}+ \gamma\mathbb{E}_{a_p' \sim [\pi_{0}]_p(\cdot \lvert (s_p', \{p\}))}\bigg[\mathbb{E}_{(s''_p, Z_{comp}(s''_p) \cap \mathcal P_z) \sim P^{\mathcal C}(\cdot \lvert (s_p', \mathcal P_z), a_p')}\bigg[[V_1^{\pi_{comp}\lvert^{\{p\}\cap Z_{comp}(s''_p)}}]_z(s_p'', Z_{comp}(s''_p) \cap \mathcal P_z) \bigg] \bigg]
\end{align*}
 
By our assumption on $\rho(I(\tau^t, z))$ we will have $z \in Z_{exec}(s_p')$ and therefore $r^{\mathcal C}_z((s'_p, \mathcal P_z), a'_p) = r^{\mathcal C}_z((s_p', \{p\}), a_p')$. We can also trivially simplify $\{p\} \cap Z_{comp}(s''_p) = Z_{comp}(s''_p)$ since $Z_{comp}(s''_p)$ is a partition over $p$. Making these replacements and rearranging terms, we may apply our other lemmas as follows:
    \begin{align}
        &\mathbb{E}_{a_p' \sim [\pi_{0}]_p(\cdot \lvert (s_p', \{p\}))}[r^{\mathcal C}_z((s_p', \{p\}), a_p')] \\
        &= [V_0^{\pi_{comp}\lvert^{\{p\}}}]_z(s_p', \mathcal P_z) \nonumber\\
        &\hspace{5ex}- \gamma \mathbb{E}_{a_p' \sim [\pi_{0}]_p(\cdot \lvert (s_p', \{p\}))}\bigg[\mathbb{E}_{(s''_p, Z_{comp}(s''_p) \cap \mathcal P_z) \sim P^{\mathcal C}(\cdot \lvert (s_p', \mathcal P_z), a_p')}\bigg[[V_1^{\pi_{comp}\lvert^{Z_{comp}(s''_p)}}]_z(s_p'', Z_{comp}(s''_p) \cap \mathcal P_z) \bigg] \bigg] \nonumber\\
        &= [V_0^{\pi_{comp}}]_z(s_p',\mathcal P_z)  \nonumber\\
        &\hspace{5ex}- \gamma \mathbb{E}_{a_p' \sim [\pi_{0}]_p(\cdot \lvert (s_p', \{p\}))}\bigg[\mathbb{E}_{(s''_p, Z_{comp}(s''_p) \cap \mathcal P_z) \sim P^{\mathcal C}(\cdot \lvert (s_p', \mathcal P_z), a_p')}\bigg[[V_1^{\pi_{comp}}]_z(s_p'', Z_{comp}(s''_p) \cap \mathcal P_z) \bigg] \bigg]  \label{using_mono_twice}\\
        & \hspace{45ex}+ \zeta(2\frac{\gamma^{c + 1}}{1 - \gamma} \tilde r_{P_z}^c)+ \zeta(2\frac{\gamma^{c - 1}}{1 - \gamma} \tilde r_{P_z}) \nonumber\\
        &= [V_0^{\pi_{comp}}]_z(s_p',\mathcal P_z)  \nonumber\\
        &\hspace{5ex}- \gamma \mathbb{E}_{a_p' \sim [\pi_{0}]_p(\cdot \lvert (s_p', \{p\}))}\bigg[\mathbb{E}_{(s''_p, Z_{comp}(s''_p) \cap \mathcal P_z) \sim P^{\mathcal C}(\cdot \lvert (s_p', \mathcal P_z), a_p')}\bigg[[V_0^{\pi_{comp}}]_z(s_p'', Z_{comp}(s''_p) \cap \mathcal P_z) \bigg] \bigg] \label{using_mono_horizon}\\
        & \hspace{45ex}+ \zeta(2\frac{(1 + \gamma^2)}{1 - \gamma}\gamma^{c - 1} \tilde r_{P_z}^c) + \zeta(\gamma^{c + \eta} \tilde r_z)\nonumber \\
        &= [V_0^{\pi_{comp}}]_z(s_p,Z_{exec}(s_p))  \nonumber\\
        &\hspace{5ex}- \gamma \mathbb{E}_{a\sim \pi_{exec}(\cdot\lvert s) }\bigg[\mathbb{E}_{s_{next}\sim  P(\cdot\lvert s,a )}\bigg[[V_0^{\pi_{comp}}]_z([s_{next}]_p, Z_{comp}([s_{next}]_p) \cap Z_{exec}(s_p)) \bigg]  \label{using_policy_replace}\\
        & \hspace{25ex}+ \zeta(2\frac{(1 + \gamma^2)}{1 - \gamma}\gamma^{c - 1} \tilde r_{P_z}^c) + \zeta(\gamma^{c + \eta} \tilde r_z) \nonumber\\
        &= [V_0^{\pi_{comp}}]_z(s,Z_{exec}(s))  \nonumber\\
        &\hspace{5ex}- \gamma \mathbb{E}_{a\sim \pi_{exec}(\cdot\lvert s) }\bigg[\mathbb{E}_{s_{next}\sim  P(\cdot\lvert s,a )}\bigg[[V_0^{\pi_{comp}}]_z(s_{next}, Z_{comp}(s_{next}) \cap Z_{exec}(s)) \bigg]  \label{using_third prop}\\
        & \hspace{25ex}+ \zeta(2\frac{(1 + \gamma^2)}{1 - \gamma}\gamma^{c - 1} \tilde r_{P_z}^c) + \zeta(\gamma^{c + \eta} \tilde r_z) \nonumber.
    \end{align}
    In \cref{using_mono_twice}  we use \cref{make_coarser}
    and \cref{using_mono_horizon} we use \cref{mono_horizon}. In \cref{using_policy_replace} we use the first and second property of \cref{extra_properties} noticing $[s_p]_z = [s_p']_z$, and $\{z\}\cap \mathcal{P}_z = \{z\} \cap Z_{exec}(s_p)$. For the second term, by the definition of $\pi_{exec}$, both $[s''_p]_z$ and $[[s_{next}]_p]_z$ will be equivalently distributed and also when $[s''_p]_z=[[s_{next}]_p]_z$ we have $Z_{comp}([s_{next}]_p) \cap Z_{exec}(s_p) \cap \{z\} = Z_{comp}(s''_p) \cap \mathcal P_z \cap \{z\}$. In \cref{using_third prop}, we use the third property of \cref{extra_properties}.

    We may now plug this back into \cref{using_first_given} and sum over $z \in Z_{exec}(s)$.
    \begin{align*}
        &\mathbb{E}_{a \sim \pi_{exec}(\cdot \lvert s)}[r(s, a)] \\
        &\hspace{5ex}= V_0^{\pi_{comp}}(s, Z_{exec}(s)) - \gamma \mathbb{E}_{a \sim \pi_{exec}(\cdot \lvert s)}\bigg[\mathbb{E}_{s_{next}\sim P(\cdot \lvert s,a )}\bigg[V_0^{\pi_{comp}}(s_{next}, Z_{comp}(s_{next}) \cap Z_{exec}(s)) \bigg]\bigg]\\
        &\hspace{25ex}+ \zeta(4\frac{(1 + \gamma^2)}{1 - \gamma}\gamma^{c - 1} \tilde r) + \zeta(\gamma^{c + \eta} \tilde r)
    \end{align*}
Notice the subtle introduction of a factor of 2 in the first constant term. This is because by the definition of $\tilde r_{\mathcal P_z}^c$,
\begin{align*}
&\sum_{z \in Z(s)} \tilde r_{\mathcal P_z}^c \leq \sum_{z \in Z(s)} \bigg(\sum_{i \in z, j \in \mathcal N\setminus z} \max_{s_i, a_i, s_j, a_j} \lvert \overline r_{i,j}(s_i, a_i, s_j, a_j) \rvert \\
&\hspace{20ex}+ \sum_{j \in z, i \in \mathcal N\setminus z}\max_{s_i, a_i, s_j, a_j} \lvert \overline r_{i,j}(s_i, a_i, s_j, a_j) \rvert\bigg) \leq 2\tilde r.
\end{align*}

Returning to our objective of this lemma, we may show
\begin{align*}
    & \mathbb{E}_{\tau \sim \pi_{exec}}\bigg[r(s(T), a(T))\bigg]\\
    & = \mathbb{E}_{\tau^T \sim \pi_{exec}}\bigg[r(s(T), a(T))\bigg]\\
    & = \mathbb{E}_{s(T), \tau^{T - 1}}\bigg[\mathbb{E}_{\tau^T \sim \pi_{exec}}\bigg[r(s(T), a(T)) \bigg\lvert s(T), \tau^{T - 1}\bigg] \bigg]\\
    & = \mathbb{E}_{s(T), \tau^{T - 1}}\bigg[\mathbb{E}_{\tau^T \sim \pi_{exec}}\bigg[V^{\pi_{comp}}_0(s(T), Z_{exec}(s(T)))  \\
    &\hspace{8ex}- \gamma \mathbb{E}_{a \sim \pi_{exec}(\cdot \lvert s(T))}\bigg[\mathbb{E}_{s_{next}\sim P(\cdot \lvert s(T),a )}\bigg[V^{\pi_{comp}}_0(s_{next}, Z_{comp}(s_{next})\cap Z_{exec}(s(T)))\bigg] \bigg]\bigg\lvert s(T), \tau^{T - 1}\bigg] \bigg]\\
    &\hspace{35ex}+ \zeta(4\frac{(1 + \gamma^2)}{1 - \gamma}\gamma^{c - 1} \tilde r) + \zeta(\gamma^{c + \eta}\tilde r)\\
    & = \mathbb{E}_{s(T), \tau^{T - 1}}\bigg[\mathbb{E}_{\tau^T \sim \pi_{exec}}\bigg[V^{\pi_{comp}}_0(s(T), Z_{exec}(s(T)))  \\
    &\hspace{1ex}- \gamma \mathbb{E}_{a(T) \sim \pi_{exec}(\cdot \lvert s(T))}\bigg[\mathbb{E}_{s(T + 1)\sim P(\cdot \lvert s(T),a(T) )}\bigg[V^{\pi_{comp}}_0(s(T + 1), Z_{comp}(s(T + 1))\cap Z_{exec}(s(T)))\bigg] \bigg]\bigg\lvert s(T), \tau^{T - 1}\bigg] \bigg]\\
    &\hspace{35ex}+ \zeta(4\frac{(1 + \gamma^2)}{1 - \gamma}\gamma^{c - 1} \tilde r) + \zeta(\gamma^{c + \eta}\tilde r)\\
    & = \mathbb{E}_{\tau \sim \pi_{exec}}\bigg[V^{\pi_{comp}}_0(s(T), Z_{exec}(s(T))) \bigg] \\
    &\hspace{5ex}- \gamma \mathbb{E}_{\tau \sim \pi_{exec}}\bigg[V^{\pi_{comp}}_0(s(T + 1), Z_{comp}(s(T + 1)\cap Z_{exec}(s(T)))  \bigg]\\
    &\hspace{35ex}+ \zeta(4\frac{(1 + \gamma^2)}{1 - \gamma}\gamma^{c - 1} \tilde r) + \zeta(\gamma^{c + \eta} \tilde r).\\
\end{align*}

\end{proof}

\begin{theorem}
\label{same_policy} 
For any policy $\pi$ in an Locally Interdependent Multi-Agent MDP $\mathcal M$, we overload notation and define $\pi(\cdot\lvert s,\mathcal P) = \pi(\cdot\lvert s)$ for all $\mathcal P$, a potentially improper policy in the corresponding Cutoff Multi-Agent MDP $\mathcal C$. Then we have $\lvert V^\pi(s) - V^{\pi}_0(s, Z_{comp}(s)) \rvert \leq \frac{\gamma^{c'}}{1 - \gamma} \tilde r$ where $c' = \lfloor \frac{\mathcal V_{comp} - \mathcal R}{2}\rfloor$.

\end{theorem}
\begin{proof}
Using the Dependence Time Lemma (\cref{dtl}), we may establish
\begin{align}
\mathbb{E}_{\tau \sim \pi\lvert_s}\bigg[\sum_{t = 0}^{c'} \gamma^t  r(s(t),a(t)) \bigg] = \mathbb{E}_{\tau' \sim \pi\lvert_{s, Z_{comp}(s)}}\bigg[\sum_{t = 0}^{c'} \gamma^t r^{\mathcal C}((s'(t), P'(t)),a'(t)) \bigg]\label{using_up_to_c}
\end{align}
where $c' = \lfloor \frac{\mathcal V_{comp} - \mathcal R}{2}\rfloor$. Consider the trajectories $\tau \sim \pi\lvert_s$ represented by $(s(t), a(t))$ and $\tau' \sim \pi\lvert_{s,Z_{comp}(s)}$ represented by $((s'(t), \mathcal P'(t)), a'(t))$ where by the definition of the transitions in the Cutoff Multi-Agent MDP, $\mathcal P'(t) = \bigcap_{t' = 0}^t Z_{comp}(s'(t'))$. Similarly we will define $\mathcal P(t) = \bigcap_{t' = 0}^t Z_{comp}(s(t'))$.
 
Firstly, notice that the trajectories $(s(t), a(t))$ and $(s'(t), a'(t))$ are equivalently distributed for all time steps by transition independence and the definition of $\pi$.
 
Secondly, we may show an equivalence in the reward using the Dependence Time Lemma for Rewards. That is for any $\delta \in \{0,\ldots,t\}$ we have $r(s(t), a(t)) = \sum_{z\in Z_{comp}(s(t - \delta))}r_z(s_z(t), a_z(t))$ and $Z_{comp}(s(T - \delta))$ is a coarser partition than $D(s)$. 
 
Since this is true for all $\delta \in \{0, \ldots, t\}$, the intersections  of $Z_{comp}(s(T - \delta))$ must also be coarser than $D(s)$. This gives 
\begin{align*}
r(s(t), a(t)) &= \sum_{z\in \bigcap_{t' \in \{0,\ldots,t\}}Z_{comp}(s(t'))}r_z(s_z(t), a_z(t)) \\
&= r^{\mathcal C}((s(t), \mathcal P(t)), a(t))
\end{align*}
which is the reward function for the Cutoff Multi-Agent MDP. This together with the first claim gives the original equality gives \cref{using_up_to_c} as desired.
 
We may substitute this into the value function to obtain
    \begin{align}
        &V^\pi(s) = \mathbb{E}_{\tau \sim \pi\lvert_s}\bigg[\sum_{t = 0}^\infty \gamma^t r(s(t),a(t))\bigg]\nonumber\\
         &= \mathbb{E}_{\tau \sim \pi\lvert_{s,Z_{comp}(s)}}\bigg[\sum_{t = 0}^{c'} \gamma^t  r^{\mathcal C}((s(t), \mathcal P(t)),a(t))\bigg] + \mathbb{E}_{\tau \sim \pi\lvert_s}\bigg[\sum_{t = c' + 1}^{\infty} \gamma^t r(s(t),a(t))\bigg]\nonumber\\
         &= \mathbb{E}_{\tau \sim \pi\lvert_{s,Z_{comp}(s)}}\bigg[\sum_{t = 0}^{c + \eta} \gamma^t  r^{\mathcal C}((s(t), \mathcal P(t)),a(t))\bigg] + \mathbb{E}_{\tau \sim \pi\lvert_s}\bigg[ \sum_{t = c' + 1}^{c + \eta} \gamma^t \sum_{(j,k) \in E^c(Z_{comp}(s))}\overline r_{j,k}(s_j(t),a_j(t),s_k(t),a_k(t))\bigg]\label{using_move_over}\\
         &\hspace{43ex}+ \mathbb{E}_{\tau \sim \pi\lvert_s}\bigg[\sum_{t = c + \eta + 1}^{\infty} \gamma^t r(s(t),a(t))\bigg]\nonumber\\
         &= V^\pi_0(s, Z_{comp}(s)) + \zeta\bigg(\frac{\gamma^{c'}}{1 - \gamma}\tilde r\bigg).\nonumber
    \end{align}
    In \cref{using_move_over}, since both trajectories in the expectations are equivalently distributed, we may bring over rewards from the second term to complete the first trajectory. In the case $c' + 1 > c + \eta$, we assume the second term is 0.
     
    Therefore
    $\lvert V^\pi(s) - V^{\pi}_0(s, Z_{comp}(s)) \rvert \leq \frac{\gamma^{c'}}{1 - \gamma} \tilde r$.
\end{proof}

\subsection{Generalized Locally Interdependent Multi-Agent MDP}
\subsubsection{Main Proof}

\label{dependent_proofs}
Here, we prove \cref{extract_opt_bound} and \cref{extract_bound} for Generalized Locally Interdependent Multi-Agent MDPs. The proof is very similar to the Stardard Locally Interdependent Multi-Agent MDP and the lemmas provided here are one-for-one analogs to the ones used in the standard case. The primary difference is in minor proof technique variations and the constants in the lemmas and final results.
 
For the remainder of the proofs in this section, assume a Generalized Locally Interdependent Multi-Agent MDP $\mathcal{M} = (\mathcal{S}, \mathcal A, P, r, \gamma, \mathcal R)$, with some Consistent Performance Policy $\pi_{comp} = (\pi_0, \pi_1 , ..., \pi_{c + \eta})$ in its corresponding extended Cutoff Multi-Agent MDP $\mathcal{C} = (\mathcal{S}^\mathcal C, \mathcal A^\mathcal C, P^\mathcal C, r^\mathcal C, \gamma, \mathcal R, \mathcal V + \xi)$. We will also assume a corresponding extracted policy $\pi_{exec}(a \lvert s(t)) = \prod_{z\in Z_{exec}(s(t))}[\pi_{0}]_z(a_z \lvert \rho(I(\tau^t,z))$ according to a valid extraction strategy $\rho(I(\tau^t,z))$ for $z \in Z_{exec}(s)$ $s \in \mathcal S$ and $t \in \{0, 1, \dots\}$.
 
We begin by introducing the Dependence Time Lemma for Transitions. The Generalized Locally Interdependent Multi-Agent MDPs will no longer have full transition independence but independence across groups for a certain number of time steps can be established similar to the Dependence Time Lemma for Rewards.
\begin{lemma}
\label{dtl_transition}
(Dependence Time Lemma for Transitions) For any local movement trajectory $(s(t), a(t))$ for $t \in \{0, \ldots, \mathcal T\}$ in a Locally Interdependent Multi-Agent MDP and a time step $T \in \{0, \ldots, \mathcal T - \delta\}$, $D(s(T + \delta))$ is a finer partition than $Z(s(T))$ and
$$ P(s' \lvert s(T + \delta), a(T + \delta)) = \prod_{z \in Z(s(T))} P_z(s'_z \lvert s_z(T + \delta),a_z(T + \delta))$$ for all $s' \in \mathcal S$ for $\delta \in \{0, \ldots, c\}$ where $c = \lfloor \frac{\mathcal V - \mathcal R}{2}\rfloor$.
\end{lemma}
\begin{proof}
From the Dependence Time Lemma for Rewards (\cref{dtl}), we know $D(s(T + \delta))$ is a finer partition than $Z(s(T))$. Therefore, by definition $ P(s' \lvert s(T + \delta), a(T + \delta)) = \prod_{g \in D(s(T))} P_g(s'_g \lvert s_g(T + \delta),a_g(T + \delta)) = \prod_{z \in Z(s(T))} P_z(s'_z \lvert s_z(T + \delta),a_z(T + \delta))$.
\end{proof}
Similar to the Standard Locally Interdependent Multi-Agent MDP case, we will present the summaries for the proofs of \cref{extract_opt_bound} and \cref{extract_bound} and provide proofs for the lemmas used.

\begin{proof}\textbf{For \cref{extract_opt_bound} (Generalized)}
\label{general_extract_opt_bound_proof}

\begin{align}
V^*(s) -& V^{\pi^{\xi, \eta}_{exec}}(s)\leq V^*(s) - V^{\pi^{\xi, \eta}}_0(s,Z_{comp}(s)) \label{using_general_extract}\\
&\hspace{10ex} + \bigg(2\gamma^2 + \gamma^{\eta + 1} + \frac{2\gamma( 1+ \gamma) + 4n(1 + \gamma^2)}{1 - \gamma}\bigg) \frac{\gamma^{c - 1}}{1 - \gamma}\nonumber\\
&\leq V^*(s) - V^{\pi^*}_0(s,Z_{comp}(s))\label{using_general_opt_sub}\\
&\hspace{10ex} + \bigg(2\gamma^2 + \gamma^{\eta + 1} + \frac{2\gamma( 1+ \gamma) + 4n(1 + \gamma^2)}{1 - \gamma}\bigg) \frac{\gamma^{c - 1}}{1 - \gamma}\nonumber\\
&\leq  \bigg(2\gamma^{c' - c + 1} + 2\gamma^2 + \gamma^{\eta + 1} + \frac{2\gamma( 1+ \gamma) + 4n(1 + \gamma^2)}{1 - \gamma}\bigg) \frac{\gamma^{c - 1}}{1 - \gamma}\label{using_general_same}
\end{align}
In \cref{using_general_extract}, we use \cref{extract_bound} and in  \cref{using_general_same} use \cref{general_same_policy}.
    
\end{proof}
\begin{proof}\textbf{For \cref{extract_bound}   (Generalized)}
\label{general_extract_bound_proof}
Similar to the proof for the Standard Locally Interdependent Multi-Agent MDP in \cref{extract_bound_proof}, the proof will consist of three steps. 1) rewrite $V^{\pi_{exec}}$ and apply the decomposition above 2) decompose and bound the $\Delta$ difference terms 3) perform additional manipulation to match the objective.
 
\noindent\textbf{Step 1:} We may rewrite the expression for $V^{\pi_{exec}}$ to obtain,
{\small
\begin{align}
    V^{\pi_{exec}}(s) =& \mathbb{E}_{\tau \sim \pi_{exec}\lvert_s}\bigg[\sum_{t = 0}^\infty \gamma^t r(s(t), a(t))\bigg]\nonumber = \sum_{t = 0}^\infty \gamma^t \mathbb{E}_{\tau \sim \pi_{exec}\lvert_s}\bigg[r(s(t), a(t))\bigg]\nonumber\\
     = &\sum_{t = 0}^\infty \gamma^t \bigg(\mathbb{E}_{\tau \sim \pi_{exec}\lvert_s}\bigg[V^{\pi_{comp}}_0(s(t), Z_{exec}(s(t))) \bigg] \nonumber\\
     &\hspace{3.5ex}- \gamma \mathbb{E}_{\tau \sim \pi_{exec}\lvert_s}\bigg[V^{\pi_{comp}}_0(s(t + 1), Z_{comp}(s(t + 1)\cap Z_{exec}(s(t)))  \bigg]\bigg)\label{using_general_condition}\\
    &\hspace{3.5ex}+ \zeta\bigg(4n\frac{(1 + \gamma^2)}{(1 - \gamma)^2}\gamma^{c - 1} \tilde r\bigg) + \zeta\bigg(\frac{\gamma^{c + \eta}}{1 - \gamma} \tilde r\bigg)\nonumber
\end{align}
}
\\
In \cref{using_general_condition} we apply \cref{general_tele_condition}. Then using the decomposition from the Standard Locally Interdependent Multi-Agent MDP proof in \cref{extract_bound_proof}, we obtain
$$V^{\pi_{exec}(s)} = V_0^{\pi_{comp}}(s, Z_{exec}(s)) - \mathbb{E}_{\tau \sim \pi_{exec}\lvert_s}\bigg[\sum_{t = 1}^\infty \gamma^t \Delta_t^\tau\bigg]+ \zeta(4n\frac{(1 + \gamma^2)}{(1 - \gamma)^2}\gamma^{c - 1} \tilde r) + \zeta(\frac{\gamma^{c + \eta}}{1 - \gamma} \tilde r)$$
where $
\Delta_t^\tau = V_0^{\pi_{comp}}(s(t), Z_{comp}(s(t))\cap Z_{exec}(s(t - 1))) - V_0^{\pi_{comp}}(s(t), Z_{exec}(s(t)))$.
\\\\
\textbf{Step 2:} Using \cref{general_z_intersect} we obtain
\begin{align}
\lvert\Delta_t^\tau \rvert&= \lvert V_0^{\pi_{comp}}(s(t), Z_{comp}(s(t))\cap Z_{exec}(s(t - 1))) - V_0^{\pi_{comp}}(s(t), Z_{exec}(s(t)))\rvert\nonumber\\
&\leq  \lvert V_0^{\pi_{comp}}(s(t), Z_{comp}(s(t))\cap Z_{exec}(s(t - 1))) - V_0^{\pi_{comp}}(s(t), Z_{exec}(s(t))\cap Z_{exec}(s(t - 1))) \rvert \nonumber\\
&\hspace{15ex}+ \lvert V_0^{\pi_{comp}}(s(t), Z_{exec}(s(t))\cap Z_{exec}(s(t - 1)))- V_0^{\pi_{comp}}(s(t), Z_{exec}(s(t)))\rvert\nonumber\\
&\leq  \frac{2\gamma^{c + 1}}{1 - \gamma}\tilde r + \frac{2\gamma^c}{1 - \gamma}\tilde r = \frac{2(1 + \gamma)\gamma^c}{1 - \gamma}\tilde r\label{using_general_z_intersect}
\end{align}\\
where in \cref{using_general_z_intersect} we apply \cref{general_z_intersect} twice.
 
Therefore, we have shown
$$\lvert V^{\pi_{exec}}(s) - V_0^{\pi_{comp}}(s, Z_{exec}(s)) \rvert \leq \frac{2(1 + \gamma)\gamma^c}{(1 - \gamma)^2}\tilde r +4n\frac{(1 + \gamma^2)}{(1 - \gamma)^2}\gamma^{c - 1} \tilde r + \frac{\gamma^{c + \eta}}{1 - \gamma} \tilde r.$$
\\
\textbf{Step 3:} Lastly, we replace $Z_{exec}(s)$ with $Z_{comp}(s)$ and adjust the bound as follows
\begin{align*}
\lvert V^{\pi_{exec}}(s) - V_0^{\pi_{comp}}(s, Z_{comp}(s)) \rvert &\leq \lvert V^{\pi_{exec}}(s) - V_0^{\pi_{comp}}(s, Z_{exec}(s)) \rvert \\
&\hspace{10ex}+ \lvert V_0^{\pi_{comp}}(s, Z_{exec}(s)) - V_0^{\pi_{comp}}(s, Z_{comp}(s)) \rvert \\
&  \leq \frac{2\gamma^{c + 1}}{1 - \gamma}\tilde r + \frac{2(1 + \gamma)\gamma^c}{(1 - \gamma)^2}\tilde r +4n\frac{(1 + \gamma^2)}{(1 - \gamma)^2}\gamma^{c - 1} \tilde r + \frac{\gamma^{c + \eta}}{1 - \gamma} \tilde r\\
&= \bigg(2\gamma^2 + \gamma^{\eta + 1} + \frac{2\gamma( 1+ \gamma) + 4n(1 + \gamma^2)}{1 - \gamma}\bigg) \frac{\gamma^{c - 1}}{1 - \gamma}\nonumber\\
\end{align*}
using \cref{general_z_intersect}.
\end{proof}
\subsubsection{Lemmas}
\noindent We now present the lemmas used in the proof above.

\begin{lemma}
\label{general_z_intersect}
Assume a local movement trajectory $(s(t), a(t))$ for $t \in \{0, \ldots, \mathcal T\}$ and $\delta \in \{0,\ldots, c\}$, and  some $T$ such that $T, T- \delta \in \{0, \ldots, \mathcal T\}$. For any $\mathcal P$ finer than $Z_{comp}(s(T))$ and a partition $\mathcal P'$ coarser than $Z_{exec}(s(T - \delta))$, for $h \in \{0,1\}$ we have the following:
\begin{align*}
&\bigg\lvert V^{\pi_{comp}}_{h}(s(T), \mathcal P \cap \mathcal P') -V^{\pi_{comp}}_{h}(s(T),\mathcal P) \bigg\rvert \leq \frac{2\gamma^{c - \delta -h + 1}}{1 - \gamma} \tilde r\\
&\bigg\lvert V^{\pi_{comp}\lvert^{\mathcal P}}_{h}(s(T),\mathcal P \cap\mathcal P') -V^{\pi_{comp}}_{h}(s(T),\mathcal P) \bigg\rvert \leq \frac{2\gamma^{c - \delta -h + 1}}{1 - \gamma} \tilde r\\
&\bigg\lvert V^{\pi_{comp}}_{h}(s(T), \mathcal P \cap\mathcal P') -V^{\pi_{comp}\lvert^{\mathcal P\cap\mathcal P'}}_{h}(s(T),\mathcal P) \bigg\rvert \leq \frac{2\gamma^{c - \delta -h + 1}}{1 - \gamma} \tilde r\\
\end{align*}
\end{lemma}
\begin{proof}
Notice that unlike \cref{z_intersect} the lemma statement has $\tilde r$ in place of $\tilde r_{P'}^c$. This will be the difference which will introduce $n$ into the constants in the final bound.  
Similar to \cref{z_intersect}, we will evoke the Dependence Time Lemma for Rewards in \cref{dtl} but also we will evoke the Dependence Time Lemma for Transitions shown in \cref{dtl_transition} to obtain
 
\begin{align} 
 &V^{\pi_{comp}}_h (s(T), \mathcal P) \nonumber\\
 &= \mathbb{E}_{\tau'_h \sim \pi_{comp}\lvert_{s(T), \mathcal P}}\bigg[\sum_{t' = h}^{c + \eta} \gamma^{t' - h} r^{\mathcal C}((s'(t'), \mathcal P'(t')), a'(t'))\bigg]\nonumber\\
 &= \mathbb{E}_{\tau'_h \sim \pi_{comp}\lvert_{s(T), \mathcal P\cap \mathcal P'}^{\mathcal P}}\bigg[\sum_{t' = h}^{c - \delta} \gamma^{t' - h} \sum_{p \in \mathcal P \cap \mathcal P'} r_{p}(s_{p}'(t'), a_{p}'(t'))\bigg] \label{using_general_dtl_val} \\
 & \hspace{15ex}+ \mathbb{E}_{\tau' \sim \pi_{comp}\lvert_{s(T), \mathcal P}}\bigg[\sum_{t' = c - \delta + 1}^{c + \eta} \gamma^{t' - h}\sum_{p \in \mathcal P} r_{p}(s_{p}'(t'), a_{p}'(t'))\bigg]\nonumber\\
 &= \mathbb{E}_{\tau'' \sim \pi_{comp}\lvert^{\mathcal P}_{s(T), \mathcal P\cap \mathcal P'}}\bigg[\sum_{t' = h}^{c + \eta} \gamma^{t' - h} \sum_{p \in \mathcal P \cap \mathcal P'} r_{p}(s_{p}'(t'), a_{p}'(t'))\bigg] \label{using_general_move_over}\\
 & \hspace{15ex}+ \mathbb{E}_{\tau' \sim \pi_{comp}\lvert_{s(T), \mathcal P}}\bigg[\sum_{t' = c - \delta + 1}^{c + \eta} \gamma^{t'-h}\sum_{p \in \mathcal P} r_p(s_p'(t'), a_p'(t'))\bigg]\nonumber\\
 &\hspace{15ex}- \mathbb{E}_{\tau'' \sim \pi\lvert^{\mathcal P}_{s(T), \mathcal P\cap \mathcal P'}}\bigg[\sum_{t' = c - \delta + 1}^{c + \eta} \gamma^{t'-h} \sum_{p \in \mathcal P \cap \mathcal P'} r_p(s_p'(t'), a_p'(t'))\bigg] \nonumber\\
 &\leq V_{h}^{\pi_{comp}\lvert^{\mathcal P}}(s(T), \mathcal P \cap \mathcal P') +\frac{2\gamma^{c - \delta -h + 1}}{1 - \gamma} \tilde r\label{using_general_val_def}\\
 &\leq V_{h}^{\pi_{comp}}(s(T), \mathcal P \cap \mathcal P')  +\frac{2\gamma^{c - \delta -h + 1}}{1 - \gamma} \tilde r\label{using_general_monotone_val}
\end{align}
In \cref{using_general_dtl_val} we apply the Dependence Time Lemma for Rewards and Transition from \cref{dtl} and \cref{dtl_transition}. In \cref{using_general_move_over} notice the difference with the standard case in \cref{z_intersect}. Here, since we have transition dependence and extended reward dependence, we lose the interpretability of the second term and must introduce a third term to complete the value function in the first term. This then requires the use of $\tilde r$ to bound these terms. Lastly, for \cref{using_monotone_val}, we use the definition of the Consistent Performance Policy.
 
We use similar steps for the other direction as follows:
\begin{align*} 
 &V^{\pi_{comp}}_h (s(T), \mathcal P\cap \mathcal P') \\
 &= \mathbb{E}_{\tau' \sim \pi_{comp}\lvert_{s(T), \mathcal P\cap \mathcal P'}}\bigg[\sum_{t' = h}^{c + \eta} \gamma^{t' - h} r((s'(t'), \mathcal P(t')), a'(t'))\bigg]\\
 &= \mathbb{E}_{\tau' \sim \pi_{comp}\lvert_{s(T), \mathcal P}^{\mathcal P\cap \mathcal P'}}\bigg[\sum_{t' = h}^{c - \delta} \gamma^{t' - h} \sum_{p \in \mathcal P} r_p(s_p'(t'), a_p'(t'))\bigg] \\
 & \hspace{15ex}+ \mathbb{E}_{\tau' \sim \pi_{comp}\lvert_{s(T), \mathcal P\cap \mathcal P'}}\bigg[\sum_{t' = c - \delta + 1}^{c + \eta} \gamma^{t' - h}\sum_{p \in P\mathcal \cap \mathcal P'} r_p(s_p'(t'), a_p'(t'))\bigg]\\
 &= \mathbb{E}_{\tau'' \sim \pi_{comp}\lvert^{\mathcal P\cap \mathcal P'}_{s(T), \mathcal P}}\bigg[\sum_{t' = 0}^{c + \eta} \gamma^{t' - h} \sum_{p \in \mathcal P } r_p(s_p'(t'), a_p'(t'))\bigg] \\
  & \hspace{15ex}+ \mathbb{E}_{\tau' \sim \pi_{comp}\lvert_{s(T), \mathcal P}}\bigg[\sum_{t' = c - \delta + 1}^{c + \eta} \gamma^{t'}\sum_{p \in \mathcal P\cap \mathcal P'} r_p(s_p'(t'), a_p'(t'))\bigg]\\
 &\hspace{15ex}- \mathbb{E}_{\tau'' \sim \pi_{comp}\lvert_{s(T), \mathcal P}^{\mathcal P\cap \mathcal P'}}\bigg[\sum_{t' = c - \delta + 1}^{c + \eta} \gamma^{t'} \sum_{p \in \mathcal P } r_p(s_p'(t'), a_p'(t'))\bigg] \\
 &\leq V_{h}^{\pi_{comp}\lvert^{\mathcal P\cap \mathcal P'}}(s(T), \mathcal P) + \frac{2\gamma^{c - \delta - h + 1}}{1 - \gamma} \tilde r\\
 &\leq V_{h}^{\pi_{comp}}(s(T), \mathcal P)  + \frac{2\gamma^{c - \delta - h + 1}}{1 - \gamma} \tilde r.
\end{align*}
To summarize our findings above, we have established
\begin{align*}
V_h^{\pi_{comp}}(s(T), \mathcal P) &\leq V_h^{\pi_{comp}\lvert^\mathcal P}(s(T), \mathcal P\cap \mathcal P') + \frac{2\gamma^{c - \delta -h + 1}}{1 - \gamma} \tilde r \\
&\leq V_h^{\pi_{comp}}(s(T), \mathcal P\cap \mathcal P') + \frac{2\gamma^{c - \delta -h + 1}}{1 - \gamma} \tilde r\\
&\leq V_h^{\pi_{comp}\lvert^{\mathcal P\cap \mathcal P'}}(s(T), \mathcal P) + \frac{4\gamma^{c - \delta -h + 1}}{1 - \gamma} \tilde r\\
&\leq V_h^{\pi_{comp}}(s(T), \mathcal P) + \frac{4\gamma^{c - \delta -h + 1}}{1 - \gamma} \tilde r.
\end{align*}
This gives us our desired results
\begin{align*}
&\bigg\lvert V^{\pi_{comp}}_{h}(s(T), \mathcal P \cap \mathcal P') -V^{\pi_{comp}}_{h}(s(T),\mathcal P) \bigg\rvert \leq \frac{2\gamma^{c - \delta -h + 1}}{1 - \gamma} \tilde r\\
&\bigg\lvert V^{\pi_{comp}\lvert^{\mathcal P}}_{h}(s(T),\mathcal P \cap\mathcal P') -V^{\pi_{comp}}_{h}(s(T),\mathcal P) \bigg\rvert \leq \frac{2\gamma^{c - \delta -h + 1}}{1 - \gamma} \tilde r\\
&\bigg\lvert V^{\pi_{comp}}_{h}(s(T), \mathcal P \cap\mathcal P') -V^{\pi_{comp}\lvert^{\mathcal P\cap\mathcal P'}}_{h}(s(T),\mathcal P) \bigg\rvert \leq \frac{2\gamma^{c - \delta -h + 1}}{1 - \gamma} \tilde r\\
\end{align*}
\end{proof}

\begin{corollary}
\label{general_make_coarser}
Consider a local movement trajectory $(s(t), a(t))$ for $t \in \{0, \ldots, \mathcal T\}$ with some $T \in \{0, \ldots, \mathcal T\}$ and $\delta \in \{0,\ldots, c\}$, such that $T- \delta \in \{0, \ldots, \mathcal T\}$. Let $\mathcal P$ be some partition finer than $Z_{comp}(s(T))$, $z \in Z_{exec}(s(T - \delta))$, 
    and  $h \in \{0,1\}$. Denoting $\mathcal P_z = \{z, \mathcal N \setminus z\}$ we have\\
    $$\bigg\lvert [V_h^{\pi_{comp}}]_z (s(T), \mathcal P\cap \mathcal P_z) - [V_h^{\pi_{comp}\lvert ^{\mathcal P}}]_z(s(T), \mathcal P\cap\mathcal P_z) \bigg\rvert \leq 4\frac{\gamma^{c - \delta -h + 1}}{1 - \gamma} \tilde r.$$
\end{corollary}
\begin{proof}
Following the same steps as \cref{make_coarser} but instead using \cref{general_z_intersect}, we obtain
\begin{align*}
  &\bigg\lvert V_h^{\pi_{comp}} (s(T), \mathcal P\cap \mathcal P_z) - V_h^{\pi_{comp}\lvert ^{\mathcal P}}(s(T), \mathcal P\cap \mathcal P_z) \bigg\rvert \\
  &\leq \bigg\lvert V_h^{\pi_{comp}} (s(T), \mathcal P\cap \mathcal P_z) - V_h^{\pi_{comp}}(s(T), \mathcal P) \bigg\rvert+ \bigg\lvert  V_h^{\pi_{comp}}(s(T), \mathcal P)- V_h^{\pi_{comp}\lvert ^{\mathcal P}}(s(T), \mathcal P\cap \mathcal P_z) \bigg\rvert \\
  &\leq 4\frac{\gamma^{c - \delta -h + 1}}{1 - \gamma} \tilde r.
\end{align*}
Using the definition of Consistent Performance Policy, we can also establish in the other direction
\begin{align}
  &V_h^{\pi_{comp}} (s(T), \mathcal P\cap \mathcal P_z) - V_h^{\pi_{comp}\lvert ^{\mathcal P}}(s(T), \mathcal P\cap \mathcal P_z) \nonumber\\
  &= [V_h^{\pi_{comp}}]_{z} (s(T), \mathcal P\cap \mathcal P_z) + [V_h^{\pi_{comp}}]_{\mathcal N\setminus z} (s(T), \mathcal P\cap \mathcal P_z)\nonumber \\
  &\hspace{10ex} - [V_h^{\pi_{comp}\lvert ^{\mathcal P}}]_z(s(T), \mathcal P\cap \mathcal P_z)- [V_h^{\pi_{comp}\lvert ^{\mathcal P}}]_{\mathcal N\setminus z}(s(T), \mathcal P\cap \mathcal P_z)\nonumber\\
  &\geq [V_h^{\pi_{comp}}]_{z} (s(T), \mathcal P\cap \mathcal P_z) - [V_h^{\pi_{comp}\lvert^\mathcal P}]_{z} (s(T), \mathcal P\cap \mathcal P_z).\nonumber
\end{align}
\\
Combining the results above, we have
$0 \leq [V_h^{\pi_{comp}}]_{z} (s(T), \mathcal P\cap \mathcal P_z) - [V_h^{\pi_{comp}\lvert^\mathcal P}]_{z} (s(T), \mathcal P\cap \mathcal P_z) \leq 4\frac{\gamma^{c - \delta -h + 1}}{1 - \gamma} \tilde r$.
\end{proof}

\begin{lemma}
\label{general_mono_horizon}
For any state $s$, partition $\mathcal P$ finer than $Z_{comp}(s)$, and $p \in \mathcal P$, 
$$\bigg\lvert [V_0^{\pi_{comp}}]_p(s,\mathcal P) - [V_1^{\pi_{comp}}]_p(s,\mathcal P)\bigg\rvert \leq \gamma^{c + \eta} \tilde r_p$$
\end{lemma}
This proof is omitted as it is identical to \cref{mono_horizon}.

\begin{lemma}
    \label{general_tele_condition} 
For any $T \in \{0,1,\dots\}$, we have

    \begin{align*}
    &\mathbb{E}_{\tau \sim \pi_{exec}\lvert_s}\bigg[r(s(T), a(T))\bigg]
     = \mathbb{E}_{\tau \sim \pi_{exec}}\bigg[V^{\pi_{comp}}_0(s(T), Z_{exec}(s(T))) \bigg] \\
     &\hspace{10ex}- \gamma \mathbb{E}_{\tau \sim \pi_{exec}}\bigg[V^{\pi_{comp}}_0(s(T + 1), Z_{comp}(s(T + 1)\cap Z_{exec}(s(T)))  \bigg]+ \zeta(4n\frac{(1 + \gamma^2)}{1 - \gamma}\gamma^{c - 1} \tilde r) + \zeta(\gamma^{c + \eta} \tilde r)
    \end{align*}
    
\end{lemma}
\begin{proof}
    Following the same steps as \cref{tele_condition},
    \begin{align}
    \mathbb{E}_{a \sim \pi_{exec}(\cdot \lvert s)}&[r(s, a)] = \mathbb{E}_{a \sim \pi_{exec}(\cdot \lvert s)}\bigg[\sum_{z \in Z_{exec}(s)}r_z(s_z, a_z)\bigg] \nonumber\\
     &= \sum_{z \in Z_{exec}(s)}\mathbb{E}_{a_z \sim \pi_{exec}(\cdot \lvert s_z)}\bigg[r_z(s_z, a_z)\bigg] \nonumber\\
    & =\sum_{z \in Z_{exec}(s)}\mathbb{E}_{(s_p', \{p\}) \sim \rho(I(\tau^t, z))}\bigg[ \mathbb{E}_{a_p \sim [\pi_{0}]_p(\cdot \lvert (s_p', \{p\}))}\bigg[r^{\mathcal C}_z((s_p', \{p\}), a_p)\bigg]\bigg].\label{using_general_first_given}
    \end{align}
 
Unrolling the value function we obtain
\begin{align*}
[V^{\pi_{comp}}_{h}&\lvert^{\{p\}}]_z(s'_p, \mathcal P_z) = \mathbb{E}_{\tau_{h} \sim \pi_{0}\lvert_{s'_p,\mathcal P_z}^{\{p\}}} \bigg[\sum_{t = h}^{c + \eta} \gamma^{t - h} r^{\mathcal C}_z((s''_p(t), \mathcal P''(t)), a''(t))\bigg]\\
&= \mathbb{E}_{a_p' \sim [\pi_{0}]_p(\cdot \lvert (s_p', \{p\}))}[r^{\mathcal C}_z((s'_p, \mathcal P_z), a'_p)] \\
&+ \gamma\mathbb{E}_{a_p' \sim [\pi_{0}]_p(\cdot \lvert (s_p', \{p\}))}\bigg[\mathbb{E}_{(s''_p, Z_{comp}(s''_p) \cap \mathcal P_z) \sim P^{\mathcal C}(\cdot \lvert (s_p', \mathcal P_z), a_p')}\bigg[[V_1^{\pi_{comp}\lvert^{\{p\}\cap Z_{comp}(s''_p)}}]_z(s_p'', Z_{comp}(s''_p) \cap \mathcal P_z) \bigg] \bigg]
\end{align*}
 
By our assumption on $\rho(I(\tau^t, z))$ we will have $z \in Z_{exec}(s_p')$ and therefore $r^{\mathcal C}_z((s'_p, \mathcal P_z), a'_p) = r^{\mathcal C}_z((s_p', \{p\}), a_p')$. Rearranging terms, we may apply our other lemmas as follows:
    \begin{align}
        &\mathbb{E}_{a_p' \sim [\pi_{0}]_p(\cdot \lvert (s_p', \{p\}))}[r^{\mathcal C}_z((s_p', \{p\}), a_p')] \\
        &= [V_0^{\pi_{comp}\lvert^{\{p\}}}]_z(s_p', \mathcal P_z) \nonumber\\
        &\hspace{5ex}- \gamma \mathbb{E}_{a_p' \sim [\pi_{0}]_p(\cdot \lvert (s_p', \{p\}))}\bigg[\mathbb{E}_{(s''_p, Z_{comp}(s''_p) \cap \mathcal P_z) \sim P^{\mathcal C}(\cdot \lvert (s_p', \mathcal P_z), a_p')}\bigg[[V_1^{\pi_{comp}\lvert^{Z_{comp}(s''_p)}}]_z(s_p'', Z_{comp}(s''_p) \cap \mathcal P_z) \bigg] \bigg] \nonumber\\
        &= [V_0^{\pi_{comp}}]_z(s_p',\mathcal P_z)  \nonumber\\
        &\hspace{5ex}- \gamma \mathbb{E}_{a_p' \sim [\pi_{0}]_p(\cdot \lvert (s_p', \{p\}))}\bigg[\mathbb{E}_{(s''_p, Z_{comp}(s''_p) \cap \mathcal P_z) \sim P^{\mathcal C}(\cdot \lvert (s_p', \mathcal P_z), a_p')}\bigg[[V_1^{\pi_{comp}}]_z(s_p'', Z_{comp}(s''_p) \cap \mathcal P_z) \bigg] \bigg]  \label{using_general_mono_twice}\\
        & \hspace{45ex}+ \zeta(4\frac{\gamma^{c + 1}}{1 - \gamma} \tilde r)+ \zeta(4\frac{\gamma^{c - 1}}{1 - \gamma} \tilde r_z) \nonumber\\
        &= [V_0^{\pi_{comp}}]_z(s_p',\mathcal P_z)  \nonumber\\
        &\hspace{5ex}- \gamma \mathbb{E}_{a_p' \sim [\pi_{0}]_p(\cdot \lvert (s_p', \{p\}))}\bigg[\mathbb{E}_{(s''_p, Z_{comp}(s''_p) \cap \mathcal P_z) \sim P^{\mathcal C}(\cdot \lvert (s_p', \mathcal P_z), a_p')}\bigg[[V_0^{\pi_{comp}}]_z(s_p'', Z_{comp}(s''_p) \cap \mathcal P_z) \bigg] \bigg] \label{using_general_mono_horizon}\\
        & \hspace{45ex}+ \zeta(4\frac{(1 + \gamma^2)}{1 - \gamma}\gamma^{c - 1} \tilde r_{P_z}^c) + \zeta(\gamma^{c + \eta} \tilde r_z)\nonumber \\
        &= [V_0^{\pi_{comp}}]_z(s_p,Z_{exec}(s_p))  \nonumber\\
        &\hspace{5ex}- \gamma \mathbb{E}_{a\sim \pi_{exec}(\cdot\lvert s) }\bigg[\mathbb{E}_{s_{next}\sim  P(\cdot\lvert s,a )}\bigg[[V_0^{\pi_{comp}}]_z([s_{next}]_p, Z_{comp}([s_{next}]_p) \cap Z_{exec}(s_p)) \bigg]  \label{using_general_policy_replace}\\
        & \hspace{25ex}+ \zeta(4\frac{(1 + \gamma^2)}{1 - \gamma}\gamma^{c - 1} \tilde r_{P_z}^c) + \zeta(\gamma^{c + \eta} \tilde r_z) \nonumber\\
        &= [V_0^{\pi_{comp}}]_z(s,Z_{exec}(s))  \nonumber\\
        &\hspace{5ex}- \gamma \mathbb{E}_{a\sim \pi_{exec}(\cdot\lvert s) }\bigg[\mathbb{E}_{s_{next}\sim  P(\cdot\lvert s,a )}\bigg[[V_0^{\pi_{comp}}]_z(s_{next}, Z_{comp}(s_{next}) \cap Z_{exec}(s)) \bigg]  \label{using_general_third_prop}\\
        & \hspace{25ex}+ \zeta(4\frac{(1 + \gamma^2)}{1 - \gamma}\gamma^{c - 1} \tilde r_{P_z}^c) + \zeta(\gamma^{c + \eta} \tilde r_z) \nonumber.
    \end{align}
    In \cref{using_general_mono_twice}  we use \cref{general_z_intersect}
    and \cref{using_general_mono_horizon} we use \cref{general_mono_horizon}. In \cref{using_general_policy_replace} and \cref{using_general_third_prop}, we use the properties described in \cref{extra_properties} which still hold in the generalized setting, together with the definition of $\pi_{exec}$.
     
    We may now plug this back into \cref{using_general_first_given} and sum over $z \in Z_{exec}(s)$.
    \begin{align*}
        &\mathbb{E}_{a \sim \pi_{exec}(\cdot \lvert s)}[r(s, a)] \\
        &\hspace{5ex}= V_0^\pi(s, Z_{exec}(s)) - \gamma \mathbb{E}_{a \sim \pi_{exec}(\cdot \lvert s)}\bigg[\mathbb{E}_{s_{next}\sim P(\cdot \lvert s,a )}\bigg[V_0^{\pi}(s_{next}, Z_{comp}(s_{next}) \cap Z_{exec}(s)) \bigg]\bigg]\\
        &\hspace{25ex}+ \zeta(4n\frac{(1 + \gamma^2)}{1 - \gamma}\gamma^{c - 1} \tilde r) + \zeta(\gamma^{c + \eta} \tilde r)
    \end{align*}
Notice that since in \cref{using_general_policy_replace} we have $\tilde r$ rather than $\tilde r^c_{\mathcal P_z}$ like in the standard version \cref{tele_condition}, summing over $z \in Z_{exec}(s)$ introduces an extra factor of $n$. This is the main difference between the standard and generalized proofs.
Returning to our objective,\\
\begin{align*}
    & \mathbb{E}_{\tau \sim \pi_{exec}}\bigg[r(s(T), a(T))\bigg]\\
    & = \mathbb{E}_{\tau^T \sim \pi_{exec}}\bigg[r(s(T), a(T))\bigg]\\
    & = \mathbb{E}_{s(T), \tau^{T - 1}}\bigg[\mathbb{E}_{\tau^T \sim \pi_{exec}}\bigg[r(s(T), a(T)) \bigg\lvert s(T), \tau^{T - 1}\bigg] \bigg]\\
    & = \mathbb{E}_{s(T), \tau^{T - 1}}\bigg[\mathbb{E}_{\tau^T \sim \pi_{exec}}\bigg[V^{\pi_{comp}}_0(s(T), Z_{exec}(s(T)))  \\
    &\hspace{8ex}- \gamma \mathbb{E}_{a \sim \pi_{exec}(\cdot \lvert s(T))}\bigg[\mathbb{E}_{s_{next}\sim P(\cdot \lvert s(T),a )}\bigg[V^{\pi_{comp}}_0(s_{next}, Z_{comp}(s_{next})\cap Z_{exec}(s(T)))\bigg] \bigg]\bigg\lvert s(T), \tau^{T - 1}\bigg] \bigg]\\
    &\hspace{35ex}+ \zeta(4n\frac{(1 + \gamma^2)}{1 - \gamma}\gamma^{c - 1} \tilde r) + \zeta(\gamma^{c + \eta}\tilde r)\\
    & = \mathbb{E}_{s(T), \tau^{T - 1}}\bigg[\mathbb{E}_{\tau^T \sim \pi_{exec}}\bigg[V^{\pi_{comp}}_0(s(T), Z_{exec}(s(T)))  \\
    &\hspace{1ex}- \gamma \mathbb{E}_{a(T) \sim \pi_{exec}(\cdot \lvert s(T))}\bigg[\mathbb{E}_{s(T + 1)\sim P(\cdot \lvert s(T),a(T) )}\bigg[V^{\pi_{comp}}_0(s(T + 1), Z_{comp}(s(T + 1))\cap Z_{exec}(s(T)))\bigg] \bigg]\bigg\lvert s(T), \tau^{T - 1}\bigg] \bigg]\\
    &\hspace{35ex}+ \zeta(4n\frac{(1 + \gamma^2)}{1 - \gamma}\gamma^{c - 1} \tilde r) + \zeta(\gamma^{c + \eta}\tilde r)\\
    & = \mathbb{E}_{\tau \sim \pi_{exec}}\bigg[V^{\pi_{comp}}_0(s(T), Z_{exec}(s(T))) \bigg] \\
    &\hspace{5ex}- \gamma \mathbb{E}_{\tau \sim \pi_{exec}}\bigg[V^{\pi_{comp}}_0(s(T + 1), Z_{comp}(s(T + 1)\cap Z_{exec}(s(T)))  \bigg]\\
    &\hspace{35ex}+ \zeta(4n\frac{(1 + \gamma^2)}{1 - \gamma}\gamma^{c - 1} \tilde r) + \zeta(\gamma^{c + \eta} \tilde r)\\
\end{align*}

\end{proof}

\begin{theorem}
\label{general_same_policy} 
For any policy $\pi$ in an Locally Interdependent Multi-Agent MDP $\mathcal M$, we overload notation and define $\pi(\cdot\lvert s,\mathcal P) = \pi(\cdot\lvert s)$ for all $\mathcal P$, a potentially improper policy in the corresponding Cutoff Multi-Agent MDP $\mathcal C$. Then we have $\lvert V^\pi(s) - V^{\pi}_0(s, Z_{comp}(s)) \rvert \leq \frac{2\gamma^{c'}}{1 - \gamma} \tilde r$ where $c' = \lfloor \frac{\mathcal V_{comp} - \mathcal R}{2}\rfloor$.

\end{theorem}
\begin{proof}
Using the Dependence Time Lemma (\cref{dtl}), we may establish
\begin{align}
\mathbb{E}_{\tau \sim \pi\lvert_s}\bigg[\sum_{t = 0}^{c'} \gamma^t  r(s(t),a(t)) \bigg] = \mathbb{E}_{\tau' \sim \pi\lvert_{s, Z(s)}}\bigg[\sum_{t = 0}^{c'} \gamma^t r^{\mathcal C}((s'(t), P'(t)),a'(t)) \bigg]\label{using_up_to_c}
\end{align}
Consider the trajectories $\tau \sim \pi\lvert_s$ represented by $(s(t), a(t))$ and $\tau' \sim \pi\lvert_{s,Z(s)}$ represented by $((s'(t), \mathcal P'(t)), a'(t))$ where by the definition of the transitions in the Cutoff Multi-Agent MDP, $\mathcal P'(t) = \bigcap_{t' = 0}^t Z_{comp}(s'(t'))$. Similarly we will define $\mathcal P(t) = \bigcap_{t' = 0}^t Z_{comp}(s(t'))$.
 
We firstly claim that the trajectories are equivalently distributed up to time $c'$. Using the Dependence Time Lemma for Transitions, for $t \in \{0, \ldots, c'\}$ and for any $\delta \in \{0,\ldots,t\}$ we have $P(s'' \lvert s(t), a(t)) = \prod_{z\in Z_{comp}(s(t - \delta))}P(s''_z \lvert s_z(t), a_z(t))$ and $Z_{comp}(s(T - \delta))$ is a coarser partition than $D(s)$.
 
Since this is true for all $\delta \in \{0, \ldots, t\}$, the intersections  of $Z_{comp}(s(T - \delta))$ must also be coarser than $D(s)$. This gives 
\begin{align*}
P(s'' \lvert s(t), a(t)) &= \prod_{z\in \bigcap_{t' \in \{0,\ldots,t\}}Z_{comp}(s(t'))}P(s''_z \lvert s_z(t), a_z(t)) \\
&= P^{\mathcal C} ((s'', Z_{comp}(s'') \cap \mathcal P(t))\lvert (s(t), \mathcal P(t)), a(t))
\end{align*}
which is exactly the transition probability for the Cutoff Multi-Agent MDP. This together with the definition of $\pi$ gives the claim that the two trajectories are equivalently distributed up to time step $c$.
 
Secondly, we may show an equivalence in the reward repeating a similar analysis.
Using the Dependence Time Lemma for Rewards, for any $\delta \in \{0,\ldots,t\}$ we have $r(s(t), a(t)) = \sum_{z\in Z_{comp}(s(t - \delta))}r_z(s_z(t), a_z(t))$ and $Z_{comp}(s(T - \delta))$ is a coarser partition than $D(s)$. 
 
Again extending this to the intersection gives\\
\begin{align*}
r(s(t), a(t)) &= \sum_{z\in \bigcap_{t' \in \{0,\ldots,t\}}Z_{comp}(s(t'))}r_z(s_z(t), a_z(t)) \\
&= r^{\mathcal C}((s(t), \mathcal P(t)), a(t))
\end{align*}
which is the reward function for the Cutoff Multi-Agent MDP. This together with the first claim gives the original equality gives \cref{using_up_to_c} as desired.
 
Using this result, we may express the value function as 
    \begin{align}
        &V^\pi(s) = \mathbb{E}_{\tau \sim \pi\lvert_s}\bigg[\sum_{t = 0}^\infty \gamma^t r(s(t),a(t))\bigg]\nonumber\\
         &= \mathbb{E}_{\tau \sim \pi\lvert_s}\bigg[\sum_{t = 0}^{c'} \gamma^t  r(s(t),a(t)) \bigg]+ \mathbb{E}_{\tau \sim \pi\lvert_s}\bigg[\sum_{t = c' + 1}^{\infty} \gamma^t r(s(t),a(t))\bigg]\nonumber\\
         &= \mathbb{E}_{\tau' \sim \pi\lvert_{s, Z_{comp}(s)}}\bigg[\sum_{t = 0}^{c'} \gamma^t r^{\mathcal C}((s'(t), \mathcal P'(t)),a'(t)) \bigg]+ \mathbb{E}_{\tau \sim \pi\lvert_s}\bigg[\sum_{t = c' + 1}^{\infty} \gamma^t r(s(t),a(t))\bigg]\nonumber\\
         &= \mathbb{E}_{\tau' \sim \pi\lvert_{s, Z_{comp}(s)}}\bigg[\sum_{t = 0}^{c + \eta} \gamma^t r^{\mathcal C}((s'(t), \mathcal P'(t)),a'(t))  \bigg]+ \mathbb{E}_{\tau \sim \pi\lvert_s}\bigg[\sum_{t = c' + 1}^{\infty} \gamma^t r(s(t),a(t))\bigg] \label{using_extra_term}\\
         &\hspace{30ex}- \mathbb{E}_{\tau' \sim \pi\lvert_{s, Z_{comp}(s)}}\bigg[\sum_{t = c' + 1}^{c + \eta} \gamma^t r^{\mathcal C}((s'(t), \mathcal P'(t)),a'(t))  \bigg]\nonumber\\
         &\leq V^\pi_0(s, Z_{comp}(s)) + \zeta\bigg(\frac{2\gamma^{c'}}{1 - \gamma}\tilde r\bigg).\nonumber
    \end{align}
    Where in \cref{using_extra_term}, we assume that the second term is $0$ when $c' + 1 > c + \eta$.
    \\\\
    Therefore, $\lvert V^\pi(s) - V^{\pi}_0(s, Z_{comp}(s)) \rvert \leq \frac{2\gamma^{c'}}{1 - \gamma} \tilde r$.
\end{proof}

\end{document}